\documentclass[onecolumn,10pt,oneside,draftclsnofoot]{IEEEtran}%

\usepackage{amsbsy}
\usepackage{floatflt} 

\usepackage{amsmath}
\usepackage{amssymb}
\usepackage{times}
\usepackage{graphicx}
\usepackage{xspace}
\usepackage{paralist} 
\usepackage{setspace} 
\usepackage{xypic}
\xyoption{curve}
\usepackage{latexsym}
\usepackage{theorem}
\usepackage{ifthen}
\usepackage{subfigure}
\usepackage[ruled,vlined]{algorithm2e}
\usepackage{booktabs}
\usepackage{mathtools}

{\theoremheaderfont{\it} \theorembodyfont{\rmfamily}
\newtheorem{theorem}{Theorem}
\newtheorem{lemma}[theorem]{Lemma}
\newtheorem{definition}[theorem]{Definition}

\newtheorem{proposition}[theorem]{Proposition}
\newtheorem{remark}[theorem]{Remark}

}

\renewcommand{\baselinestretch}{1.0}

\def\ln{{\rm ln}}

\title{Distributed Parameter Estimation in Sensor Networks:
Nonlinear Observation Models and Imperfect Communication}

\author{Soummya Kar, Jos\'e M.~F.~Moura and Kavita Ramanan
\thanks{Manuscript received August 29, 2008; revised February 05, 2011; accepted November 15, 2011. The work of Soummya Kar and Jos\'e M.~F.~Moura was supported by NSF under grants CCF-1011903 and CCF-1018509, and by AFOSR grant FA95501010291. The work of Kavita Ramanan was supported by the NSF under
grants DMS 0405343 and CMMI 0728064.}
\thanks{Soummya Kar and Jos\'e M.~F.~Moura are with the Department of Electrical and Computer Engineering,
Carnegie Mellon University, Pittsburgh, PA, USA 15213 (e-mail:
soummyak@andrew.cmu.edu, moura@ece.cmu.edu, ph: (412) 268-6341,
fax: (412) 268-3890.)}
\thanks{Kavita Ramanan is with the Division of Applied Mathematics, Brown University, Providence, RI 02912
(e-mail: Kavita{\_}Ramanan@brown.edu.)}
\thanks{Communicated by M.~Lops, Associate Editor for Detection and Estimation.}
\thanks{Digital Object Identifier 10.1109/TIT.2012.219450}
\thanks{\copyright 2012 IEEE. Personal use of this material is permitted. Permission from IEEE must be obtained for all other uses, in any current or future media, including reprinting/republishing this material for advertising or promotional purposes, creating new collective works, for resale or redistribution to servers or lists, or reuse of any copyrighted component of this work in other works.}}

\begin{document}
\maketitle 
\maketitle

\begin{abstract}
The paper studies distributed static parameter
(vector) estimation in sensor networks with nonlinear observation models and noisy inter-sensor communication. It introduces \emph{separably estimable} observation models that generalize the observability condition in linear centralized estimation to nonlinear distributed estimation. It studies two distributed estimation algorithms in separably estimable models, the $\mathcal{NU}$ (with its linear counterpart
$\mathcal{LU}$) and the $\mathcal{NLU}$. Their update rule combines a \emph{consensus} step (where each sensor updates the state by weight averaging it with its neighbors' states) and an \emph{innovation} step (where each sensor processes its local current observation.) This makes the three algorithms of the \textit{consensus + innovations} type, very different from traditional consensus. The paper proves consistency (all sensors
reach consensus almost surely and converge to the true parameter
value,) efficiency, and asymptotic unbiasedness. For $\mathcal{LU}$ and $\mathcal{NU}$, it proves asymptotic normality and provides convergence rate guarantees. The three algorithms are characterized by appropriately
chosen decaying weight sequences.
Algorithms $\mathcal{LU}$ and $\mathcal{NU}$ are analyzed in the framework of
stochastic approximation theory; algorithm $\mathcal{NLU}$
exhibits mixed time-scale behavior and biased perturbations, and
its analysis requires a different approach that is developed in the paper.
\end{abstract}

\section{Introduction}
\label{introduction}

\subsection{Background and Motivation}
\label{backmot}
The paper studies distributed inference, in particular, distributed estimation, as  \textit{consensus+innovations} algorithms that generalize  distributed consensus by combining, at each time step, cooperation among agents (\emph{consensus}) with assimilation of their observations (\emph{innovations}). Our \textit{consensus + innovations} algorithms contrast with:
 \begin{inparaenum}[i)]
 \item standard consensus, see the extensive literature, e.g., \cite{SensNets:Olfati04,jadbabailinmorse03,Boyd,tsp07-K-M,Asilomar07-K-M,ICASSP08-K-M,Mesbahi,Tuncer,yildizscaglione07,Kashyap,Fagnani,Nedic,Huang,olfatisaberfaxmurray07}, where in each time step only local averaging of the neighbors' states occurs, and no observations are processed; and
 \item distributed estimation algorithms, see recent literature,e.g., \cite{Mesbahi-parameter,olfati-saber:2005,Giannakis-est,Khan-Moura,Schenato-Kalman}, where between measurement updates a large number of consensus steps (theoretically, an infinite number of steps) is taken.
 \end{inparaenum}
 \emph{Combined} consensus+innovations algorithms are natural when a distributed network estimates a spatially varying random field defined at~$M$ spatial locations, say, for simplicity, a temperature field. The goal is to reconstruct at each and every sensor an accurate image of the \emph{entire} spatial distribution of the $M$-dimensional field, assuming that at each time step each sensor makes a noisy measurement of the temperature at its single location. Without cooperation (no consensus step,) the processing of the successive temperature readings at each sensor (successive innovation steps) leads to a reliable estimate of the temperature at the sensor location--but provides no clue about the temperature distribution at the other $M-1$ locations. On the other hand, if sensors cooperate  (consensus iterates,) but only process the initial measurement, as in traditional consensus, they converge to the average temperature across the field, not to an estimate of the $M$-dimensional temperature distribution. The distributed consensus+innovations algorithms that we introduce achieve both; each sensor converges to an estimate of the entire $M$-dimensional field by combining \emph{consensus} and \emph{processing} of the sensors measurements. Subsequent to this paper, analysis of detection consensus+innovations type algorithms is, e.g., in~\cite{bajovicjakoveticxaviresinopolimoura-11,jakoveticmouraxavier-11}.

Important questions that arise with consensus+innovations algorithms include:
\begin{inparaenum}[i)]
\item \emph{convergence}: \label{convergence-1} do the algorithms converge and if so in what sense;
 \item \emph{consensus}: do the agents reach a consensus on their field estimates; \item \emph{distributed versus centralized}: \label{goodness-1} how good is the distributed field estimate at each sensor when compared with the \emph{centralized} estimate obtained by a fusion center, in other words are the distributed estimate sequences consistent, and asymptotically unbiased, efficient, or normal; and \item \emph{rate of convergence}: \label{rateconvergence-1}  what is the rate at which the distributed estimators converge.
\end{inparaenum}
These questions are very distinct from the convergence issues considered in the ``consensus only'' literature.

 We present three distributed consensus+innovations inference algorithms: $\mathcal{LU}$ for linear observation models  (as when each sensor makes a noisy reading of the temperature at its location, see Section~\ref{lbm-numstud};) and two algorithms, $\mathcal{NU}$ and $\mathcal{NLU}$, for nonlinear observation models (like in power grids when each sensor measures a phase differential through a sinusoidal modulation, see Section~\ref{NLU-power}.)   The paper introduces the conditions on the sensor observations model (separable estimability that we define) \emph{and} on the communication network (connectedness on average) for the distributed estimates to converge.
         The \emph{separable estimability}, akin to global observability, and \emph{connectedness}, is an intuitively pleasing condition and in a sense minimal--distributed estimation cannot do better than (the optimal) centralized estimator, hence, the model better be (globally) observable (but not necessarily locally observable;) and if the sensors need to cooperate to assimilate the data collected by the distributed network, the network should be connected (on average,) or the sensors will work in isolation.

 In contrast with other settings, e.g., \emph{linear} distributed least-mean-square (LMS) approaches to parameter estimation, e.g, \cite{Sayed-LMS,Stankovic-parameter,Giannakis-LMS,Nedic-parameter}, we study distributed estimation in the usual framework of linear or nonlinear observations of a (vector) parameter in noise\footnote{In this paper, we restrict attention to static parameter (fields). Time-varying parameters/signals are considered elsewhere, see, for example, \cite{KM-JSTSP},\cite{Riccati-weakcons} for estimation/filtering of fading (non-stationary) parameters and general time-varying linear dynamical systems, respectively.}, when the dimension, $M_n$, of the observation at each sensor~$n$ is $M_{n}\ll M$, and the parameter estimation model is locally unobservable, i.e., each individual sensor cannot recover the entire $M$-dimensional parameter from its $M_n$-dimensional observation, even if noiseless. Through cooperation (consensus) a (local) sensor estimator may converge to an estimate of the entire $M$-dimensional field, by simultaneously combining at each time~$i$, its estimate, its observation (innovation), and the estimates received from the sensors with which it communicates. We show in the paper conditions under which this holds.

We extend this distributed estimation model to include sensor and link or communication channel failures, random communication protocols, and quantized communication. These conditions make the problem more realistic when a large number of agents are involved since inexpensive sensors are bounded to fail at random times,  packet loss in wireless digital communications cause links to fail intermittently, agents can communicate asynchronously via a random protocol like gossip or one of its variants, and the  agents may be resource constrained and have a limited bit budget for communication. We make no distributional assumptions on the sensors and link failures,  they can be spatially correlated, \cite{jakoveticxaviermoura-jul10}, but are temporally uncorrelated\footnote{This dynamic network is more general and subsumes the erasure network model, where the link failures are independent over space \emph{and} time.}.
We show that, under these broad conditions, the three \emph{distributed} estimation algorithms, $\mathcal{LU}$, $\mathcal{NU}$, and $\mathcal{NLU}$, are consistent if the observation model is separable estimable  (see Section~\ref{nonlin-obs-mod}) and the network is connected on average.

\textbf{Algorithms $\mathcal{LU}$, $\mathcal{NU}$, and $\mathcal{NLU}$}: $\mathcal{LU}$ applies when the noisy observations are linear on the parameter.   For the linear model, the separably estimable condition reduces to a rank condition on the global observability Grammian. $\mathcal{LU}$ combines at each time iteration the consensus term  with the innovations associated with the new observation. Note that, in this algorithm, as well as with the other two nonlinear algorithms, the dimension of the local observation for sensor~$n$, $M_n$, is much smaller than the dimension~$M$ of the field, i.e., $M_{n}\ll M$. The algorithm $\mathcal{NU}$ generalizes $\mathcal{LU}$ to nonlinear separably estimable models. It is very important to note that, in both algorithms, $\mathcal{LU}$ and $\mathcal{NU}$, the same asymptotically decaying to zero time-varying weight sequence is associated with the consensus and innovation updates;  in other words, both the consensus and innovation terms of the algorithm exhibit the same decay rate. Because of this, it is enough to resort to stochastic  approximation techniques to prove consistency, asymptotic unbiasedness, and asymptotic normality for both algorithms, $\mathcal{LU}$ and $\mathcal{NU}$. For a treatment of general distributed stochastic algorithms see \cite{tsitsiklisphd84,tsitsiklisbertsekasathans86,Bertsekas-survey,Kushner-dist}. Beyond consistency, we characterize explicitly for the $\mathcal{LU}$ algorithm the asymptotic variance and compare it with the asymptotic variance of the centralized optimal scheme. For the $\mathcal{NU}$ algorithm, and general models, it is difficult to find explicitly a Lyapounov function (as needed by stochastic approximation). However, with  weaker assumptions on the nonlinear observation model (Lipschitz continuity and certain growth properties,) we guarantee the existence of a Lyapounov function; hence, demonstrate asymptotic normality of the $\mathcal{NU}$ estimates, see Theorems~\ref{nlth} and~\ref{thnl} in Section~\ref{nlalg1}. These conditions are much easier to verify than guessing the form of a Lyapounov function. Also, in the proof of Theorem~\ref{nlth}, we actually  show how to use these conditions to determine a Lyapounov function explicitly, which can then be used to analyze convergence rates.

The third algorithm, $\mathcal{NLU}$, applies when the observation models are only continuous and not Lipschitz continuous.  $\mathcal{NLU}$ is however a mixed time-scale algorithm, where the consensus time-scale dominates the innovations time-scale, and consists of unbiased perturbations (detailed explanation is provided in the paper.) Because of this mixed time-scales, the $\mathcal{NLU}$ algorithm does not fall under the purview of standard stochastic approximation theory, and to show its consistency requires an altogether different framework as developed in the paper, see Theorems~\ref{theorem_tilde} and~\ref{theorem_untrans}, in Section~\ref{nlalg2}.

The two algorithms $\mathcal{NLU}$ and $\mathcal{NU}$ represent different tradeoffs. We show consistency for $\mathcal{NLU}$ under weaker assumptions (observation model continuity) than for $\mathcal{NU}$ (Lipschitz continuity plus growth conditions.) On the other hand, when these more stringent conditions hold, $\mathcal{NU}$  provides convergence rate guarantees and asymptotic normality; these follow from standard stochastic approximation theory that apply to $\mathcal{NU}$ but not to $\mathcal{NLU}$.

\textbf{Brief comment on the literature.} We contrast our work with relevant recent literature on distributed estimation. Papers~\cite{Mesbahi-parameter,Giannakis-est,tsp06-K-A-M,Khan-Moura} study estimation in \emph{static} networks, where either the sensors take a single snapshot of the field at the start and then
initiate distributed consensus protocols (or more generally distributed optimization, as  in~\cite{Giannakis-est}) to fuse the initial estimates, or the observation rate of the sensors is assumed to be much slower than the inter-sensor communicate rate, thus permitting a separation of the two time-scales. On the contrary, our consensus+innovations algorithm combines fusion (consensus) and observation (innovation) updates in the same iteration. The network is \emph{dynamic} with channel failures, the protocols are \emph{random}, and the sensors fail. Unlike~\cite{Mesbahi-parameter,Giannakis-est,tsp06-K-A-M,Khan-Moura}, our approach does not require distributional assumptions on the observation noise, and we make explicit the structural assumptions on the observation model (separable estimability) and network connectivity needed to guarantee consistent parameter estimates at every sensor. These structural assumptions are quite weak and are necessary for centralized estimators to obtain consistency.

There is considerable work in \emph{linear} distributed least-mean-square (LMS) approaches to parameter estimation in \emph{static} networks, e.g., \cite{Sayed-LMS,Stankovic-parameter,Giannakis-LMS,Nedic-parameter}. While LMS is also a consensus+innovation type algorithm, we show how our linear algorithm $\mathcal{LU}$ and LMS are quite distinct, with a very different setup and goal. In LMS, for example, for channel estimation or channel equalization, \cite{alisayed-adaptive}, in adaptive filtering, \cite{VictorSolo}, or in system identification, see~\cite{Ljung-book}, the observations $z_n(i)$ are the output of a noisy finite impulse response channel (or a linear system to be identified) excited by a random input sequence $u(i)$ (these random input sequences are the regressors.) The unknown channel impulse response~$\theta$ is to be estimated by a stochastic gradient type algorithm that has available (in the channel estimation or training phase) \emph{both} the random inputs \emph{and} the regressors. In contrast, in the distributed estimation problem we study, for example, for the $\mathcal{LU}$, the observations at sensor~$n$ and time~$i$ are
\begin{equation}
\label{eqn:luobservation1}
\mathbf{z}_{n}(i)=H_{n}(i)\mathbf{\theta}^{\ast}+\mathbf{\zeta}_{n}(i)
\end{equation}
For example, in~\eqref{eqn:luobservation1}, the observation matrices $H_{n}(i)$ could be of the form,
\begin{equation}
\label{eqn:sensorfailures-1}
H_{n}(i)=\frac{1}{p}\delta_n(i)\overline{H}_n
\end{equation}
where $\delta_n(i)$ is a zero-one Bernoulli variable to account for sensor failures, $p>0$ denotes the sensing probability, and the mean value $\overline{H}_n$ models the normal operation of the sensor, e.g., measuring the local temperature, or an average of local temperatures. Equation~(\ref{eqn:luobservation1}) is the usual model in parameter estimation or waveform filtering, while~(\ref{eqn:sensorfailures-1}) extends this model in a significant way to random intermittent measurements. In our distributed estimation algorithms, we do \emph{not} know the random observation matrix~$H_{n}(i)$ (only its first and second moment), while in the LMS where they play the role of the regressors, they are usually known to the LMS algorithm.

We contrast further our \emph{linear distributed} $\mathcal{LU}$ with  linear distributed LMS. References~\cite{Sayed-LMS,Giannakis-LMS,Nedic-parameter}  use non-decaying combining weights that lead to a residual tracking error; under appropriate assumptions, these algorithms can be adapted to certain time-varying tracking scenarios; we consider time varying processes in other work, \cite{KM-JSTSP,Riccati-weakcons}. Reference~\cite{Stankovic-parameter} considers decaying weight sequences as we do in $\mathcal{LU}$, thereby establishing also $\mathcal{L}_{2}$ convergence to the desired parameter value. All these works deal with distributed \emph{linear} problems, while our work emphasizes distributed estimators for linear and nonlinear sensor observation models and establishes their convergence properties. We present the necessary (minimal) structural conditions that the distributed sensing model (given) and the inter-sensor communication network should satisfy to guarantee the existence of \emph{successful} distributed estimators. Also, apart from treating generic separably estimable nonlinear observation models, in the linear case, our algorithms $\mathcal{NU}$ and $\mathcal{LU}$ lead to asymptotic normality in addition to consistency and asymptotic unbiasedness in random time-varying networks with quantized inter-sensor communication and sensor failures.

\textbf{Remark.} We noted that the $\mathcal{NLU}$ algorithm is mixed time scale; this means a stochastic algorithm where two potentials act in the same update step with different weight or gain sequences. This should not be confused with (centralized) stochastic algorithms with coupling (see~\cite{Borkar-stochapp}), where a quickly switching parameter influences the relatively slower dynamics of another state, leading to \emph{averaged} dynamics. We note further in this context that~\cite{Gelfand-Mitter} (and references therein) develops methods to analyze mixed time scale (centralized) algorithms in the context of simulated annealing. In~\cite{Gelfand-Mitter}, the role of our innovation (new observation) potential is played by a  martingale difference term. However, in our study, the innovation is not a martingale difference process, and a key step in the analysis is to derive pathwise strong approximation results to characterize the rate at which the innovation process converges to a martingale difference process.

%
%
%
%

We briefly comment on the organization of the remaining of the paper.
The rest of this section introduces notation and preliminaries to
be adopted throughout the paper. To motivate the generic nonlinear
problem, we study the linear case (algorithm $\mathcal{LU}$) in
Section~\ref{lbm}. Section~\ref{nonlin} studies the generic
separably estimable models and the algorithm $\mathcal{NU}$,
whereas algorithm $\mathcal{NLU}$ is presented in
Section~\ref{nlalg2}. Finally, Section~\ref{conclusion} concludes
the paper.

\subsection{Notation}
\label{notgraph} For completeness, this subsection sets notation and presents preliminaries on algebraic graph theory, matrices, and dithered quantization to be used in the sequel.

\textbf{Preliminaries}: We adopt the following notation. $\mathbb{R}^{k}$: the $k$-dimensional Euclidean space; $I_{k}$: $k\times k$ identity matrix; $\mathbf{1}_{k},\mathbf{0}_{k}$: column vector of ones and zeros in $\mathbb{R}^{k}$, respectively; $P_{k}=\frac{1}{k}\mathbf{1}_{k}\mathbf{1}_{k}^{T}$: the rank one $k\times k$ projector, whose
 only non-zero eigenvalue is one, and the corresponding normalized eigenvector is
$\left(1/\sqrt{k}\right)\mathbf{1}_{k}$; $\left\|\cdot\right\|$: the standard
Euclidean 2-norm when applied to a vector and the induced 2-norm when applied to matrices, which is equivalent to the matrix spectral radius for symmetric matrices; $\theta\in \mathcal{U}\subset \mathbb{R}^{M}$: the parameter to be estimated; $\mathbf{\theta}^{\ast}$:  the true (but unknown) value of the
parameter $\theta$; $\mathbf{x}_{n}(i)\in\mathbb{R}^{M}$: the
estimate of~$\mathbf{\theta}^{\ast}$ at time~$i$ at sensor~$n$--without
loss of generality (wlog), the initial estimate, $\mathbf{x}_{n}(0)$, at time~$0$ at sensor~$n$ is a non-random quantity; $\left(\Omega,\mathcal{F}\right)$: common measurable space where all the random objects are defined; $\mathbb{P}_{\mathbf{\theta}^{\ast}} \left[\cdot\right]$ and $\mathbb{E}_{\mathbf{\theta}^{\ast}}\left[\cdot\right]$: the probability and expectation operators when the true (but unknown) parameter value is $\mathbf{\theta}^{\ast}$--when the context is clear, we abuse notation by dropping the subscript. All inequalities involving random variables are to be interpreted a.s.~(almost surely.)

\textbf{Spectral graph theory}:  For the \emph{undirected} graph $G=(V,E)$, $V=\left[1\cdots N\right]$ is the set of nodes or vertices, $|V|=N$, and~$E$ is the set of edges. The unordered pair $(n,l)\in E$ if there exists an edge between nodes~$n$ and~$l$. The graph~$G$ is simple if devoid of self-loops and multiple edges and connected if there exists a path\footnote{A path between nodes $n$ and $l$ of length $m$ is a sequence
$(n=i_{0},i_{1},\cdots,i_{m}=l)$ of vertices, such that $(i_{k},i_{k+1})\in E\:\forall~0\leq k\leq m-1$.} between each pair of nodes. The neighborhood of node~$n$ is $\Omega_{n}=\left\{l\in V\,|\,(n,l)\in E\right\}$.
The degree $d_{n}=|\Omega_{n}|$ of node~$n$ is the number of edges with~$n$ as one end point, and $D=\mbox{diag}\left(d_{1}\cdots d_{N}\right)$ is the degree matrix, the diagonal matrix with diagonal entries the degrees~$d_n$. The structure of the graph can be described by the symmetric $N\times N$ adjacency matrix, $A=\left[A_{nl}\right]$, $A_{nl}=1$, if $(n,l)\in E$, $A_{nl}=0$, otherwise.
  The graph Laplacian matrix, $L$, is $L=D-A$; it is a
 a positive semidefinite matrix whose eigenvalues can be ordered as $0=\lambda_{1}(L)\leq\lambda_{2}(L)\leq\cdots\leq\lambda_{N}(L)$.
The smallest eigenvalue $\lambda_{1}(l)$ is zero, with $\left(1/\sqrt{N}\right)\mathbf{1}_{N}$ being the corresponding normalized eigenvector. The multiplicity of the zero eigenvalue equals the number of connected components of the network; for a connected graph, $\lambda_{2}(L)>0$. This second eigenvalue is the algebraic connectivity or the Fiedler value of
the network; see \cite{FanChung,Mohar,SensNets:Bollobas98} for detailed treatment of graphs and their spectral theory.

\textbf{Kronecker product}: Since we are dealing with vector parameters, most of the matrix manipulations will involve Kronecker products. For example, the Kronecker product of the $N\times N$ matrix~$L$ and $I_{M}$ will be an $NM\times NM$ matrix, denoted by $L\otimes I_{M}$. We will deal often with matrices of the form $C=\left[I_{NM} -bL\otimes I_{M}-aI_{NM}-P_{N}\otimes I_{M}\right]$. It follows from the properties of Kronecker products and of the matrices $L$ and $P_{N}$ that the eigenvalues of the matrix~$C$ are $-a$ and $1-b\lambda_{i}(L)-a,\:2\leq i\leq N$, each being repeated~$M$ times.

We now review results from statistical quantization theory.

\textbf{Dithered quantization}: We assume that all sensors are equipped with identical uniform, dithered quantizers  $\mathbf{q}(\cdot): \mathbb{R}^{M}\rightarrow\mathcal{Q}^{M}$ applied componentwise, with countable alphabet $\mathcal{Q}^{M}=\left\{[k_{1}\Delta,\cdots,k_{M}\Delta]^{T}\left|\rule{0cm}{.35cm}\right.k_{i}\in\mathbb{Z},~\forall
i\right\}$. We assume the dither satisfies the Schuchman conditions (see~\cite{Schuchman,Lipshitz,SripadSnyder,Gray},) so that the error sequence for subtractively dithered systems (\cite{Lipshitz})
$\left\{\varepsilon(i)\right\}_{i\geq 0}$
\begin{equation}
\label{vareps}
\varepsilon(i)=q(y(i)+\nu(i))-(y(i)+\nu(i))
\end{equation}
is an i.i.d. sequence of uniformly distributed random variables on $[-\Delta/2,\Delta/2)$, which is independent of the input sequence $\left\{y(i)\right\}_{i\geq 0}$. In~(\ref{vareps}), the dither sequence $\left\{\nu(i) \right\}_{i\geq 0}$ is i.i.d. uniformly distributed random variables on $[-\Delta/2,\Delta/2)$, independent of the input sequence $\left\{y(i)\right\}_{i\geq 0}$; we refer to~\cite{karmoura-quantized} where we use this model and make further relevant comments.

\textbf{Consistency and asymptotic unbiasedness}: We recall standard definitions from
sequential estimation theory (see, for example, \cite{Lehmann-estimation}).
\begin{definition}[Consistency]: A sequence of
estimates $\left\{\mathbf{x}^{\bullet}(i)\right\}_{i\geq 0}$ is
called consistent if
\begin{equation}
\label{consis_def}
\mathbb{P}_{\mathbf{\theta}^{\ast}}
\left[\lim_{i\rightarrow\infty}\mathbf{x}^{\bullet}(i)
=\mathbf{\theta}^{\ast}\right]=1,
\:\:\forall\mathbf{\theta}^{\ast}\in\mathcal{U}
\end{equation}
or, in other words, if the estimate sequence converges a.s.~to the
true parameter value. The above definition of consistency is also
called strong consistency. When the convergence is in probability, we get weak consistency. In this paper, we use the term consistency to mean strong consistency, which implies weak consistency.
\end{definition}

\begin{definition}[Asymptotic Unbiasedness]:
\end{definition}
A sequence of estimates
$\left\{\mathbf{x}^{\bullet}(i)\right\}_{i\geq 0}$ is called
asymptotically unbiased if
\begin{equation}
\label{eq:consis_def}
\lim_{i\rightarrow\infty}\mathbb{E}_{\mathbf{\theta}^{\ast}}\left[\mathbf{x}^{\bullet}(i)\right]=\mathbf{\theta}^{\ast},~~\forall\mathbf{\theta}^{\ast}\in\mathcal{U}
\end{equation}
\section{Distributed Linear Parameter Estimation: Algorithm~$\mathcal{LU}$}
\label{lbm}
In this section, we consider the algorithm $\mathcal{LU}$ for \emph{distributed} parameter estimation when the observation model is linear. This problem motivates the generic \emph{separably estimable} nonlinear observation models considered in Sections~\ref{nonlin} and~\ref{nlalg2}. Section~\ref{lbm-probform} sets up the distributed linear estimation problem and presents the algorithm~$\mathcal{LU}$. Section~\ref{lbm-conv} establishes the consistency and asymptotic unbiasedness of the $\mathcal{LU}$ algorithm, where we show that, under the $\mathcal{LU}$ algorithm, all sensors converge a.s.~to the true parameter value, $\mathbf{\theta}^{\ast}$. Convergence rate analysis (asymptotic normality) is carried out in Section~\ref{lbm-convrate}, while Section~\ref{lbm-numstud} illustrates~$\mathcal{LU}$ with an example.
\subsection{Problem Formulation: Algorithm $\mathcal{LU}$}
\label{lbm-probform}
Let $\mathbf{\theta}^{\ast}\in\mathbb{R}^{M}$ be an $M$-dimensional parameter to be estimated by a network of~$N$ sensors. Sensor~$n$ makes the observations:
\begin{equation}
\label{obsmod}
\mathbf{z}_{n}(i)=H_{n}(i)\mathbf{\theta}^{\ast}+\mathbf{\zeta}_{n}(i)\:\in\mathbb{R}^{M_{n}}
\end{equation}
Each sensor observes only a subset of~$M_n$ of the components of~$\theta^{\ast}$, or $M_n$ linear combinations of a few components of $\theta^{\ast}$, with $M_{n}\ll M$. We make the following assumptions.
\begin{itemize}[\setlabelwidth{A.1)}]
\item{\textbf{(A.1)}}\textbf{Observation Noise}: The
observation noise process,
$\left\{\mathbf{\zeta}(i)=\left[\mathbf{\zeta}^{T}_{1}(i),
\cdots,\mathbf{\zeta}^{T}_{N}(i)\right]^{T}\right\}_{i\geq 0}$ is
 i.i.d.~zero mean, with finite second moment. In
particular, the observation noise covariance is bounded and independent of~$i$
\begin{equation}
\label{obsnoise}
\mathbb{E}\left[\mathbf{\zeta}(i)\mathbf{\zeta}^{T}(j)\right]=S_{\mathbf{\zeta}} \delta_{ij}, \:\forall i,j\geq 0
\end{equation}
where the Kronecker symbol $\delta_{ij}=1$ if $i=j$ and zero
otherwise. Note that the observation noises at different sensors
may be correlated during a particular iteration.
Eqn.~(\ref{obsnoise}) states only temporal independence. The
spatial correlation of the observation noise makes our model
applicable to practical sensor network problems, for instance, for
distributed target localization, where the observation noise is
generally correlated across sensors.
%
\item{\textbf{(A.2)}}\textbf{Sensor Failures}: The
observation matrices,
$\left\{\left[H_{1}(i),\cdots,H_{N}(i)\right]\right\}_{i\geq 0}$,
form an i.i.d. sequence with mean $\left[\overline{H}_{1},\cdots,
\overline{H}_{N}\right]$ and finite second moment. In particular,
we have
\begin{equation}
\label{Hcond}H_{n}(i) =
\overline{H}_{n}+\widetilde{H}_{n}(i),\:\forall i,n
\end{equation}
where, $\overline{H}_{n}=\mathbb{E}\left[H_{n}(i)\right],\:\forall
i,n$ and the sequence $\left\{\left[\widetilde{H}_{1}(i),\cdots,\widetilde{H}_{N}(i) \right]\right\}_{i\geq 0}$ is  zero mean i.i.d.~with finite second moment.
 Here, also, we require only temporal independence of the observation matrices, but allow them to be spatially correlated. For example, $H_n(i)=\delta_n(i) \overline{H}_{n}$, with $\delta_n(i)$ an iid sequence of Bernoulli variables modeling intermittent sensor failures.

\begin{remark}
\label{rem:LU100} The $\mathcal{LU}$ update does not use the instantaneous observation matrices, $H_{n}(i)$, only their ensemble averages. This is a distinction between $\mathcal{LU}$ and LMS. LMS (for example, in adaptive filtering) assumes  the random matrices $H_{n}(i)$ are, together with the observations, also available  (see Chapter 4 of~\cite{Chen-Guo}). In parameter estimation, the $H_{n}(i)$ model sensor failures and  $\mathcal{LU}$ has no control over their instantiations. Hence, while in LMS it may be reasonable to use the instantaneous values of the random regressors $H_{n}(i)$, in the setting we consider, the instantaneous realizations of the observation matrices are not available.
\end{remark}
\item{\textbf{(A.3)}}\textbf{Mean Connectedness, Link Failures, and Random Protocols}: The graph Laplacians
\begin{equation}
\label{Lcond}
L(i) = \overline{L}+\widetilde{L}(i),~\forall i\geq 0
\end{equation}
are a sequence of i.i.d.~matrices with mean $\overline{L} = \mathbb{E}\left[L(i)\right]$. We make no distributional assumptions on the $\{L(i)\}$. Although independent at different times, during the same iteration, the link failures can be spatially dependent, i.e., correlated. This is more general
and subsumes the erasure network model, where the link failures are independent over space \emph{and} time. Wireless sensor networks motivate this model since interference among the wireless communication channels correlates the link failures over space, while, over time, it is still reasonable to assume that the
channels are memoryless or independent. Connectedness of the graph is an important issue. The random instantiations~$G(i)$ of the graph need not be connected; in fact, all these instantiations may be disconnected. We only require  the graph to be connected on \emph{average}. This is captured by requiring that $\lambda_{2} \left(\overline{L}\right)>0$, enabling us to capture a broad class of asynchronous communication models; for example, the random asynchronous gossip protocol analyzed
in~\cite{Boyd-Gossip} satisfies $\lambda_{2}\left(\overline{L} \right)>0$ and hence falls under this framework.
\item{\textbf{(A.4)}} \textbf{Independence}: The
sequences $\left\{L(i)\right\}_{i\geq
0}$, $\left\{\mathbf{\zeta}_{n}(i) \right\}_{1\leq n\leq N,~i\geq
0}$, $\left\{H_{n}(i)\right\}_{1\leq n\leq N, i\geq
0}$, $\left\{\nu_{nl}^{m}(i)\right\}$ are mutually independent.

We introduce the distributed observability condition  required for convergence of the  $\mathcal{LU}$ linear estimation algorithm.
\begin{definition}[\hspace{-.2cm}Distributed observability]
 \label{def:distributedobservability}
 The observation system~\eqref{obsmod} is \emph{distributedly observable} if the matrix~$G$ is full rank
\begin{equation}
\label{G} G = \sum_{n=1}^{N}\overline{H}_{n}^{T}\overline{H}_{n}
\end{equation}
\label{def:distributedobservability-1}
\end{definition}
This distributed observability extends the observability condition for a centralized estimator that is needed to get a consistent estimate of the parameter $\mathbf{\theta}^{\ast}$. We note that the information available to the $n$-th sensor at any time~$i$ about the corresponding observation matrix is just the mean $\overline{H}_{n}$, and \emph{not} the random $H_{n}(i)$. Hence, the state update equation uses only the
$\overline{H}_{n}$'s, as given in~\eqref{algRE} below.
\item{\textbf{(A.5)}} \textbf{Observability}: The distributed observation system~\eqref{obsmod} is \emph{distributedly observable} in the sense of definition~\ref{def:distributedobservability}.
\end{itemize}

\textbf{Algorithm $\mathcal{LU}$}: We consider now the algorithm $\mathcal{LU}$ for distributed parameter estimation in the linear observation model~(\ref{obsmod}). Starting from some initial deterministic estimate of the parameters\footnote{The initial states may be random, we assume deterministic for notational simplicity.}, $\mathbf{x}_{n}(0)\in \mathbb{R}^{M}$, each sensor~$n$ generates  a sequence of estimates, $\left\{ \mathbf{x}_{n}(i)\right\}_{i\geq 0}$
%
%
%
%
%
%
 by the following distributed iterative algorithm:
\begin{align}
\label{algRE}
&\mathbf{x}_{n}(i+1)=\mathbf{x}_{n}(i)-\alpha(i)
\left[b\sum_{l\in\Omega_{n}(i)}\left(\mathbf{x}_{n}(i)
-\mathbf{q}\left(\mathbf{x}_{l}(i)\right.\right.
\right.
\\
\nonumber
&
\left.
\phantom{\sum_{l\in\Omega_{n}(i)}}
\left.\left.
+\mathbf{\nu}_{nl}(i)\right)\right)
-\overline{H}_{n}^{T}\left(\mathbf{z}_{n}(i)-\overline{H}_{n}\mathbf{x}_{n}(i)\right)\right]
\end{align}
where  $\left\{\mathbf{q}(\mathbf{x}_{l}(i)+ \mathbf{\nu}_{nl}(i))\right\}_{l\in\Omega_{n}(i)}$ is the dithered quantized exchanged data. In~(\ref{algRE}),  the sequence of weights $\left\{\alpha(i)\right\}$ satisfies the persistence condition \textbf{B5} given in the Appendix~\ref{res_stoch_app}; $b>0$ is a constant and  $\left\{\alpha(i)\right\}_{i\geq0}$ is a sequence of weights with properties to be defined below. The algorithm~(\ref{algRE}) is distributed because for sensor~$n$ it involves only the data from the sensors in its neighborhood~$\Omega_n(i)$. Using~\eqref{vareps}, the state update can be written as
\begin{align}
\nonumber
&\mathbf{x}_{n}(i+1)=\mathbf{x}_{n}(i)-\alpha(i)
\left[b\sum_{l\in\Omega_{n}(i)}\left(\mathbf{x}_{n}(i)-\mathbf{x}_{l}(i)\right)
\right.
\\
\label{algRE1-b}
&\left.
\hspace{-1cm}
\phantom{\sum_{l\in\Omega_{n}(i)}}
-\overline{H}_{n}^{T}\left(\mathbf{z}_{n}(i)
-\overline{H}_{n}\mathbf{x}_{n}(i)\right)
-b\mathbf{\nu}_{nl}(i)-b\mathbf{\varepsilon}_{nl}(i)\right]
\end{align}
We rewrite~(\ref{algRE1-b}) in compact form. Define the random
vectors, $\mathbf{\Upsilon}(i)$ and
$\mathbf{\Psi}(i)\in\mathbb{R}^{NM}$ with vector
components
\begin{equation}
\label{defUpsilon}
\mathbf{\Upsilon}_{n}(i)=-\sum_{l\in\Omega_{n}(i)}\mathbf{\nu}_{nl}(i)\:\:\mbox{and}\:\:
 \mathbf{\Psi}_{n}(i) =
-\sum_{l\in\Omega_{n}(i)}\mathbf{\varepsilon}_{nl}(i)
\end{equation}
It follows from the Schuchman conditions on the dither, see Section~\ref{notgraph} and \cite{karmoura-quantized}, that
\begin{align}
\label{dith1}
\mathbb{E}\left[\mathbf{\Upsilon}(i)\right]&=
\mathbb{E}\left[\mathbf{\Psi}(i)\right]=\mathbf{0}, \: \forall i
\\
\nonumber
\sup_{i}\mathbb{E}\left[\left\|\mathbf{\Upsilon}(i)\right\|^{2}\right]
&=
\sup_{i}\mathbb{E}\left[\left\|\mathbf{\Psi}(i)\right\|^{2}\right]
\\
\label{dith2}
&\leq\frac{N(N-1)M\Delta^{2}}{12}
\end{align}
from which we then have
\begin{align}
\nonumber
&
\hspace{-.225cm}
\sup_{i}\mathbb{E}\left[\left\|\mathbf{\Upsilon}(i)+\mathbf{\Psi}(i)\right\|^{2}\right]
\leq
2\sup_{i}\mathbb{E}\left[\left\|\mathbf{\Upsilon}(i)\right\|^{2}\right]
+
\\
\label{dith40}
&
\hspace{.31cm}
+2\sup_{i}\mathbb{E}\left[\left\|\mathbf{\Psi}(i)\right\|^{2}\right]
\leq
 \frac{N(N-1)M\Delta^{2}}{3}
  =  \eta_{q}
\end{align}
The iterations in~(\ref{algRE}) can be written in compact form. Stack all sensors state estimates in a long state vector estimate $\mathbf{x}(i)=\left[\mathbf{x}^{T}_{1}(i)\cdots\mathbf{x}^{T}_{N}(i)\right]^{T}$ and define the matrices
\begin{align}
\label{D-H}
\overline{D}_{\overline{H}}&=\mbox{diag}\left[\overline{H}_{1}^{T}\cdots
\overline{H}_{N}^{T}\right]
\\
\label{D-H-2}
D_{\overline{H}}&=\overline{D}_{\overline{H}}\overline{D}^T_{\overline{H}}
=\mbox{diag}\left[\overline{H}_{1}^{T}\overline{H}_{1}\cdots
\overline{H}_{N}^{T}\overline{H}_{N}\right]
\end{align}
Then, the compact vector form of the $\mathcal{LU}$ algorithm is
\begin{align}
\label{algRE1}
\mathbf{x}(i+1)&=\mathbf{x}(i)-\alpha(i)\left[b(L(i)\otimes
I_{M})\mathbf{x}(i)
\overline{D}^T_{\overline{H}}\right.
\\
\nonumber
&
\left.
-\overline{D}_{\overline{H}}
\left(\mathbf{z}(i)-\overline{D}^T_{\overline{H}}\mathbf{x}(i)\right)
+b\mathbf{\Upsilon}(i)+b\mathbf{\Psi}(i)\right]
\end{align}
In the $\mathcal{LU}$ algorithm~(\ref{algRE1}), the  covariance matrix of the noise is defined as
\begin{equation}
\label{dith5}
S_{q}=\mathbb{E}\left[\left(\mathbf{\Upsilon}(i)
+\mathbf{\Psi}(i)\right)\left(\mathbf{\Upsilon}(i)
+\mathbf{\Psi}(i)\right)^{T}\right]
\end{equation}
\textbf{Markov}: We characterize the state vector estimate $\mathbf{x}(i)$. Consider the filtration, $\left\{\mathcal{F}^{\mathbf{x}}_{i} \right\}_{i\geq 0}$, given by
\begin{equation}
\label{natF}
\mathcal{F}^{\mathbf{x}}_{i}
=\sigma\left(\mathbf{x}(0),\left\{L(j),\mathbf{z}(j),
\mathbf{\Upsilon}(j),\mathbf{\Psi}(j)\right\}_{0\leq j<i}\right)
\end{equation}
From~\textbf{(A1)}--\textbf{(A4)},  $L(i),\mathbf{z}(i)$,
$\mathbf{\Upsilon}(i)$, $\mathbf{\Psi}(i)$ are independent of
$\mathcal{F}^{\mathbf{x}}_{i}$; so,  $\left\{\mathbf{x}(i),
\mathcal{F}^{\mathbf{x}}_{i}\right\}_{i\geq 0}$ is a Markov process.
\vspace*{-.25cm}
\subsection{Consistency of $\mathcal{LU}$}
\label{lbm-conv}
We consider consistency and asymptotic unbiasedness.
%
\begin{lemma}
\label{possemdef} Consider $\mathcal{LU}$ under
Assumptions~\textbf{(A.1)-(A.5)}. Then,
$\left[b\overline{L}\otimes I_{M}+D_{\overline{H}}\right]$ is
symmetric positive definite.
\end{lemma}
\begin{proof}
Symmetricity is obvious. It also follows from the properties of
Laplacian matrices and the structure of $D_{\overline{H}}$ that
these matrices are positive semidefinite. Then the matrix
$\left[b\overline{L}\otimes I_{M}+D_{\overline{H}}\right]$ is
positive semidefinite, being the sum of two positive semidefinite
matrices. To prove positive definiteness, assume, on the
contrary, that the matrix $\left[b\overline{L}\otimes
I_{M}+D_{\overline{H}}\right]$ is not positive definite. Then,
there exists, $\mathbf{x}\in\mathbb{R}^{NM}$, such that
$\mathbf{x}\neq\mathbf{0}$ and
\begin{equation}
\label{possemdef:1} \mathbf{x}^{T}\left[b\overline{L}\otimes
I_{M}+D_{\overline{H}}\right]\mathbf{x}=\mathbf{0}
\end{equation}
From the positive semidefiniteness of $\overline{L}\otimes I_{M}$ and
$D_{\overline{H}}$, and the fact that $b>0$, it follows
\begin{equation}
\label{possemdef:2} \mathbf{x}^{T}\left[\overline{L}\otimes
I_{M}\right]\mathbf{x}=0,\:\:\mathbf{x}^{T}D_{\overline{H}}\mathbf{x}=0
\end{equation}
Partition $\mathbf{x}$ as
$\mathbf{x}=\left[\mathbf{x}_{1}^T\cdots\mathbf{x}_{N}^T\right]^{T},
\mathbf{x}_{n}\in\mathbb{R}^{M}, \forall 1\leq n\leq N$.
It follows from the properties of Laplacian matrices and the fact
that $\lambda_{2}(\overline{L})>0$, that~\eqref{possemdef:2}
holds \emph{iff}
\begin{equation}
\label{possemdef:4} \mathbf{x}_{n}=\mathbf{a},~\forall n
\end{equation}
where $\mathbf{a}\in\mathbb{R}^{M}$, and
$\mathbf{a}\neq\mathbf{0}$. Also, \eqref{possemdef:2} implies
\begin{equation}
\label{possemdef:5}
\sum_{n=1}^{N}\mathbf{x}_{n}^{T}\overline{H}_{n}^{T}\overline{H}_{n}\mathbf{x}_{n}=0
\end{equation}
Let $G$ be as in~\eqref{G}. Equations~\eqref{possemdef:5} and~\eqref{possemdef:4} imply
\begin{equation}
\label{possemdef:6} \mathbf{a}^{T}G\mathbf{a}=0,
\end{equation}
a contradiction, since $G>0$  by Assumption~\textbf{(A.5)} and $\mathbf{a}\neq\mathbf{0}$. Thus,  $\left[b\overline{L}\otimes
I_{M}+D_{\overline{H}}\right]>0$.
\end{proof}
%
\vspace*{-.3cm}
\begin{theorem}[$\mathcal{LU}$: Asymptotic unbiasedness]
\label{asymunb} Let the $\mathcal{LU}$ algorithm under~\textbf{(A.1)-(A.5)}. Then,
$\left\{\mathbf{x}_{n}(i)\right\}_{i\geq 0}$,  at sensor~$n$ is
asymptotically unbiased
\begin{equation}
\label{unb1}
\lim_{i\rightarrow\infty}\mathbb{E}
\left[\mathbf{x}_{n}(i)\right]=\mathbf{\theta}^{\ast},\:\:1\leq n \leq N
\end{equation}
\end{theorem}
\begin{proof}
Taking expectations on both sides of \eqref{algRE1} and by
the independence assumption~\textbf{(A.4)},
\begin{align}
\nonumber
\mathbb{E}\left[\mathbf{x}(i+1)\right]&=
\mathbb{E}\left[\mathbf{x}(i)\right]-\alpha(i)\left[b\left(\overline{L}\otimes
I_{M}\right)\mathbb{E}\left[\mathbf{x}(i)\right]+\right.
\\
\label{unb2}
&
\left.
+D_{\overline{H}}\mathbb{E}\left[\mathbf{x}(i)\right]
-\overline{D}_{\overline{H}}\mathbb{E}\left[\mathbf{z}(i)\right]\right]
\end{align}
Subtracting $\mathbf{1}_{N}\otimes\mathbf{\theta}^{\ast}$ from
both sides of~\eqref{unb2}, noting that
\begin{align}
\label{unb0}
\left(\overline{L}\otimes
I_{M}\right)\left(\mathbf{1}_{N}\otimes\mathbf{\theta}^{\ast}\right)&=\mathbf{0},
\\
\overline{D}_{\overline{H}}\mathbb{E}\left[\mathbf{z}(i)\right]
&=D_{\overline{H}}\left(\mathbf{1}_{N}\otimes\mathbf{\theta}^{\ast}\right)
\end{align}
we have
\begin{align}
\nonumber
\mathbb{E}\left[\mathbf{x}(i+1)\right]
&-\mathbf{1}_{N}\otimes\mathbf{\theta}^{\ast}
=\left[I_{NM}-\alpha(i)\left(b\overline{L}\otimes
I_{M}+\right.\right.
\\
\label{unb3}
&+
\left.
\left.
D_{\overline{H}}\right)\right]\left[\mathbb{E}\left[\mathbf{x}(i)\right]
- \mathbf{1}_{N}\otimes\mathbf{\theta}^{\ast}\right]
\end{align}
Let $\lambda_{\min}\left(b\overline{L}\otimes
I_{M}+D_{\overline{H}}\right)$, $\lambda_{\max}\left(b\overline{L}\otimes
I_{M}+D_{\overline{H}}\right)$ be the smallest and largest eigenvalues of
the positive definite matrix $\left[b\overline{L}\otimes
I_{M}+D_{\overline{H}}\right]$ (Lemma~\ref{possemdef}.) Since
$\alpha(i)\rightarrow 0$ (Assumption~\textbf{(B.5)}, Appendix~\ref{res_stoch_app}),
\begin{equation}
\label{unb4}
\exists i_{0} \ni: \:\: \alpha(i_{0})\leq\frac{1}{\lambda_{\max}\left(b\overline{L}\otimes
I_{M}+D_{\overline{H}}\right)},\:\:\forall i\geq i_{0}
\end{equation}
Continuing the recursion in \eqref{unb3}, we have, for
$i>i_{0}$,
\begin{align}
\nonumber
\mathbb{E}\left[\mathbf{x}(i)\right]
-\mathbf{1}_{N}\otimes\mathbf{\theta}^{\ast}
&=\left(\prod_{j=i_{0}}^{i-1}\left[I_{NM}-\alpha(j)\left(b\overline{L}\otimes
I_{M}+
\right.\right.\right.
\\
\label{unb5-1}
&
\hspace{-1.5cm}
\left.
\phantom{\prod_{j=i_{0}}^{i-1}}
+
\left.
\left.
D_{\overline{H}}\right)\right]
\right)
\left[\mathbb{E}\left[\mathbf{x}(i_{0})\right]
-\mathbf{1}_{N}\otimes\mathbf{\theta}^{\ast}\right]
\end{align}
Eqn.~\eqref{unb5-1} implies
\begin{align}
\label{unb6}
&
\left\|\mathbb{E}\left[\mathbf{x}(i)\right]
-\mathbf{1}_{N}\otimes\mathbf{\theta}^{\ast}\right\|
\leq\left(\prod_{j=i_{0}}^{i-1}\left\|I_{NM}-\alpha(j)
\right.\right.
\\
\nonumber
&
\left.
\hspace{-.8cm}
\phantom{\prod_{j=i_{0}}^{i-1}}
\left.
\left(b\overline{L}\otimes
I_{M}+D_{\overline{H}}\right)\right\|\right)\left\|
\mathbb{E}\left[\mathbf{x}(i_{0})\right]
-\mathbf{1}_{N}\otimes\mathbf{\theta}^{\ast}\right\|,~~i>i_{0}
\end{align}
It follows from \eqref{unb4}
\begin{align}
\label{unb7}
&
\hspace{-1cm}
\left\|I_{NM}-\alpha(j)\left(b\overline{L}\otimes
I_{M}+D_{\overline{H}}\right)\right\|=
\\
\nonumber
&
1-\alpha(j)\lambda_{\min}
\left(b\overline{L}\otimes I_{M}+D_{\overline{H}}\right),\:\:j\geq i_{0}
\end{align}
Eqns.~(\ref{unb6},\ref{unb7}) now give for $i>i_{0}$
\begin{align}
\label{unb8}
&\left\|\mathbb{E}\left[\mathbf{x}(i)\right]
-\mathbf{1}_{N}\otimes\mathbf{\theta}^{\ast}\right\|
\leq
\left(\prod_{j=i_{0}}^{i-1}\left(1-
\alpha(j)
\right.
\right.
\\
\nonumber
&
\hspace{+.5cm}
\left.
\phantom{\prod_{j=i_{0}}^{i-1}}
\left.
\lambda_{\min}
\left(b\overline{L}\otimes I_{M}+D_{\overline{H}}\right)
\right)\right)
\left\|\mathbb{E}\left[\mathbf{x}(i_{0})\right]
-\mathbf{1}_{N}\otimes\mathbf{\theta}^{\ast}\right\|,
\end{align}
Finally, from the inequality $1-a\leq e^{-a}$, $0\leq a\leq 1$, get
\begin{align}
\nonumber
&
\hspace{-.2cm}
\left\|\mathbb{E}\left[\mathbf{x}(i)\right]
-\mathbf{1}_{N}\otimes\mathbf{\theta}^{\ast}\right\|
\leq
e^{-\lambda_{\min}\left(b\overline{L}\otimes I_{M}+D_{\overline{H}}\right) \sum_{j=i_{0}}^{i-1}
\alpha(j)}
\\
\label{unb9}
&
\hspace{2.25cm}
\left\|\mathbb{E}\left[\mathbf{x}\left(i_{0}\right)\right]
-\mathbf{1}_{N}\otimes\mathbf{\theta}^{\ast}\right\|,
\hspace{.25cm}
i>i_{0}
\end{align}
Since, $\lambda_{\min}\left(b\overline{L}\otimes I_{M}+D_{\overline{H}}\right)>0$ and the weight sequence sums to infinity, the theorem follows since
\begin{equation}
\label{unb10}
\lim_{i\rightarrow\infty}\left\|\mathbb{E}\left[\mathbf{x}(i)\right]
-\mathbf{1}_{N}\otimes\mathbf{\theta}^{\ast}\right\|=0
\end{equation}
\end{proof}
%
%
%
%
Before proceeding to Theorems~\ref{REas} and~\ref{asynorm} establishing the consistency and asymptotic normality of the $\mathcal{LU}$, the reader may refer to Appendix~\ref{res_stoch_app}, where useful results on stochastic approximation are discussed.
\begin{theorem}[$\mathcal{LU}$: Consistency]
\label{REas} Consider  $\mathcal{LU}$  under (\textbf{A.1})--(\textbf{A.5}). Then, the estimate sequence $\left\{\mathbf{x}_{n}(i)\right\}_{i\geq 0}$ at sensor $n$ is  consistent
\begin{equation}
\label{REas:1}
\mathbb{P}\left[\lim_{i\rightarrow\infty}\mathbf{x}_{n}(i)
=\mathbf{\theta}^{\ast},~\forall n\right]=1
\end{equation}
\end{theorem}
\begin{proof}
The proof follows by showing that $\left\{\mathbf{x}(i)\right\}_{i\geq 0}$ satisfies the Assumptions~\textbf{(B.1)-(B.5)} of Theorem~\ref{RM} (Appendix~\ref{res_stoch_app}).
Recall the filtration, $\left\{\mathcal{F}_{i}^{\mathbf{x}}\right\}_{i\geq
0}$, in \eqref{natF}. Rewrite \eqref{algRE1} by adding and subtracting the
vector $\mathbf{1}_{N}\otimes\mathbf{\theta}^{\ast}$ and noting that
\begin{align}
\label{REas:2}
&
\left(\overline{L}\otimes I_{M}\right) \left(\mathbf{1}_{N}\otimes\mathbf{\theta}^{\ast}\right)=\mathbf{0}\\
\label{Reas:3}
&\mathbf{x}(i+1)=
\mathbf{x}(i)-\alpha(i)\left[b\left(\overline{L}\otimes
I_{M}\right)\left(\mathbf{x}(i)-
\right.
\hspace{-.3cm}
\phantom{{\overline{D}^T}^T_{\overline{H}}}
\right.
\\
\nonumber
&
\left.
\left.
\hspace{2cm}-
\mathbf{1}_{N}\otimes\mathbf{\theta}^{\ast}\right)
+b\left(\widetilde{L}(i)\otimes I_{M}\right) \mathbf{x}(i)
+
\right.
 \\
\nonumber
&
 \left.
 +D_{\overline{H}}
\left(\mathbf{x}(i)
\hspace{-.05cm}
-
\hspace{-.05cm}
\mathbf{1}_{N}\otimes\mathbf{\theta}^{\ast}\right)
\hspace{-.05cm}
-
\hspace{-.05cm}
\overline{D}_{\overline{H}}\left(\mathbf{z}(i)
\hspace{-.05cm}
-
\hspace{-.05cm}
\overline{D}^T_{\overline{H}}\mathbf{1}_{N}\otimes\mathbf{\theta}^{\ast}\right)
\hspace{-.1cm}
+
\right.
\\
\nonumber
&
\left.
\hspace{-.3cm}
\phantom{{\overline{D}^T}^T_{\overline{H}}}
\hspace{+3cm}
+\,\,b\mathbf{\Upsilon}(i)+b\mathbf{\Psi}(i)\right]
\end{align}
In the notation of Theorem~\ref{RM}, Appendix~\ref{res_stoch_app}, let $R\left(\mathbf{x}\right)$ and $\Gamma\left(i+1,\mathbf{x},\omega\right)$ as in~\eqref{Reas:5} and~\eqref{Reas:6} below and rewrite \eqref{Reas:3}
\begin{align}
\nonumber
&
\mathbf{x}(i+1)=\mathbf{x}(i)+\alpha(i)\left[R(\mathbf{x}(i))+\Gamma
\left(i+1,\mathbf{x}(i),\omega\right)\right]\\
\label{Reas:5}
&R\left(\mathbf{x}\right)=-\left[b\overline{L}\otimes
I_{M}+D_{\overline{H}}\right]\left(\mathbf{x}-\mathbf{1}_{N}\otimes\mathbf{\theta}^{\ast}\right)
\\
\label{Reas:6}
&\Gamma\left(i+1,\mathbf{x},\omega\right)=-\left[b\left(\widetilde{L}(i)\otimes
I_{M}\right)\mathbf{x}-
\right.
\\
\nonumber
&
\left.
\hspace{1cm}
-
\overline{D}_{\overline{H}}\left(\mathbf{z}(i)
-\overline{D}^T_{\overline{H}}\mathbf{1}_{N}\otimes\mathbf{\theta}^{\ast}\right)
+b\mathbf{\Upsilon}(i)+b\mathbf{\Psi}(i)\right]
\end{align}
Under the Assumptions~\textbf{(A.1)-(A.5)}, for fixed $i+1$, the
random family,
$\left\{\Gamma\left(i+1,\mathbf{x},\omega\right)
\right\}_{\mathbf{x}\in\mathbb{R}^{NM}}$, is $\mathcal{F}_{i+1}^{\mathbf{x}}$ measurable, zero-mean
and independent of $\mathcal{F}_{i}^{\mathbf{x}}$. Hence, the
assumptions~\textbf{(B.1)-(B.2)} of Theorem~\ref{RM} are satisfied.

We now show the existence of a stochastic potential function
$V(\cdot)$ satisfying the remaining Assumptions~\textbf{(B.3)-(B.4)}
of Theorem~\ref{RM}. To this end, define
\begin{align}
\label{Reas:7}
V\left(\mathbf{x}\right)&=\left(\mathbf{x}-\mathbf{1}_{N}\otimes
\mathbf{\theta}^{\ast}\right)^{T}\left[b\overline{L}\otimes I_{M} +D_{\overline{H}}\right]\\
&
\nonumber
\hspace{4.5cm}
\left(\mathbf{x}-\mathbf{1}_{N}\otimes\mathbf{\theta}^{\ast}\right)
\end{align}
Clearly, $V\left(\mathbf{x}\right)\in\mathbb{C}_{2}$ with bounded second
order partial derivatives. It follows from the positive
definiteness of $\left[b\overline{L}\otimes
I_{M}+D_{\overline{H}}\right]$ (Lemma~\ref{possemdef}), that
\begin{equation}
\label{Reas:8}
V\left(\mathbf{1}_{N}\otimes\mathbf{\theta}^{\ast}\right)
=0,~~V\left(\mathbf{x}\right)>0,~\mathbf{x}\neq\mathbf{1}_{N}
\otimes\mathbf{\theta}^{\ast}
\end{equation}
Since the matrix $\left[b\overline{L}\otimes
I_{M}+D_{\overline{H}}\right]$ is positive definite, the matrix
$\left[b\overline{L}\otimes I_{M}+D_{\overline{H}}\right]^{2}$ is
also positive definite and hence, there exists a constant
$c_{1}>0$, such that
\begin{align}
\nonumber
&\left(\mathbf{x}-\mathbf{1}_{N}\otimes\mathbf{\theta}^{\ast}\right)^{T}
\left[b\overline{L}\otimes
I_{M}+D_{\overline{H}}\right]^{2}
\left(\mathbf{x}-\mathbf{1}_{N}\otimes\mathbf{\theta}^{\ast}\right)\geq
\\
\label{Reas:9}
&
\hspace{1.5cm}
\geq
c_{1}\|\mathbf{x}-\mathbf{1}_{N}\otimes\mathbf{\theta}^{\ast}\|^{2}, \:\:\:\forall
\mathbf{x}\in\mathbb{R}^{NM}
\end{align}
It then follows that
\begin{align}
\label{Reas:10}
&\sup_{\left\|\mathbf{x}-\mathbf{1}_{N}\otimes\mathbf{\theta}^{\ast}\right\|
>\epsilon}\left(R\left(\mathbf{x}\right),V_{\mathbf{x}}\left(\mathbf{x}\right)\right)
=
\\
&
\nonumber
-2\inf_{\left\|\mathbf{x}-\mathbf{1}_{N}\otimes\mathbf{\theta}^{\ast}\right\|
>\epsilon}\left\{\left(\mathbf{x}-\mathbf{1}_{N}\otimes\mathbf{\theta}^{\ast}\right)^{T}
\left[b\overline{L}\otimes
I_{M}+D_{\overline{H}}\right]^{2}\right.
\\
&
\nonumber
\left.\phantom{\left[b\overline{L}\otimes I_{M}+D_{\overline{H}}\right]^{2}|}
\left(\mathbf{x}-\mathbf{1}_{N}\otimes\mathbf{\theta}^{\ast}\right)
\right\}
\nonumber
\\
&
\leq
-2\inf_{\left\|\mathbf{x}-\mathbf{1}_{N}\otimes\mathbf{\theta}^{\ast}\right\|
>\epsilon}c_{1}\left\|\mathbf{x}-\mathbf{1}_{N}\otimes
\mathbf{\theta}^{\ast}\right\|^{2}
\\
&
\nonumber
\leq -2c_{1}\epsilon^{2}
<0
\end{align}
Thus, Assumption~\textbf{(B.3)} is satisfied. From
\eqref{Reas:5}
\begin{align}
\nonumber
&
\left\|R\left(\mathbf{x}\right)\right\|^{2}
=
\\
&
\nonumber
 =
\left(\mathbf{x}-\mathbf{1}_{N}\otimes\mathbf{\theta}^{\ast}\right)^{T}
\left[b\overline{L}\otimes
I_{M}+D_{\overline{H}}\right]^{2}
\left(\mathbf{x}-\mathbf{1}_{N}\otimes\mathbf{\theta}^{\ast}\right)\nonumber
\\
\label{Reas:11}
&
 =
-\frac{1}{2}\left(R\left(\mathbf{x}\right),V_{\mathbf{x}}
\left(\mathbf{x}\right)\right)
\end{align}
From \eqref{Reas:6} and the independence Assumption~\textbf{(A.4)}
\begin{align}
\label{Reas:12}
&\mathbb{E}\left[\left\|\Gamma\left(i+1,\mathbf{x},\omega\right)\right\|^{2}\right] =
\\
\nonumber
&
=
\mathbb{E}\left[\left(\mathbf{x}-
\mathbf{1}_{N}\otimes\mathbf{\theta}^{\ast}\right)^{T}
\left(b\widetilde{L}(i)\otimes I_{M}\right)^{2}
\left(\mathbf{x}-\mathbf{1}_{N}\otimes\mathbf{\theta}^{\ast}\right)\right]
\nonumber
\\ &
\nonumber
\hspace{1.5cm}
+\mathbb{E}\left[\left\|\overline{D}_{\overline{H}}\left(\mathbf{z}(i)
-\overline{D}^T_{\overline{H}}\mathbf{1}_{N}
\otimes\mathbf{\theta}^{\ast}\right)\right\|^{2}\right]+
\\
&
\nonumber
\hspace{1.5cm}
+b^{2}\mathbb{E}\left[\left\|\mathbf{\Upsilon}(i)
+\mathbf{\Psi}(i)\right\|^{2}\right]
\end{align}
Since the random matrix $\widetilde{L}(i)$ takes values in a
finite set, there exists a constant $c_{2}>0$, such that, $\forall
\mathbf{x}\in\mathbb{R}^{NM}$,
\begin{align}
\nonumber
&
\left(\mathbf{x}-\mathbf{1}_{N}\otimes\mathbf{\theta}^{\ast}\right)^{T}
\left(b\widetilde{L}(i)\otimes I_{M}\right)^{2} \left(\mathbf{x}-\mathbf{1}_{N}\otimes\mathbf{\theta}^{\ast}\right)\leq
\\
\label{Reas:13}
&
\hspace{3.5cm}\leq
c_{2}\|\mathbf{x}-\mathbf{1}_{N}\otimes\mathbf{\theta}^{\ast}\|^{2}
\end{align}
Again, since $\left(b\overline{L}\otimes
I_{M}+D_{\overline{H}}\right)$ is positive definite, there exists a
constant $c_{3}>0$, such that, $\forall
\mathbf{x}\in\mathbb{R}^{NM}$,
\begin{align}
\nonumber
&
\left(\mathbf{x}-\mathbf{1}_{N}\otimes\mathbf{\theta}^{\ast}\right)^{T}\left[b\overline{L}\otimes
I_{M}+D_{\overline{H}}\right]\left(\mathbf{x}-\mathbf{1}_{N}\otimes\mathbf{\theta}^{\ast}\right)\geq
\\
\label{Reas:14}
&
\hspace{3.5cm}
\geq
c_{3}\|\mathbf{x}-\mathbf{1}_{N}\otimes\mathbf{\theta}^{\ast}\|^{2}.
\end{align}
We then have from \eqref{Reas:13}-\eqref{Reas:14}
\begin{align}
\nonumber
&
\mathbb{E}\left[\left(\mathbf{x}-\mathbf{1}_{N}\otimes\mathbf{\theta}^{\ast}\right)^{T}\left(b\widetilde{L}(i)\otimes
I_{M}\right)^{2}\left(\mathbf{x}-\mathbf{1}_{N}\otimes\mathbf{\theta}^{\ast}\right)\right]
 \leq
\\
\nonumber
&
\leq
\hspace{-.05cm}
\frac{c_{2}}{c_{3}}\left(\mathbf{x}-\mathbf{1}_{N}\otimes
\mathbf{\theta}^{\ast}\right)^{T}
\hspace{-.05cm}
\left[b\overline{L}\otimes
I_{M}+D_{\overline{H}}\right]
\hspace{-.05cm}
\left(\mathbf{x}-\mathbf{1}_{N}\otimes\mathbf{\theta}^{\ast}\right)\nonumber
\\
\label{Reas:15}
&
= c_{4}V\left(\mathbf{x}\right)
\end{align}
for some constant $c_{4}=\frac{c_{2}}{c_{3}}>0$. The term
\[
\mathbb{E}
\hspace{-.03cm}
\left[
\hspace{-.03cm}
\left\|\overline{D}_{\overline{H}}\mathbf{z}(i)
-D_{\overline{H}}\mathbf{1}_{N}\otimes\mathbf{\theta}^{\ast}
\right\|^{2}
\hspace{-.03cm}
\right]
\hspace{-.03cm}
+b^{2}\mathbb{E}\left[\left\|\mathbf{\Upsilon}(i)+\mathbf{\Psi}(i)\right\|^{2}\right]
\]
is bounded by a finite constant $c_{5}>0$, as it follows from
Assumptions~\textbf{(A.1)-(A.5)}. We then have from
\eqref{Reas:11}-\eqref{Reas:12}
\begin{align}
\label{Reas:16}
&\|R\left(\mathbf{x}\right)\|^{2}+\mathbb{E}\left[\left\|\Gamma\left(i+1,\mathbf{x},\omega\right)\right\|^{2}\right]
\leq
\\
&
\leq
-\frac{1}{2}\left(R\left(\mathbf{x}\right),V_{\mathbf{x}}\left(\mathbf{x}\right)\right)+
c_{4}V\left(\mathbf{x}\right)+c_{5}\leq
\nonumber
\\
\nonumber
&
\leq
c_{6}\left(1+V\left(\mathbf{x}\right)\right)-\frac{1}{2}\left(R\left(\mathbf{x}\right),V_{\mathbf{x}}\left(\mathbf{x}\right)\right)
\end{align}
where $c_{6}=\max\left(c_{4},c_{5}\right)>0$. This verifies Assumption~\textbf{(B.4)} of Theorem~\ref{RM}. Assumption~\textbf{(B.5)} is satisfied by the choice of $\left\{\alpha(i)\right\}_{i\geq 0}$. It then
follows that the process $\left\{\mathbf{x}(i)\right\}_{i\geq 0}$ converges
a.s. to $\mathbf{1}_{N}\otimes\mathbf{\theta}^{\ast}$. In other
words,
\begin{equation}
\label{Reas:17}
\mathbb{P}[\lim_{i\rightarrow\infty}\mathbf{x}_{n}(i)=\mathbf{\theta}^{\ast},~\forall
n]=1
\end{equation}
which establishes consistency of  $\mathcal{LU}$.
\end{proof}
The proof above can be modified to show $\mathcal{L}_{2}$ convergence of the sensor estimates to $\mathbf{\theta}^{\ast}$. Due to the fact that the $\mathcal{LU}$ update rule is linear, the driving noise terms are $\mathcal{L}_{2}$ bounded, and the stable (as shown in the proof) Lyapunov function $V(\cdot)$ assumes a positive definite quadratic form. Hence, by studying the recursion of the deterministic sequence $\{\mathbb{E}[V(\mathbf{x}(i))]\}$ and by similar arguments\footnote{Note, that Lemma 4 in~\cite{karmoura-quantized} does not assume the additional term due to new observations at each iteration. However, this does not pose difficulties as the observation weights are the same as the consensus weights.}  as in~\cite{karmoura-quantized} (Lemma 4), we conclude the following:
\begin{lemma}[Mean square convergence]
\label{L_2conv}
Let the hypotheses of Theorem~\ref{REas} hold and, in addition, the weight sequence $\{\alpha(i)\}$ satisfy the following:
\begin{equation}
\label{L_2conv1}
\alpha(i)=\frac{a}{(i+1)^{\tau}}
\end{equation}
where $a>0$ and $.5<\tau\leq 1$. Then, the a.s. convergence in Theorem~\ref{REas} holds in $\mathcal{L}_{2}$ also, i.e., for all $n$,
\begin{equation}
\label{L_2conv2}
\lim_{i\rightarrow\infty}\mathbb{E}\left[\left\|\mathbf{x}_{n}(i)-\mathbf{\theta}^{\ast}\right\|^{2}\right]=0
\end{equation}
\end{lemma}

\subsection{Asymptotic Variance: $\mathcal{LU}$}
\label{lbm-convrate} In this subsection, we carry out a
convergence rate analysis of the $\mathcal{LU}$ algorithm by
studying its moderate deviation characteristics. We summarize here
some definitions and terminology from the statistical literature,
used to characterize the performance of sequential estimation
procedures (see~\cite{Lehmann-estimation}).

\begin{definition}[Asymptotic Normality] A sequence of
estimates $\left\{\mathbf{x}^{\bullet}(i)\right\}_{i\geq 0}$ is asymptotically normal if for every $\mathbf{\theta}^{\ast}\in\mathcal{U}$, there exists a positive semidefinite matrix $S(\mathbf{\theta}^{\ast})\in\mathbb{R}^{M\times M}$, such that,
\begin{equation}
\label{def_normality}
\lim_{i\rightarrow\infty}\sqrt{i}\left(\mathbf{x}^{\bullet}(i)- \mathbf{\theta}^{\ast}\right)\Longrightarrow\mathcal{N} \left(\mathbf{0}_{M},S(\mathbf{\theta}^{\ast})\right)
\end{equation}
The matrix $S(\mathbf{\theta}^{\ast})$ is called the asymptotic
variance of the estimate sequence
$\left\{\mathbf{x}^{\bullet}(i)\right\}_{i\geq 0}$.
\end{definition}

In the following we prove the asymptotic normality of the
$\mathcal{LU}$ algorithm and explicitly characterize the resulting
asymptotic variance. To this end, define
\begin{align}
\nonumber
S_{H}&=\mathbb{E}\left[\left(\overline{D}_{\overline{H}}\left[
\begin{array}{lll}
                    \widetilde{H}_{1}(i) &  &  \\
                    \ddots & \ddots & \ddots \\
                    & & \widetilde{H}_{N}(i)
                   \end{array}
          \right]\mathbf{1}_{N}\mathbf{\theta}^{\ast}\right)
          \right.
\\
\label{asnor1}
&
\left.
          \left(\overline{D}_{\overline{H}}\left[
\begin{array}{lll}
                    \widetilde{H}_{1}(i) &  &  \\
                    \ddots & \ddots & \ddots \\
                    & & \widetilde{H}_{N}(i)
                   \end{array}
          \right]\mathbf{1}_{N}\mathbf{\theta}^{\ast}\right)^{T}\right]
\end{align}
Let $\lambda_{\min}\left(b\overline{L}\otimes I_{M}+D_{\overline{H}}\right)$ be the smallest eigenvalue of $\left[b\overline{L}\otimes
I_{M}+D_{\overline{H}}\right]$ and recall $S_{\mathbf{\zeta}},S_{q}$ in \eqref{obsnoise} and \eqref{dith5}.

We now state the main result of this subsection, establishing the
asymptotic normality of the $\mathcal{LU}$ algorithm.

\begin{theorem}[$\mathcal{LU}$: Asymptotic efficiency/ normality]
\label{asynorm} Let the $\mathcal{LU}$ algorithm under~\textbf{(A.1)-(A.5)} with link weight sequence, $\left\{\alpha(i)\right\}_{i\geq 0}$ that is given by:
\begin{equation}
\label{asynorm1} \alpha(i)=\frac{a}{i+1},\: \forall i
\end{equation}
for some constant $a>0$. Let $\left\{\mathbf{x}(i)\right\}_{i\geq 0}$ be the
state sequence generated. Then, if $a>\frac{1}{2\lambda_{\min}\left(b\overline{L} \otimes I_{M}+D_{\overline{H}}\right)}$, we have
\begin{equation}
\label{asynorm2}
\sqrt(i)\left(\mathbf{x}(i)-\mathbf{1}_{N}\otimes \mathbf{\theta}^{\ast}\right)\Longrightarrow \:\mathcal{N}(\mathbf{0},S(\mathbf{\theta}^{\ast}))
\end{equation}
where
\begin{align}
\label{asynorm3}
S(\mathbf{\theta}^{\ast})&=a^{2}\int_{0}^{\infty}e^{\Sigma
v}S_{0}e^{\Sigma v}dv,
\\
\label{asynorm3-b}
%
 \Sigma&=-a\left[b\overline{L}\otimes
I_{M}+D_{\overline{H}}\right]+\frac{1}{2}I,
\end{align}
and
%
\begin{align}
\label{asynorm3-c}
S_{0}&=S_{H}+\overline{D}_{\overline{H}}S_{\mathbf{\zeta}} \overline{D}^{T}_{\overline{H}}+b^{2}S_{q}
\end{align}
In particular, at any sensor~$n$, the estimate sequence,
$\left\{\mathbf{x}_{n}(i)\right\}_{i\geq 0}$ is asymptotically normal:
\begin{equation}
\label{asynorm6}
\sqrt(i)\left(\mathbf{x}_{n}(i)-\mathbf{\theta}^{\ast}\right) \Longrightarrow\:\mathcal{N}(\mathbf{0},S_{nn}(\mathbf{\theta}^{\ast}))
\end{equation}
where, $S_{nn}(\mathbf{\theta}^{\ast})\in\mathbb{R}^{M\times M}$
denotes the $n$-th principal block of $S(\mathbf{\theta}^{\ast})$.
\end{theorem}

\begin{proof}
The proof involves a step-by-step verification of
Assumptions~\textbf{(C.1)-(C.5)} of Theorem~\ref{RM} (Appendix~\ref{res_stoch_app}), since the
Assumptions~\textbf{(B.1)-(B.5)} are already shown to be satisfied
(see, Theorem~\ref{REas}.) Recall
$R\left(\mathbf{x}\right)$ and $\Gamma\left(i+1,\mathbf{x},\omega\right)$ from
Theorem~\ref{REas} (\eqref{Reas:5}-\eqref{Reas:6}).
From \eqref{Reas:5}, Assumption~\textbf{(C.1)} of Theorem~\ref{RM} is satisfied with
\begin{equation}
\label{asynorm9} B=-\left[b\overline{L}\otimes I_{M}+D_{\overline{H}}\right]
\end{equation}
and $\delta\left(\mathbf{x}\right)\equiv 0$. Assumption~\textbf{(C.2)} is satisfied by hypothesis, while the condition $a>\frac{1}{2\lambda_{\min} \left(b\overline{L}\otimes I_{M}+D_{\overline{H}}\right)}$ implies
\begin{align*}
\Sigma&=-a\left[b\overline{L}\otimes
I_{M}+D_{\overline{H}}\right]+\frac{1}{2}I_{NM}
=aB+\frac{1}{2}I_{NM}
\end{align*}
is stable, and hence Assumption~\textbf{(C.3)}. To verify
Assumption~\textbf{(C.4)}, we have from Assumption~\textbf{(A.4)}
\begin{align}
\label{asynorm11}
A\left(i,\mathbf{x}\right) & =
\mathbb{E}\left[\Gamma\left(i+1,\mathbf{x},\omega\right)
\Gamma^{T}\left(i+1,\mathbf{x},\omega\right)\right]
\\
\nonumber
&
=  b^{2}\mathbb{E}\left[\left(\widetilde{L}(i)\otimes I_{M}\right)\mathbf{x}\mathbf{x}^{T}\left(\widetilde{L}(i)\otimes I_{M}\right)^{T}\right]
\\
\nonumber
&
+
\mathbb{E}\left[\left(\overline{D}_{\overline{H}}\mathbf{z}(i)
-D_{\overline{H}}\mathbf{1}_{N}\otimes\mathbf{\theta}^{\ast}\right)
\right.
\\
\nonumber
&
\left.
\hspace{2cm}
\left(\overline{D}_{\overline{H}}\mathbf{z}(i)
-D_{\overline{H}}\mathbf{1}_{N}\otimes\mathbf{\theta}^{\ast}\right)^{T}
\right]
\\
\nonumber
&
+
b^{2}\mathbb{E}\left[\left(\mathbf{\Upsilon}(i)
+\mathbf{\Psi}(i)\right)\left(\mathbf{\Upsilon}(i)
+\mathbf{\Psi}(i)\right)^{T}\right]
\end{align}
From the i.i.d.~assumptions, we note that all the three terms on the R.H.S.~of \eqref{asynorm11} are independent of~$i$, and, in particular, the last two terms are constants. For the first term, we note that
\begin{equation}
\label{asynorm12}
\hspace{-.05cm}
\lim_{\mathbf{x}\rightarrow\mathbf{1}_{N}\otimes\mathbf{\theta}^{\ast}}
\hspace{-.2cm}
\mathbb{E}\left[\left(\widetilde{L}(i)\otimes
I_{M}\right)
\hspace{-.05cm}
\mathbf{x}\mathbf{x}^{T}
\hspace{-.1cm}
\left(\widetilde{L}(i)\otimes
I_{M}\right)^{T}\right]
\hspace{-.15cm}
= \mathbf{0}
\end{equation}
from the bounded convergence theorem, as the entries of
$\left\{\widetilde{L}(i)\right\}_{i\geq 0}$ are bounded and
\begin{equation}
\label{asynorm13} \left(\widetilde{L}(i)\otimes
I_{M}\right)\left(\mathbf{1}_{N}\otimes\mathbf{\theta}^{\ast}\right)=\mathbf{0}
\end{equation}
For the second term on the R.H.S.~of \eqref{asynorm11}, we
have
\begin{align}
\nonumber
&
\mathbb{E}
\hspace{-.05cm}
\left[
\hspace{-.05cm}
\left(\overline{D}_{\overline{H}}\mathbf{z}(i)
\hspace{-.05cm}
-
\hspace{-.05cm}
D_{\overline{H}}\mathbf{1}_{N}\otimes\mathbf{\theta}^{\ast}\right)
\hspace{-.1cm}
\left(\overline{D}_{\overline{H}}\mathbf{z}(i)
\hspace{-.05cm}
-
\hspace{-.05cm}
D_{\overline{H}}\mathbf{1}_{N}\otimes\mathbf{\theta}^{\ast}\right)^{T}\right]
\\
\nonumber
&
=
\mathbb{E}\left[\left(\overline{D}_{\overline{H}}\left[
\begin{array}{lll}
                    \widetilde{H}_{1}(i) &  &  \\
                    \ddots & \ddots & \ddots \\
                    & & \widetilde{H}_{N}(i)
                   \end{array}
          \right]\mathbf{1}_{N}\mathbf{\theta}^{\ast}\right)
\right.
\nonumber
\\
\nonumber
&
\hspace{1.25cm}
\left.\left(\overline{D}_{\overline{H}}\left[
\begin{array}{lll}
                    \widetilde{H}_{1}(i) &  &  \\
                    \ddots & \ddots & \ddots \\
                    & & \widetilde{H}_{N}(i)
                   \end{array}
          \right]\mathbf{1}_{N}\mathbf{\theta}^{\ast}\right)^{T}\right]
+
\\
\nonumber
&
\hspace{2cm}
+
          \mathbb{E}\left[\overline{D}_{\overline{H}}
          \mathbf{\zeta}\mathbf{\zeta}^{T}
          \overline{D}^{T}_{\overline{H}}\right]\nonumber
\\
\label{asynorm14}
&
= S_{H}+\overline{D}_{\overline{H}}S_{\mathbf{\zeta}}\overline{D}^{T}_{\overline{H}}
\end{align}
where the last step follows from
\eqref{asnor1},\eqref{obsnoise}. Finally, we note the third
term on the R.H.S. of \eqref{asynorm11} is $b^{2}S_{q}$, see~\eqref{dith5}. We thus have from
\eqref{asynorm11}-\eqref{asynorm14}
\begin{equation}
\label{asynorm15}
\lim_{i\rightarrow\infty,~\mathbf{x}\rightarrow\mathbf{x}^{\ast}}A\left(i,\mathbf{x}\right)
=
S_{H}+\overline{D}_{\overline{H}}S_{\mathbf{\zeta}}\overline{D}^{T}_{\overline{H}}+b^{2}S_{q}\nonumber
= S_{0}
\end{equation}
We now verify Assumption~\textbf{(C.5)}. Consider a fixed $\epsilon>0$. We note that \eqref{assc5.1} is a restatement of the uniform integrability of the random family, $\left\{\|\Gamma\left(i+1,\mathbf{x},\omega\right)\|^{2} \right\}_{i\geq 0, \,\|\mathbf{x}-\mathbf{\theta}^{\ast}\|<\epsilon}$. From
\eqref{Reas:6}, we have
\begin{align}
&\left\|\Gamma\left(i+1,\mathbf{x},\omega\right)\right\|^{2}
= \left\|b\left(\widetilde{L}(i)\otimes
I_{M}\right)\mathbf{x}-
\right.
\\
\nonumber
&
\left.
\hspace{.75cm}
-
\left(\overline{D}_{\overline{H}}\mathbf{z}(i)
-D_{\overline{H}}\mathbf{1}_{N}
\otimes\mathbf{\theta}^{\ast}\right)+b\mathbf{\Upsilon}(i)
+b\mathbf{\Psi}(i)\right\|^{2}\nonumber
\\
&
\label{asynorm16}
= \left\|b\left(\widetilde{L}(i)\otimes
I_{M}\right)\left(\mathbf{x}-\mathbf{\theta}^{\ast}\right)
-
\right.
\\
\nonumber
&
\left.
\hspace{.75cm}
-
\left(\overline{D}_{\overline{H}}\mathbf{z}(i)-D_{\overline{H}}\mathbf{1}_{N}
\otimes\mathbf{\theta}^{\ast}\right)+b\mathbf{\Upsilon}(i)
+b\mathbf{\Psi}(i)\right\|^{2}
\\
\nonumber
&
 \leq  9\left[\left\|\left(b\widetilde{L}(i)\otimes
I_{M}\right)\left(\mathbf{x}-\mathbf{\theta}^{\ast}\right)\right\|^{2}
+
\right.
\\
\nonumber
&
\left.
+
\left\|\overline{D}_{\overline{H}}\mathbf{z}(i)
-D_{\overline{H}}\mathbf{1}_{N}\otimes\mathbf{\theta}^{\ast}\right\|^{2}
+b^{2}\left\|\mathbf{\Upsilon}(i)+\mathbf{\Psi}(i)\right\|^{2}\right]
\end{align}
where we used~\eqref{asynorm13} and the inequality,
$\|\mathbf{y}_{1}+\mathbf{y}_{2}+\mathbf{y}_{3}\|^{2}\leq
9\left[\left\|\mathbf{y}_{1}\right\|^{2}
+\left\|\mathbf{y}_{2}\right\|^{2}+\left\|\mathbf{y}_{3}\right\|^{2}\right]$,
for vectors $\mathbf{y}_{1},\mathbf{y}_{2},\mathbf{y}_{3}$. From
\eqref{Reas:13} we note that, if
$\|\mathbf{x}-\mathbf{\theta}^{\ast}\|<\epsilon$,
\begin{equation}
\label{asynorm17} \left\|\left(b\widetilde{L}(i)\otimes
I_{M}\right)\left(\mathbf{x}-\mathbf{\theta}^{\ast}\right)\right\|^{2}\leq
c_{2}\epsilon^{2}
\end{equation}
From~(\ref{asynorm16}), the family [defined in~\eqref{asynorm18} below]
\[
\left\{\widetilde{\Gamma}\left(i+1,\mathbf{x},\omega\right)\right\}_{i\geq
0, \|\mathbf{x}-\mathbf{\theta}^{\ast}\|<\epsilon}
 \]
 dominates the family
\[
\left\{\|\Gamma\left(i+1,\mathbf{x},\omega\right)\|^{2}\right\}_{i\geq
0,~\|\mathbf{x}-\mathbf{\theta}^{\ast}\|<\epsilon},
\]
 where
\begin{align}
\nonumber
&
\widetilde{\Gamma}\left(i+1,\mathbf{x},\omega\right) =
9\left[c_{2}\epsilon^{2}+\left\|\overline{D}_{\overline{H}}
\mathbf{z}(i)-D_{\overline{H}}\mathbf{1}_{N}\otimes
\mathbf{\theta}^{\ast}\right\|^{2}
\right.
\\
\label{asynorm18}
&
\left.
\hspace{3.5cm}
+\,
b^{2}\left\|\mathbf{\Upsilon}(i)+\mathbf{\Psi}(i)\right\|^{2}\right]
\end{align}
The family
$\left\{\widetilde{\Gamma}\left(i+1,\mathbf{x},\omega\right)\right\}_{i\geq
0,~\|\mathbf{x}-\mathbf{\theta}^{\ast}\|<\epsilon}$ is i.i.d. and
hence uniformly integrable (see~\cite{Kallenberg}). Then the
family $\left\{\|\Gamma\left(i+1,\mathbf{x},\omega\right)\|^{2}\right\}_{i\geq
0,~\|\mathbf{x}-\mathbf{\theta}^{\ast}\|<\epsilon}$ is also
uniformly integrable since it is dominated by the uniformly
integrable family
$\left\{\widetilde{\Gamma}\left(i+1,\mathbf{x},\omega\right)\right\}_{i\geq
0,~\|\mathbf{x}-\mathbf{\theta}^{\ast}\|<\epsilon}$
(see~\cite{Kallenberg}). Thus~\textbf{(C.1)-(C.5)}
are verified and the theorem follows.
\end{proof}

\subsection{A Simulation Example}
\label{sim-ex}
Fig.~\ref{fig:LU}~(b) shows the performance  of $\mathcal{LU}$ for the network of $N=45$ sensors in Fig.~\ref{fig:LU}~(a), where the sensors are deployed randomly on a $25\times 25$ grid. The sensors communicate in a fixed radius and are further constrained to have a maximum of $6$ neighbors per node. The true parameter $\mathbf{\theta}^{\ast}\in\mathbb{R}^{45}$.  Each node is associated with a single component of $\mathbf{\theta}^{\ast}$, i.e., $\overline{H}_n=\mathbf{e}_n^T$, the unit vector of zeros, except entry~$n$ that is~1. For the experiment, each component of $\mathbf{\theta}^{\ast}$ is generated by an instantiation of a zero mean Gaussian random variable of variance 25. The parameter $\mathbf{\theta}^{\ast}$ represents the state of the field to be estimated. In this example, the field is  white, stationary, and hence each sample of the field has the same Gaussian distribution and is independent of the others. More generally, the components of $\mathbf{\theta}^{\ast}$ may correspond to random field samples, as dictated by the sensor deployment, that can possibly arise from the discretization of a field governed by a PDE. Each sensor observes the corresponding field component in additive Gaussian noise. For example, sensor 1 observes $z_{1}(t)=\theta^{\ast}_{1}+\zeta_{1}(t)$, where $\zeta_{1}(t)\sim\mathcal{N}(0,1)$. Clearly, such a model satisfies the distributed observability condition
\begin{equation}
\label{fig_LU1}
G=\sum_{n=1}^{N}\overline{H}_{n}^{T}\overline{H}_{n}=I_{45}=G^{-1}
\end{equation}
Fig.~\ref{fig:LU}(b) shows the normalized error at every sensor plotted against the iteration index $i$ for an instantiation of the algorithm. The normalized error for the $n$-th sensor at time $i$ is given by the quantity $\left\|\mathbf{x}_{n}(i)-\mathbf{\theta}^{\ast}\right\|/45$, i.e., the estimation error normalized by the dimension of $\mathbf{\theta}^{\ast}$. We note that the errors converge to zero as established by the theoretical findings. The decrease is rapid at the beginning and slows down at $i$ increases. This is a standard property of stochastic approximation based algorithms, consequence of the decreasing weight sequence $\alpha(i)$ required for convergence.
\begin{figure*}
\centering %
\subfigure[][]{\includegraphics[height=2.1in, width=2.5in ]{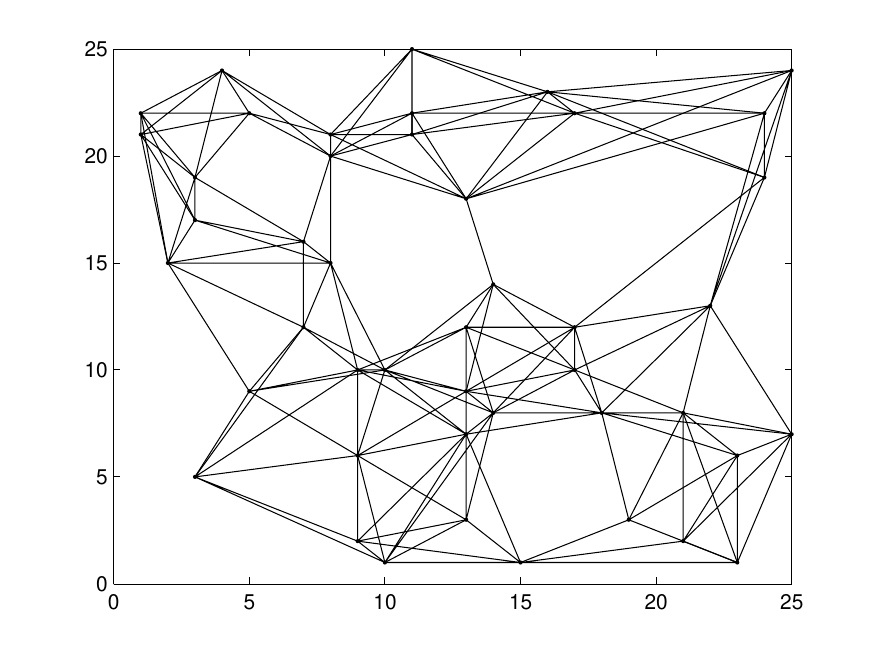}}
\subfigure[][]{\includegraphics[height=2.1in, width=2.5in ]{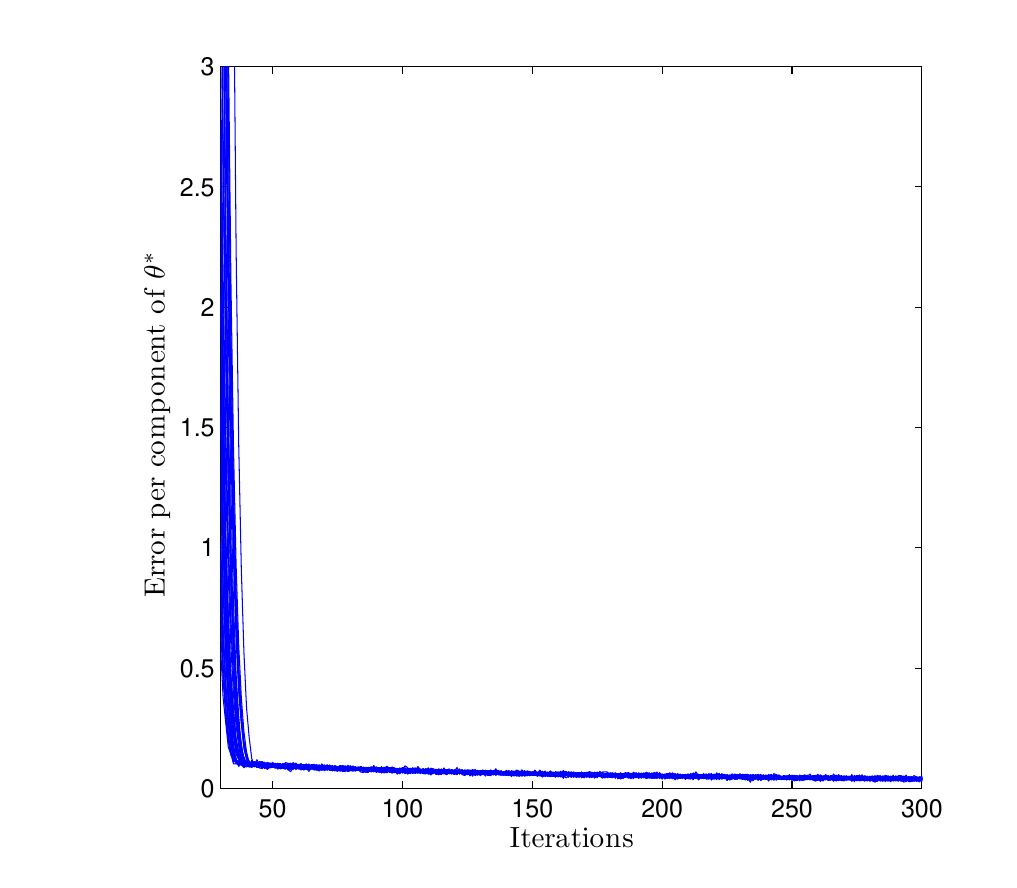}}
\caption{Illustration of distributed linear parameter estimation. (a)
Example network deployment of 45 nodes.  (b) Convergence of normalized estimation error
at each sensor.}
\label{fig:LU}
\end{figure*}
 From the plots, although the individual sensors are low rank observations of the true parameter, by collaborating, each sensor reconstructs the true parameter value, as desired.
\subsection{An Example}
\label{lbm-numstud} From Theorem~\ref{asynorm} and
\eqref{asnor1}, we note that the asymptotic variance is
independent of $\mathbf{\theta}^{\ast}$, if the observation
matrices are non-random. In that case, it is possible to optimize
(minimize) the asymptotic variance over the weights $a$ and $b$.
In the following, we study a special case permitting explicit
computations and that leads to interesting results.
Consider a scalar parameter $(M=1)$ and let each sensor $n$ have
the same i.i.d. observation model,
\begin{equation}
\label{numstud1} z_{n}(i)=h\theta^{\ast}+\zeta_{n}(i)
\end{equation}
where $h\neq 0$ and $\{\zeta_{n}(i)\}_{i\geq 0,~1\leq n\leq N}$ is
a family of independent zero mean Gaussian random variables with
variance $\sigma^{2}$. In addition, assume unquantized
inter-sensor exchanges. We define the average asymptotic variance
per sensor attained by the algorithm $\mathcal{LU}$ as
\begin{equation}
\label{numstud2}
S_{\mathcal{LU}}=\frac{1}{N}\mbox{Tr}\left(S\right)
\end{equation}
where~$S$ is given by \eqref{asynorm3} in Theorem~\ref{asynorm}. From Theorem~\ref{asynorm}, we have $S_{0}=\sigma^{2}h^{2}I_{N}$ and, hence, from~\eqref{asynorm3}
\begin{align}
\label{numstud3} S_{\mathcal{LU}} & =
\frac{a^{2}\sigma^{2}h^{2}}{N}\mbox{Tr}\left(\int_{0}^{\infty}e^{2\Sigma
v}dv\right)
\\
\nonumber
&
=
\frac{a^{2}\sigma^{2}h^{2}}{N}\int_{0}^{\infty}\mbox{Tr}\left(e^{2\Sigma
v}\right)dv
\end{align}
From \eqref{asynorm3-b} the eigenvalues of $2\Sigma v$ are
$\left[-2ab\lambda_{n}(\overline{L})-\left(2ah^{2}-1\right)\right]v$
for $1\leq n\leq N$ and we have
\begin{align}
\nonumber
S_{\mathcal{LU}}
&
=
\frac{a^{2}\sigma^{2}h^{2}}{N}\sum_{n=1}^{N}\int_{0}^{\infty}
e^{\left[-2ab\lambda_{n}(\overline{L})-\left(2ah^{2}-1\right)\right]v}dv
\\
\nonumber
&
=
\frac{a^{2}\sigma^{2}h^{2}}{N}\sum_{n=1}^{N}\frac{1}{2ab\lambda_{n}(\overline{L})+\left(2ah^{2}-1\right)}\nonumber
\\
\nonumber
&
= \frac{a^{2}\sigma^{2}h^{2}}{N\left(2ah^{2}-1\right)}
+\frac{a^{2}\sigma^{2}h^{2}}{N}
\\
\label{numstud4}
&
\hspace{1cm}
\sum_{n=2}^{N}\frac{1}{2ab\lambda_{n}(\overline{L})+\left(2ah^{2}-1\right)}
\end{align}
In this case, the constraint
$a>\frac{1}{2\lambda_{\min}(b\overline{L}\otimes
I_{M}+D_{\overline{H}})}$ in Theorem~\ref{asynorm} reduces to
$a>\frac{1}{2h^{2}}$, and hence the problem of optimum $a,b$
design to minimize $S_{\mathcal{LU}}$ is given by
\begin{equation}
\label{numstud5}
S_{\mathcal{LU}}^{\ast}=\inf_{a>\frac{1}{2h^{2}},~b>0}S_{\mathcal{LU}}
\end{equation}
It is to be noted, that the first term on the last step of
\eqref{numstud4} is minimized at $a=\frac{1}{h^{2}}$ and the
second term (always non-negative under the constraint) goes to
zero as $b\rightarrow\infty$ for any fixed $a>0$. Hence, we have
\begin{equation}
\label{numstud6} S_{\mathcal{LU}}^{\ast}=\frac{\sigma^{2}}{Nh^{2}}
\end{equation}
The above shows that, by setting $a=\frac{1}{h^{2}}$ and $b$
sufficiently large in $\mathcal{LU}$,
$S_{\mathcal{LU}}$ is arbitrarily close to $S_{\mathcal{LU}}^{\ast}$.

We compare this optimum achievable asymptotic variance per sensor,
$S_{\mathcal{LU}}^{\ast}$, attained by
$\mathcal{LU}$ to that attained by a centralized scheme.
In the centralized scheme, there is a central estimator, which
receives measurements from all the sensors and computes an
estimate based on all measurements. In this case, the sample mean
estimator is an efficient estimator (in the sense of Cram\'er-Rao)
and the estimate sequence $\{x_{c}(i)\}_{i\geq 0}$ is given by
\begin{equation}
\label{numstud9} x_{c}(i)=\frac{1}{Nih}\sum_{n,i}z_{n}(i)
\end{equation}
and we have
\begin{equation}
\label{numstud10}
\sqrt{i}\left(x_{c}(i)-\theta^{\ast}\right)\,\sim\, \left(0,\mathcal{S}_{c} \right)
\end{equation}
where, $S_{c}$ is the variance (which is also the one-step Fisher information in this case, see,~\cite{Lehmann-estimation}) and is given by
\begin{equation}
\label{numstud7} S_{c}=\frac{\sigma^{2}}{Nh^{2}}
\end{equation}
From \eqref{numstud6} we note that,
\begin{equation}
\label{numstud8} S_{\mathcal{LU}}^{\ast}=S_{c}
\end{equation}
Thus the average asymptotic variance attainable by the distributed algorithm $\mathcal{LU}$ is the same as that of the optimum (in the sense of Cram\'er-Rao) centralized estimator having access to all information simultaneously. This is an interesting result, as it holds irrespective of the network topology. In particular, however sparse the inter-sensor communication graph is, the
optimum achievable asymptotic variance is the same as that of the
centralized efficient estimator. Note that weak convergence itself is a limiting result, and, hence, the rate of convergence in \eqref{asynorm2} in Theorem~\ref{asynorm} will, in general, depend on the network topology.

\subsection{Some generalizations}
\label{gen-LU} We discuss some generalizations of the basic $\mathcal{LU}$ scheme before proceeding to the nonlinear observation models addressed in the subsequent sections. We start by revisiting the scalar example in Section~\ref{lbm-numstud} for which the distributed $\mathcal{LU}$ is shown to achieve the performance of the optimal centralized estimator. Interestingly, the above example is not an isolated special case and has several important implications. The observation that by increasing $b>0$ we can achieve asymptotic variance as close as desired to the centralized estimator hints to a more general time-scale separation in the case of unquantized transmissions. Intuitively, for a fixed $b>0$, the weight associated to the consensus potential is $b\alpha(i)$, which goes to zero at the same rate as that of the innovation potential. Hence, in the long run, a non-negligible (in the scale $\{\alpha(i)\}$) amount of time is required to disseminate new information acquired by a sensor. In other words, the rate of uncertainty reduction in $\mathcal{LU}$ depends on both the rate of new information acquisition at the sensors and the rate of information dissemination in the network. On the contrary, in a centralized scenario, no additional time is incurred for information dissemination and the rate of uncertainty reduction is the same as the rate of information acquisition. This is manifested, in general, in the asymptotic variance of $\mathcal{LU}$, which is larger than its centralized counterpart due to the additional overhead of the mixing terms (Theorem~\ref{asynorm}). This suggests that, if the mixing can be carried out at a \emph{faster} scale, the additional overhead due to the mixing time will not be observed at the time scale of observation acquisition and, in effect, the distributed scheme will lead to similar asymptotic variance as in the centralized setting. This is noted in Section~\ref{lbm-numstud}, where increasing the relative weight $b$ of the consensus or mixing potential leads to a time scale separation between information dissemination and acquisition. Increasing $b$ beyond bounds suggests that we replace the decreasing weight sequence $\{\alpha(i)\}$ from the consensus term and retain it with a constant weight, or more generally, a weight sequence $\{\beta(i)\}$, that asymptotically dominates $\{\alpha(i)\}$. Such a mixed time scale extension of the $\mathcal{LU}$ is introduced and analyzed in~\cite{KM-JSTSP}. The results in~\cite{KM-JSTSP} show that the conclusion in Section~\ref{lbm-numstud} for the scalar example (the distributed achieves the centralized performance in terms of asymptotic variance) holds in more general vector parameter settings by appropriately tuning the consensus and innovation weights. This is significant, as it justifies the applicability of distributed estimation schemes over centralized approaches.

The development in this paper assumes stationarity of the sensor observations over time. While this is applicable and is a commonly used assumption in many statistical models, some scenarios inherently lead to non-stationary observation time series. For example, consider a distributed sensor network monitoring a target that fades over time. In this example, the sensor observation models are no longer stationary as the SNR (signal to noise ratio) decays over time, the decay rate being a function of the fading characteristics. Treating nonstationarity requires modification of the algorithm (intuitively, the update rules are no longer stationary) and is pursued in~\cite{KM-JSTSP}.

The above did not exploit the physical significance of the parameter $\mathbf{\theta}$. That $\mathbf{\theta}$ may itself come from a spatially distributed random field was only implicit in the distributed observation model. Typical examples include instrumenting a spatially distributed random field (say a temperature surface) with a sensor network. Another example is of cyberphysical systems, where a network of physical entities equipped with sensors are deployed over a large geographical region. A well known example in this setting is the power grid, a large distributed network of generators and loads. Our results imply that, under appropriate observability conditions, the physical field $\mathbf{\theta}$ may be reconstructed completely at each node (sensor).\footnote{A node, in this context, refers to the physical entity at a geographical location, for example, a generator in a power grid. The sensing or measurement unit associated to a node is referred to as a sensor. The state of a node represents the field intensity at that point, for example, the phase of a generator.} However, for such systems, the parameter $\mathbf{\theta}$ representing the physical field is quite large dimensional, may be of the order of $10^{3}$ or more, as exemplified by the power grid.\footnote{For problems involving infinite dimensional systems, such as the temperature distribution over a domain in the Euclidean space, any reasonable discretization would lead to a large dimensional $\mathbf{\theta}$.} It is then impractical and unnecessary to reconstruct the high dimensional parameter in its entirety at each node. On the other hand, the node may be interested only in its state, or those of its close neighbors. In general, the observation at each sensor reflects the coupling of a few local physical states and hence, acting alone, a node may not be able to recover its state uniquely. In~\cite{Thesis-Kar}, we develop approaches to address this problem, where each node wants to reconstruct a few components\footnote{These components may vary from node to node.} of the large state vector. The estimation approach would lead to low dimensional data exchanges between neighboring sensors (nodes) and local estimate updates would involve only those components, the node wants to reconstruct. Due to the partial information exchange between sensors and the fact that sensors may have different goals, the distributed observability no longer culminates to the sum of network connectivity and global connectivity, but requires more subtle relations between the observation model and the network topology. In general, the scope of such problems of distributed estimation with partial inter-sensor information exchange is quite broad and challenging, and we refer the reader to~\cite{Thesis-Kar} (Chapter 5) for an exposition.

\section{Nonlinear Observation Models: Agorithm $\mathcal{NU}$}
\label{nonlin} The previous section developed the algorithm $\mathcal{LU}$ for distributed parameter estimation when the observation model is linear. In this section, we extend the previous development to accommodate more general classes of nonlinear observation models. We comment briefly on the organization of this section. In Section~\ref{nlprobform}, we introduce notation and setup the problem, and in Section~\ref{nlalg1} we present the $\mathcal{NU}$ algorithm for distributed parameter estimation for nonlinear observation models and establish conditions for its consistency.

%

\subsection{Nonlinear Observation Models}
\label{nonlin-obs-mod} Similar
to Section~\ref{lbm}, let
$\mathbf{\theta}^{\ast}\in\mathcal{U}\subset\mathbb{R}^{M}$ be the true but unknown parameter value. We assume that the domain $\mathcal{U}$ is an open set in $\mathbb{R}^{M}$. In the general case,
 the observation model at each sensor $n$ consists
of an i.i.d. sequence $\left\{\mathbf{z}_{n}(i)\right\}_{i\geq 0}$ in
$\mathbb{R}^{M_{N}}$ with
\begin{equation}
\label{nlprobform1}
\mathbb{P}_{\mathbf{\theta}^{\ast}}[\mathbf{z}_{n}(i)\in\mathcal{D}]=\int_{\mathcal{D}}dF_{n,\mathbf{\theta}^{\ast}},~~\forall~\mathcal{D}\in\mathbb{B}^{M_{N}}
\end{equation}
where $F_{n,\mathbf{\theta}^{\ast}}$ denotes the distribution function of the random vector $\mathbf{z}_{n}(i)$. For consistent parameter estimates, even in centralized settings, some form of observability needs to be imposed on the nonlinear model. In the following, we assume that the distributed observation model is \emph{separably estimable}, a notion which we introduce now.
\begin{definition}[Separably Estimable]\label{sep_est_def} Let $\left\{\mathbf{z}_{n}(i)\right\}_{i\geq 0}$ be the i.i.d. observation
sequence at sensor~$n$, where $1\leq n\leq N$. We call the parameter estimation problem to be separably estimable, if there exist functions $g_{n}(\cdot):\mathbb{R}^{M_{N}}\longmapsto\mathbb{R}^{\overline{M}},\, \forall 1\leq n\leq N$, such that the function $h(\cdot):\mathcal{U}\longmapsto\mathbb{R}^{\overline{M}}$ is continuous and invertible on $\mathcal{U}$
\begin{equation}
\label{nlprobform2}
h(\mathbf{\theta})=\frac{1}{N}\sum_{n=1}^{N}
\mathbb{E}_{\mathbf{\theta}}\left[g_{n}(\mathbf{z}_{n}(i))\right]
\end{equation}
\end{definition}
\begin{remark}
\label{rem-sepest}
Before providing examples of separably estimable observation models and demonstrating the applicability of the notion, we comment on the definition.
\begin{itemize}
\item[(i)] We note that the factor $\frac{1}{N}$ in
\eqref{nlprobform2} is just for notational convenience, as
will be seen later. In fact, the $\frac{1}{N}$ can be absorbed by redefining the functions $g_{n}(\cdot)$. Also, it is implicitly assumed that the random vectors $g_{n}(\mathbf{z}_{n}(i))$ are integrable w.r.t. the measures $\mathbb{P}_{\mathbf{\theta}}$ for $\mathbf{\theta}\in\mathcal{U}$.
\item[(ii)] Let $h(\mathcal{U})\subset\mathbb{R}^{\overline{M}}$ denote the range of $h(\cdot)$. The continuity of $h(\cdot)$ implies that $h(\mathcal{U})$ is open. Let $h^{-1}:h(\mathcal{U})\longmapsto\mathbb{R}^{M}$ denote the inverse of $h(\cdot)$ (which is necessarily continuous on $h(\mathcal{U})$.) It then follows that $h^{-1}(\cdot)$ has a measurable extension defined over all of $\mathbb{R}^{\overline{M}}$. In the following we will assume that $h^{-1}(\cdot)$ has been measurably extended and, by abusing notation, denote this extension by $h^{-1}$.
\item[(iii)] We will show that the notion of separably estimable models introduced above is, in fact, necessary and sufficient to guarantee the existence of consistent distributed estimation procedures for a wide range of practical scenarios. This condition may also be viewed as a natural generalization of the observability constraint of Assumption~\textbf{(A.5)} in the linear model. Indeed, if, assuming the linear model, we define $g_{n}(\mathbf{z}_{n}(i))=\overline{H}_{n}^{T}\mathbf{z}_{n}(i), \,\forall 1\leq n\leq N$ in \eqref{nlprobform2}, we have $h(\mathbf{\theta})=G\mathbf{\theta}$, where $G$ is defined in \eqref{G}. Then, invertibility of~(\ref{nlprobform2}) is equivalent to Assumption~\textbf{(A.5)}, i.e., to invertibility of~$G$; hence, the linear model is an example of a separably estimable problem. Note that, if an observation model is separably estimable, then the choice of functions $g_{n}(\cdot)$ is not unique. Indeed, given a separably estimable model, it is important to figure out an appropriate decomposition, as in \eqref{nlprobform2}, because the convergence properties of the algorithms (Algorithm $\mathcal{NU}$, Section~\ref{nlalg1}) to be studied are intimately related to the behavior of these functions. Finally, we note that, in general, $\overline{M}\neq M$, and the dimension $\overline{M}$ of the range space of $h(\cdot)$ is very much linked to the memory and transmission requirements of the distributed algorithm $\mathcal{NLU}$ to be studied in Section~\ref{nlalg2}. In this sense, the function $h(\cdot)$ plays the role of a complete sufficient statistic as used in classical (centralized) estimation, the major difference being the \emph{distributed computability} (to be made precise later) of $h(\cdot)$ in the current setting.
\end{itemize}
\end{remark}

In Sections~\ref{nlalg1} and~\ref{nlalg2}, respectively, we will present algorithms $\mathcal{NU}$ and $\mathcal{NLU}$ for distributed parameter estimation in separably estimable models. While the $\mathcal{NLU}$ provides consistent parameter estimates for all separably estimable models, the $\mathcal{NU}$ requires further (mainly of the Lipschitz type) conditions on the functions $g_{n}(\cdot)$ and $h(\cdot)$. However, in cases where the $\mathcal{NU}$ is applicable, it automatically leads to convergence rate guarantees in the context of asymptotic normality. These differences are further clarified in Section~\ref{conclusion}. Before discussing these algorithms in detail, we provide examples of separably estimably models in the following.

\textbf{Examples: Signal in additive noise models}

We now demonstrate an important and large class of distributed observation models possessing the separably estimable property, thus justifying the generality and applicability of the notion.

A wide range of observation models are of the signal in additive noise type. In particular, for each $n$, denote by $\{\mathbf{\zeta}_{n}(i)\}$ the zero mean i.i.d. observation noise at the $n$-th sensor of arbitrary distribution (the distribution may vary from sensor to sensor.) The sensor observation model is said to be of \emph{signal in additive noise type}, if the observation sequence $\{\mathbf{z}_{n}(i)\}$ at the $n$-th sensor is of the form:
\begin{equation}
\label{sigaddnoise}
\mathbf{z}_{n}(i)=f_{n}(\mathbf{\theta}^{\ast})+\mathbf{\zeta}_{n}(i)
\end{equation}
Here $f_{n}:\mathcal{U}\longmapsto\mathbb{R}^{M_{n}}$ denotes the transformed (nonlinearly) signal (or parameter) observed at sensor $n$, further corrupted by additive noise.
The following simple proposition characterizes the subclass of signal in additive noise observable models with the separably estimable property:
\begin{proposition}
\label{prop-sigadd}
Let $f:\mathcal{U}\longmapsto\mathbb{R}^{\sum_{n=1}^{N}M_{n}}$ be defined by, $f(\mathbf{\theta})=[f_{1}^{T}(\mathbf{\theta})\cdots f_{N}^{T}(\mathbf{\theta})]^{T}$. Then, the above signal in additive noise observation model (see~\eqref{sigaddnoise}) is separably estimable if $f(\cdot)$ is continuous and invertible on $\mathcal{U}$.
\end{proposition}
Before providing the rather straightforward proof, we note the consequences of Proposition~\ref{prop-sigadd}. Consider a hypothetical centralized estimator having access to all the sensor observations at all times. Clearly, the $\sum_{n=1}^{N}M_{n}$ dimensional i.i.d. observation sequence $\{\mathbf{z}(i)\}$ at such a center is given by:
\begin{equation}
\label{prop-sigadd1}
\mathbf{z}(i)=f(\mathbf{\theta}^{\ast})+\mathbf{\zeta}(i)
\end{equation}
In general, for arbitrary statistics of the noise sequence $\{\mathbf{\zeta}(i)\}$, it is necessary that the function $f$ be invertible, for the center to yield a consistent estimate of the parameter. In fact, for consistent centralized estimates, the invertibility of $f(\cdot)$ is required, even when the observation noise is identically zero. On the other hand, Proposition~\ref{prop-sigadd} asserts that the invertibility of $f(\cdot)$ (and its continuity) is sufficient to guarantee that the model is separably estimable and hence the existence of consistent distributed estimation schemes. Hence, at least in the class of widely adopted signal in additive noise models, centralized observability is equivalent to distributed observability (formulated here in terms of separable estimability.) This further justifies the notion of separable estimability as a reasonable generalization of the concept of centralized observability to distributed nonlinear settings. In Section~\ref{NLU-power}, we will show that the $\mathcal{NLU}$ algorithm provides a completely distributed approach to the static phase estimation problem in power grids of generators and loads based on line flow measurements, an important practical example of a distributed nonlinear signal in additive noise model.
\begin{proof}
The proof follows in a straightforward manner from the definition. For each $n$, define the function $g_{n}:\mathbb{R}^{M_{n}}\longmapsto\mathbb{R}^{\sum_{n=1}^{N}M_{n}}$ by
\begin{equation}
\label{prop-sigadd2}
g_{n}(\mathbf{y})=[\mathbf{0}_{M_{1}}^{T} \mathbf{0}_{M_{2}}^{T}\cdots\mathbf{y}^{T}\cdots\mathbf{0}_{M_{N}}^{T}]^{T},~~\forall\mathbf{y}\in\mathbb{R}^{M_{n}}
\end{equation}
Recall $\mathbf{0}_{M_{1}}\in\mathbb{R}^{M_{1}}$ denotes the column vector of $M_{1}$ zeros and so on. By the independence of the noise sequence $\{\mathbf{\zeta}(i)\}$, it the follows that
\begin{equation}
\label{prop-sigadd3}
\sum_{n=1}^{N}\mathbb{E}_{\mathbf{\theta}^{\ast}}\left[g_{n}(\mathbf{z}_{n}(i))\right]=f(\mathbf{\theta}^{\ast})
\end{equation}
The continuity and invertibility of $f(\cdot)$ then establishes the separable estimability of the model (Definition~\ref{sep_est_def}) by the correspondence $h(\cdot)=\frac{1}{N}f(\cdot)$.
\end{proof}
By using the same arguments we demonstrate a larger class of separably estimable models as follows:
\begin{proposition}
\label{prop-sigadd4}
Let the observation sequence $\{\mathbf{z}_{n}(i)\}$ at the $n$-th sensor be of the form:
\begin{equation}
\label{prop-sigadd5}
\mathbf{z}_{n}(i)=f_{n}(\mathbf{\theta}^{\ast},\mathbf{\zeta}^{1}_{n}(i))+\mathbf{\zeta}_{n}^{2}(i)
\end{equation}
where $f_{n}:\mathcal{U}\times\mathbb{R}^{M^{1}_{n}}\longmapsto\mathbb{R}^{M_{n}}$, $\{\mathbf{\zeta}^{1}_{n}(i),\mathbf{\zeta}^{2}_{n}(i)\}\in\mathbb{R}^{M^{1}_{n}}\times\mathbb{R}^{M_{n}}$ is a temporally i.i.d. sequence and $\{\mathbf{\zeta}^{2}_{n}(i)\}$ is zero mean. Assuming that the moments exist, define the function $\overline{f}_{n}:\mathcal{U}\longmapsto\mathbb{R}^{M_{n}}$, for each $n$, by
\begin{equation}
\label{prop-sigadd6}
\overline{f}_{n}(\mathbf{\theta})=\mathbb{E}_{\mathbf{\theta}}\left[f_{n}(\mathbf{\theta}^{\ast},\mathbf{\zeta}^{1}_{n}(i))\right]
\end{equation}
Further, let $\overline{f}:\mathcal{U}\longmapsto\mathbb{R}^{\sum_{n=1}^{N}M_{n}}$ be defined by, $\overline{f}(\mathbf{\theta})=[\overline{f}_{1}^{T}(\mathbf{\theta})\cdots \overline{f}_{N}^{T}(\mathbf{\theta})]^{T}$. Then, the observation model in~\eqref{prop-sigadd5} is separably estimable if $\overline{f}(\cdot)$ is continuous and invertible on $\mathcal{U}$.
\end{proposition}
\begin{remark} The generic model considered in Proposition~\ref{prop-sigadd4} subsumes the class of signals with multiplicative noise models, by suitably defining the functions $f_{n}(\mathbf{\theta}^{\ast},\mathbf{\zeta}^{1}_{n}(i))$ and setting the additive noise component $\mathbf{\zeta}_{n}^{2}(i)$ to zero. We also note that a general guideline for choosing the functions $g_{n}(\cdot)$ for the signal in additive noise type models based on problem data is given in~\eqref{prop-sigadd2}. From a similar line of reasoning, it follows that the same choice of $g_{n}(\cdot)$ works for the larger class of separably estimable models considered in Proposition~\ref{prop-sigadd4}.
\end{remark}

In the following subsection, we present the algorithm $\mathcal{NU}$ for distributed parameter estimation in nonlinear separably estimable observation models.

\subsection{Algorithm $\mathcal{NU}$ and Assumptions}
\label{nlprobform} Before introducing the algorithm, we formally state the generic observation
and communication assumptions required by the $\mathcal{NU}$.

\begin{itemize}[\setlabelwidth{D.1)}]

\item{\textbf{(D.1)}}\textbf{Separably Estimable Model}: The nonlinear observation model~\eqref{nlprobform1} is separably estimable (Definition~\ref{sep_est_def}).

In particular, at iteration~$i$,  the observations across different sensors need not be independent. In other words, we allow spatial correlation, but require temporal independence. Also, other than the structural assumption of \emph{separable estimability}, no assumptions are required on the noise statistics, in particular, its distribution.
\item{\textbf{(D.2)}}\textbf{Random Link Failure, Quantized
Communication}: The random link failure model is the model given in Section~\ref{notgraph}; similarly, we assume quantized inter-sensor communication with subtractive dithering.
\item{\textbf{(D.3)}}\textbf{Independence and Moment Assumptions}: The sequences $\left\{L(i)\right\}_{i\geq 0}$,$\left\{\mathbf{z}_{n}(i) \right\}_{1\leq n\leq N,~i\geq 0}$,$\left\{\nu_{nl}^{m}(i)\right\}$ (dither sequence, as in \eqref{lbm-probform}) are mutually independent. Let $\overline{M}=\sum_{n=1}^N M_n$ and define $h_{n}:\mathbb{R}^{M}\longmapsto \mathbb{R}^{\overline{M}}$, by
\begin{equation}
\label{nlprobform3}
h_{n}(\mathbf{\theta})=
\mathbb{E}_{\mathbf{\theta}}\left[g_{n}(\mathbf{z}_{n}(i))\right],\:\:\forall
1\leq n\leq N
\end{equation}
We make the assumption $\forall\mathbf{\theta}\in\mathcal{U}$:
\begin{equation}
\label{nlprobform4}
\mathbb{E}_{\mathbf{\theta}}\left[\left\|\frac{1}{N}
\sum_{n=1}^{N}g_{n}(\mathbf{z}_{n}(i))-h(\mathbf{\theta})\right\|^{2}\right]
\hspace{-.075cm}
=
\hspace{-.075cm}
\eta(\mathbf{\theta})<\infty,
\end{equation}
We thus assume the existence of quadratic moments of the (transformed) random variables $g_{n}(\mathbf{z}_{n}(i))$. For example, under the reasonable hypotheses of Propositions~\ref{prop-sigadd}-\ref{prop-sigadd4}, the functions $g_{n}(\cdot)$ may be taken to be linear, and Assumption~\textbf{(D.3)} then coincides with the existence of quadratic moment of the observations $\mathbf{z}_{n}(i)$. In general, since the choice of the functions $g_{n}(\cdot)$ for a separably estimable model is not unique, the moment Assumption~\textbf{(D.3)} may enter as a selection criterion of the transformations $g_{n}(\cdot)$.
\end{itemize}
In Section~\ref{nlalg1} and Section~\ref{nlalg2}, we give two algorithms, $\mathcal{NU}$ and $\mathcal{NLU}$, respectively, for the distributed estimation problem \textbf{(D.1)-(D.3)} and provide conditions for consistency and other properties of the estimates.

\subsection{Algorithm $\mathcal{NU}$}
\label{nlalg1} In this subsection, we present the algorithm~$\mathcal{NU}$ for distributed parameter estimation in separably estimable models under Assumptions~\textbf{(D.1)-(D.3)}.

\textbf{Algorithm $\mathcal{NU}$}: Each sensor $n$ performs the
following estimate update:
\begin{align}
\label{nalg1}
&
\mathbf{x}_{n}(i+1)=\mathbf{x}_{n}(i)-\alpha(i)
\left[\sum_{l\in\Omega_{n}(i)}\beta
\left(\mathbf{x}_{n}(i)-
\right.
\right.
\\
\nonumber
&
\hspace{-1cm}
\left.
\phantom{\sum_{l\in\Omega_{n}(i)}}
\left.
\mathbf{q}(\mathbf{x}_{l}(i)
+\nu_{nl}(i))\right)+\mathcal{K}_{n}\left(h_{n}(\mathbf{x}_{n}(i))
-g_{n}(\mathbf{z}_{n}(i))\right)\right]
\end{align}
based on $\mathbf{x}_{n}(i)$, $\left\{\mathbf{q}(\mathbf{x}_{l}(i)
+\nu_{nl}(i))\right\}_{l\in\Omega_{n}(i)}$, and $\mathbf{z}_{n}(i)$, which are all available to it at time $i$.  The sequence, $\left\{\mathbf{x}_{n}(i)\in\mathbb{R}^{M}\right\}_{i\geq 0}$, is the estimate (state) sequence generated at sensor $n$. The weight
sequence $\left\{\alpha(i)\right\}_{i\geq 0}$ satisfies the persistence
condition of Assumption~\textbf{(B.5)} and $\beta>0$ is chosen to
be an appropriate constant. Finally, $\mathcal{K}_{n}\in\mathbb{R}^{M\times\overline{M}}$ is an appropriately chosen matrix gain, possibly varying from sensor to sensor. Similar to \eqref{algRE1-b} the
above update can be written in compact form as
\begin{align}
\label{nalg2}
&
\mathbf{x}(i+1)=\mathbf{x}(i)-\alpha(i)\left[\beta(L(i)\otimes
I_{M})\mathbf{x}(i)+
\right.
\\
\nonumber
&
\hspace{.95cm}
\left.
\phantom{I_{M}}
+
\mathcal{K}\left(M(\mathbf{x}(i))-J(\mathbf{z}(i))\right)+\mathbf{\Upsilon}(i)+\mathbf{\Psi}(i)\right]
\end{align}
where $\mathbf{\Upsilon}(i),\mathbf{\Psi}(i)$ are as in
\eqref{defUpsilon}-\eqref{dith2} and
$\mathbf{x}(i)=[\mathbf{x}^{T}_{1}(i)\cdots\mathbf{x}^{T}_{N}(i)]^{T}$
is the vector of sensor states (estimates). The functions
$M(\mathbf{x}(i))$ and $J(\mathbf{z}(i))$ are given by
\begin{align}
\label{nlalg3}
M(\mathbf{x}(i))&=\left[h^{T}_{1}(\mathbf{x}_{1}(i))\cdots
h^{T}_{N}(\mathbf{x}_{N}(i))\right]^{T},
\\
\label{nlalg3-b}
J(\mathbf{z}(i))&=\left[g^{T}_{1}(\mathbf{z}_{1}(i))\cdots
g^{T}_{N}(\mathbf{z}_{N}(i))\right]^{T}
\end{align}
and $\mathcal{K}=\mbox{diag}(\mathcal{K}_{1},\cdots,\mathcal{K}_{N})$ is the block diagonal matrix of gains.

As an example, for the linear observation model, by defining $g_{n}(\mathbf{z}_{n}(i))$ to be $\overline{H}_{n}^{T}\mathbf{z}_{n}(i)$ (and choosing the matrix gains $\mathcal{K}_{n}$ to be $I_{M}$), the $\mathcal{NU}$ reduces to the $\mathcal{LU}$ updates~\eqref{algRE1}.

We note that the update scheme in \eqref{nalg2} is nonlinear
 and hence convergence properties can, in general, be characterized through the existence of appropriate stochastic Lyapunov
functions. In particular, if we can show that the iterative scheme
in \eqref{nalg2} falls under the purview of a general result
like Theorem~\ref{RM}, we can establish properties like
consistency, normality etc. To this end, we note, that
\eqref{nalg2} can be written as
\begin{align}
\nonumber
&
\mathbf{x}(i+1)
=
\mathbf{x}(i)-\alpha(i)\left[\beta\left(\overline{L}\otimes
I_{M}\right)\left(\mathbf{x}(i)
-\mathbf{1}_{N}\otimes\mathbf{\theta}^{\ast}\right)
\right.
\\
\nonumber
&
\hspace{-.15cm}
+
\left.
\beta\left(\widetilde{L}(i)\otimes
I_{M}\right)\mathbf{x}(i)+\mathcal{K}\left(M(\mathbf{x}(i))
-M(\mathbf{1}_{N}\otimes\mathbf{\theta}^{\ast})\right)\right.
\nonumber
\\
\label{nlalg4}
& \left.
-\mathcal{K}\left(J(\mathbf{z}(i))-M(\mathbf{1}_{N}\otimes\mathbf{\theta}^{\ast})\right)+\mathbf{\Upsilon}(i)+\mathbf{\Psi}(i)\right]
\end{align}
which becomes in the notation of Theorem~\ref{RM}
\begin{align}
\label{nalg5}
&
\hspace{-.2cm}
\mathbf{x}(i+1)
\hspace{-.1cm}
=
\hspace{-.1cm}
\mathbf{x}(i)
\hspace{-.1cm}
+
\hspace{-.1cm}
\alpha(i)
\hspace{-.1cm}
\left[R(\mathbf{x}(i))+\Gamma
\left(i+1,\mathbf{x}(i),\omega\right)\right]
\\
\label{nalg6}
&
\hspace{-.2cm}
R\left(\mathbf{x}\right)=-\left[\beta\left(\overline{L}\otimes
I_{M}\right)\left(\mathbf{x}-\mathbf{1}_{N}
\otimes\mathbf{\theta}^{\ast}\right)
\right.
\\
&
\nonumber
\left.
\hspace{2cm}
\phantom{\overline{L}}
+\mathcal{K}\left(M\left(\mathbf{x}\right)-M(\mathbf{1}_{N}\otimes\mathbf{\theta}^{\ast})\right)\right]
\\
\label{nalg7}
&
\hspace{-.15cm}
\Gamma\left(i+1,\mathbf{x},\omega\right)
=-\left[\beta\left(\widetilde{L}(i)\otimes
I_{M}\right)\mathbf{x}-
\right.
\\
\nonumber
&
\left.
\phantom{\widetilde{L}}
-
\mathcal{K}\left(J(\mathbf{z}(i))
-M(\mathbf{1}_{N}\otimes\mathbf{\theta}^{\ast})\right)+
\mathbf{\Upsilon}(i)+\mathbf{\Psi}(i)\right]
\end{align}
Consider the filtration, $\left\{\mathcal{F}_{i}\right\}_{i\geq 0}$,
\begin{equation}
\label{nalg8}
\hspace{-.3cm}
\mathcal{F}_{i}
\hspace{-.1cm}
=
\hspace{-.1cm}
\sigma
\hspace{-.1cm}
\left(\mathbf{x}(0),
\hspace{-.1cm}
\left\{
\hspace{-.1cm}
L(j),\left\{
\hspace{-.05cm}
\mathbf{z}_{n}(j)\right\}_{1\leq
N},
\hspace{-.1cm}
\mathbf{\Upsilon}(j),\mathbf{\Psi}(j)
\hspace{-.1cm}
\right\}_{0\leq j<i}
\hspace{-.05cm}
\right)
\hspace{-.225cm}
\end{equation}
Clearly, under~\textbf{(D.1)-(D.3)}, the
$\left\{\mathbf{x}(i)\right\}_{i\geq 0}$ generated by~$\mathcal{NU}$ is Markov w.r.t.
$\left\{\mathcal{F}_{i}\right\}_{i\geq 0}$, and the definition in
\eqref{nalg7} renders the random family,
$\left\{\Gamma\left(i+1,\mathbf{x},\omega\right)\right\}_{\mathbf{x}\in\mathbb{R}^{NM}}$, $\mathcal{F}_{i+1}$ measurable, zero-mean, and independent of
$\mathcal{F}_{i}$ for fixed $i+1$. Thus~\textbf{(B.1)-(B.2)} of Theorem~\ref{RM} are satisfied,
and we have the following.
\begin{proposition}[\hspace{-.05cm}$\mathcal{NU}$:Consistency/ asymp.~normality] 
\label{prop_nu} Let the sequence
$\left\{\mathbf{x}(i)\right\}_{i\geq 0}$ be generated by $\mathcal{NU}$. Let
$R\left(\mathbf{x}\right),\Gamma\left(i+1,\mathbf{x},\omega\right),\mathcal{F}_{i}$ be as in \eqref{nalg6}-\eqref{nalg7}.
 Then, if there exists a function $V\left(\mathbf{x}\right)$ satisfying~\textbf{(B.3)-(B.4)} at $\mathbf{x}^{\ast}=\mathbf{1}_{N}\otimes \mathbf{\theta}^{\ast}$, the estimate sequence $\left\{\mathbf{x}_{n}(i) \right\}_{i\geq 0 }$ at any sensor $n$ is  consistent. In other words,
\begin{equation}
\label{prop_nu1}
\mathbb{P}_{\mathbf{\theta}^{\ast}}[\lim_{i\rightarrow\infty}\mathbf{x}_{n}(i)=\mathbf{\theta}^{\ast},~\forall
n]=1
\end{equation}
If, in addition, \textbf{(C.1)-(C.4)} are satisfied, the sequence $\left\{\mathbf{x}_{n}(i)\right\}_{i\geq 0 }$ at any sensor~$n$ is asymptotically normal.
\end{proposition}
Proposition~\ref{prop_nu} states that, a.s.~asymptotically, the network reaches consensus, and the estimates at each sensor converge to the true value of the parameter vector~$\theta^{\star}$. The Proposition relates these convergence properties of~$\mathcal{NU}$ to the existence of suitable Lyapunov
functions. For a particular observation model characterized by
the corresponding functions $h_{n}(\cdot),g_{n}(\cdot)$, if one
can come up with an appropriate Lyapunov function satisfying the
assumptions of Proposition~\ref{prop_nu}, then consistency and asymptotic normality are guaranteed. Existence of a suitable Lyapunov condition is sufficient for consistency, but may not be necessary. In particular, there may be observation models for which the $\mathcal{NU}$ algorithm is consistent, but there exists no Lyapunov function satisfying the assumptions of Proposition~\ref{prop_nu}.\footnote{This is because converse
theorems in stability theory do not hold in general, see,~\cite{Krasovskii}.} Also, even if a suitable Lyapunov
function exists, it may be difficult to guess its form, because
there is no systematic (constructive) way of coming up with
Lyapunov functions for generic models.

However, for our problem of interest, some additional weak
assumptions on the observation model, for example, Lipschitz
continuity of the functions $h_{n}(\cdot)$, will guarantee the
existence of suitable Lyapunov functions, thus establishing
convergence properties of the $\mathcal{NU}$ algorithm. The rest
of this subsection studies this issue and presents different
sufficient conditions on the observation model, which guarantee
that the assumptions of Proposition~\ref{prop_nu} are satisfied,
leading to the a.s.~convergence of the~$\mathcal{NU}$ algorithm. For the development in the rest of the subsection, we assume that $\overline{M}=M$ in the decomposition~\eqref{nlprobform2} and $\mathcal{K}_{n}=I_{M}$ for all $n$. The extensions of Theorems~\ref{nlth}-\ref{thnl} to $\overline{M}\neq M$ and arbitrary gains $\mathcal{K}_{n}$ are immediate. We start with the following definition:
\begin{definition}[Consensus Subspace] We define the consensus
subspace, $\mathcal{C}\subset\mathbb{R}^{MN}$ as
\begin{equation}
\label{nalg9} \mathcal{C}=\left\{\mathbf{y}\in\mathbb{R}^{NM}\left|\rule{0cm}{.35cm}\right.\mathbf{y}=\mathbf{1}_{N}\otimes\widetilde{\mathbf{y}},
\,\widetilde{\mathbf{y}}\in\mathbb{R}^{M}\right\}
\end{equation}
\end{definition}
For $\mathbf{y}\in\mathbb{R}^{NM}$, we denote its
component in $\mathcal{C}$ by $\mathbf{y}_{\mathcal{C}}$ and its
orthogonal component by $\mathbf{y}_\mathcal{C}^{\perp}$.

\begin{theorem}[$\mathcal{NU}$:Consistency under Lipschitz on $h_n$]
\label{nlth} Let $\left\{\mathbf{x}(i)\right\}_{i\geq 0}$ be the state
sequence generated by the $\mathcal{NU}$ algorithm (Assumptions~\textbf{(D.1)-(D.3)}). Further,
$\forall \mathbf{\theta}\neq\widetilde{\mathbf{\theta}}\in\mathbb{R}^{M}, 1\leq n\leq N$, let the functions $h_{n}(\cdot)$ be:
\begin{enumerate}
\item Lipschitz continuous with
constants $k_{n}>0$, i.e.,
\begin{equation}
\label{nlth1}
\|h_{n}(\mathbf{\theta})-h_{n}(\widetilde{\mathbf{\theta}})\|\leq
k_{n}\|\mathbf{\theta}-\widetilde{\mathbf{\theta}}\|
\end{equation}
\item
\begin{equation}
\label{nlth2}
\left(\mathbf{\theta}-\widetilde{\mathbf{\theta}}\right)^{T}\left(h_{n}(\mathbf{\theta})-h_{n}(\widetilde{\theta})\right)\geq
0.
\end{equation}
\end{enumerate}
Define $K$ as
\begin{equation}
\label{nlth3} K=\max(k_{1},\cdots,k_{N})
\end{equation}
Then, for every $\beta>0$, the estimate sequence is consistent. In other words,
\begin{equation}
\label{nlth4}
\mathbb{P}_{\mathbf{\theta}^{\ast}}\left[\lim_{i\rightarrow\infty}\mathbf{x}_{n}(i)=\mathbf{\theta}^{\ast},~\forall
n\right]=1
\end{equation}
\end{theorem}
The proof is in Appendix~\ref{app:NU}. The conditions in
\eqref{nlth1}-\eqref{nlth2} are much easier to verify than  guessing a Lyapunov function.
Also, as will be shown in the proof, the conditions in
Theorem~\ref{nlth} determine a Lyapunov function explicitly, which
may be used to analyze properties like convergence rate. The
Lipschitz assumption is quite common in the stochastic
approximation literature, while the assumption in
\eqref{nlth2} holds for a large class of functions. As a
matter of fact, in the one-dimensional case ($M=1$), it is
satisfied if the functions $h_{n}(\cdot)$ are non-decreasing. Also, in general, it can be shown from the proof (Appendix~\ref{app:NU}) that the Lipschitz continuity
assumption in Theorem~\ref{nlth} may be replaced by continuity of
the functions $h_{n}(\cdot),~1\leq n\leq N$, and linear growth
conditions, i.e., for constants $c_{n,1},c_{n,2}>0$,
\begin{equation}
\label{nlth20}
\hspace{-.05cm}
\left\|h_{n}(\mathbf{\theta})\right\|^{2}
\hspace{-.1cm}
\leq
\hspace{-.1cm}
c_{n,1}
\hspace{-.1cm}
+
\hspace{-.1cm}
c_{n,2}\|\mathbf{\theta}\|^{2},
\forall\mathbf{\theta}\in\mathbb{R}^{M},1\leq n\leq N.
\end{equation}

We now present another set of sufficient conditions that
guarantee consistency of  $\mathcal{NU}$.
If the observation model is separably estimable, in some cases
even if the underlying model is nonlinear, it may be possible to
choose the functions, $g_{n}(\cdot)$, such that the function
$h(\cdot)$ possesses nice properties. This is the next result.
\begin{theorem}[$\mathcal{NU}$:Consistency--$h$ strict monotonicity]
\label{thnl} Consider the~$\mathcal{NU}$ algorithm (Assumptions~\textbf{(D.1)-(D.3)}). Suppose that the functions
$g_{n}(\cdot)$ can be chosen, such that the functions
$h_{n}(\cdot)$ are Lipschitz continuous with constants $k_{n}>0$
and the function $h(\cdot)$ satisfies
\begin{equation}
\label{thnl2}
\hspace{-.1cm}
\left(\mathbf{\theta}-\widetilde{\mathbf{\theta}}\right)^{T}
\hspace{-.1cm}
\left(h(\mathbf{\theta})-h(\widetilde{\mathbf{\theta}})\right)
\hspace{-.1cm}
\geq
\hspace{-.1cm}
\gamma\|\mathbf{\theta}-\widetilde{\mathbf{\theta}}\|^{2},
\:\forall\mathbf{\theta},\widetilde{\mathbf{\theta}}\in\mathbb{R}^{M}
\end{equation}
for some constant $\gamma>0$. Then, for
\[
K=\max(k_{1},\cdots,k_{N}),
\]
if
\[
\beta>\frac{K^{2}+K\gamma}{\gamma\lambda_{2}{\overline{L}}},
\]
 the
algorithm $\mathcal{NU}$ is  consistent, i.e.,
\begin{equation}
\label{thnl3}
\mathbb{P}_{\mathbf{\theta}^{\ast}}\left[\lim_{i\rightarrow\infty}\mathbf{x}_{n}(i)=\mathbf{\theta}^{\ast},~\forall
n\right]=1
\end{equation}
\end{theorem}
The proof is provided in Appendix~\ref{app:NU}. We comment that, in comparison to
Theorem~\ref{nlth}, strengthening the assumptions on $h(\cdot)$, see \eqref{thnl2}, considerably weakens the assumptions on the functions
$h_{n}(\cdot)$. Eqn.~(\ref{thnl2}) is an analog of strict monotonicity. For example, if~$h(\cdot)$ is linear, the left hand side of \eqref{thnl2} becomes a quadratic and the condition says that this quadratic is strictly away from zero, i.e., monotonically increasing with rate~$\gamma$.

\section{Nonlinear Observation Models: Algorithm $\mathcal{NLU}$}
\label{nlalg2} In this Section, we present the algorithm
$\mathcal{NLU}$ for distributed estimation in separably estimable
observation models. As explained later, this is a mixed
time-scale algorithm, where the consensus time-scale dominates the
observation update time-scale as time progresses. The
$\mathcal{NLU}$ algorithm is based on the fact that, for separably
estimable models, it suffices to know $h(\mathbf{\theta}^{\ast})$,
because $\mathbf{\theta}^{\ast}$ can be unambiguously determined
from the invertible function $h(\mathbf{\theta}^{\ast})$. To be
precise, if the function $h(\cdot)$ has a continuous inverse, then
any iterative scheme converging to $h(\mathbf{\theta}^{\ast})$
will lead to  consistent estimates, obtained by inverting the
sequence of iterates. The algorithm $\mathcal{NLU}$ is shown to
yield  consistent and unbiased estimators at each sensor for any
separably observable model, under the assumption that the function
$h(\cdot)$ has a continuous inverse. Thus, the algorithm
$\mathcal{NLU}$ presents a more reliable alternative than the
algorithm $\mathcal{NU}$, because, as shown in
Section~\ref{nlalg1}, the convergence properties of the latter
can be guaranteed only under certain assumptions on the
observation model.
We briefly comment on the organization of this section. The
$\mathcal{NLU}$ algorithm for separably estimable observation
models is presented in Section~\ref{algNLUform}.
Section~\ref{NLU_disres} offers interpretations of the
$\mathcal{NLU}$ algorithm and presents the main results regarding
consistency, mean-square convergence, asymptotic unbiasedness
proved in the paper. In Section~\ref{proofs_NLU} we prove the
main results about the $\mathcal{NLU}$ algorithm and provide
insights behind the analysis (in particular, why standard
stochastic approximation results cannot be used directly to give
its convergence properties.) Finally, Section~\ref{dis_future}
presents discussions on the $\mathcal{NLU}$ algorithm and suggests
future research directions.
\subsection{Algorithm $\mathcal{NLU}$}
\label{algNLUform}
\textbf{Algorithm $\mathcal{NLU}$}: Let
$\mathbf{x}(0)=[\mathbf{x}_{1}^{T}\cdots\mathbf{x}^{T}_{N}]^{T}$
be the initial set of states (estimates) at the sensors. The
$\mathcal{NLU}$ generates the sequence
$\left\{\mathbf{x}_{n}(i)\right\}\in\mathbb{R}^{M}$ at the $n$-th sensor according
to the  distributed recursive scheme:
\begin{align}
\label{nlu1}
&
\mathbf{x}_{n}(i+1)=h^{-1}\left(h(\mathbf{x}_{n}(i))-
\phantom{\sum_{l\in\Omega_{n}(i)}}
\right.
\\
\nonumber
&
-
\left.
\beta(i)\left(\sum_{l\in\Omega_{n}(i)}\left[h(\mathbf{x}_{n}(i))
-\mathbf{q}\left(h(\mathbf{x}_{l}(i))+\nu_{nl}(i)\right)\right]\right)
\right.
\\
\nonumber
&
\left.
\phantom{\sum_{l\in\Omega_{n}(i)}}
-\alpha(i)\left[h(\mathbf{x}_{n}(i))-g_{n}(\mathbf{z}_{n}(i))\right]\right)
\end{align}
based on the information,
\[
\mathbf{x}_{n}(i),\left\{\mathbf{q}\left(h(\mathbf{x}_{l}(i))
+\nu_{nl}(i)\right)\right\}_{l\in\Omega_{n}(i)},\mathbf{z}_{n}(i),
\]
available to it at time $i$ (we assume that at time $i$ sensor $l$
sends a quantized version of $h(\mathbf{x}_{l}(i))+\nu_{nl}(i)$ to
sensor $n$.) Here $h^{-1}(\cdot)$ denotes the inverse of the
function $h(\cdot)$ and $\left\{\beta(i)\right\}_{i\geq
0},\left\{\alpha(i)\right\}_{i\geq 0}$ are appropriately chosen weight
sequences. In the sequel, we analyze the $\mathcal{NLU}$ algorithm
under the model Assumptions~\textbf{(D.1)-(D.3)}, and in addition we
assume:

\begin{itemize}[\setlabelwidth{D.1)}]

\item{\textbf{(D.4)}}: There exists $\epsilon_{1}>0$, such that, $\forall\mathbf{\theta}\in\mathcal{U}$,  the
following moment exists:
\begin{equation}
\label{nlu00}
\mathbb{E}_{\mathbf{\theta}}\left[\left\|J(\mathbf{z}(i))-PJ(\mathbf{z}(i))\right\|^{2+\epsilon_{1}}\right]
=\kappa(\mathbf{\theta})
<\infty,
\end{equation}
where $J(\mathbf{z}(i))$ is defined in~\eqref{nlalg3-b} and the matrix $P$ is given by
\begin{align}
\label{nlu10000}
P
&
=\frac{1}{N}\left(\mathbf{1}_{N}\otimes
I_{\overline{M}}\right)\left(\mathbf{1}_{N}\otimes
I_{\overline{M}}\right)^{T}
\\
\nonumber
&
=\frac{1}{N}\left(\mathbf{1}_{N}\mathbf{1}_{N}^{T}\right)\otimes I_{\overline{M}}
\\
\nonumber
&
=
P_{N}\otimes I_{\overline{M}}
\end{align}
The above moment condition is slightly stronger than the moment assumption required by the $\mathcal{NU}$ algorithm in \eqref{nlprobform4}, where only existence of the quadratic moment of the random variables $g_{n}(\mathbf{z}_{n}(i))$ was assumed. For example, for the models considered in Propositions~\ref{prop-sigadd}-\ref{prop-sigadd4}, the above condition coincides with the existence of slightly higher than quadratic moments of the observations $\mathbf{z}_{n}(i)$. The latter is clearly justified for any reasonable observation noise distribution.

We also define, $\forall\mathbf{\theta}\in\mathcal{U}$:
\begin{align}
\label{nlu0001}
\mathbb{E}_{\mathbf{\theta}}
\left[\left\|J(\mathbf{z}(i))-PJ(\mathbf{z}(i))\right\|\right]
&
=
\kappa_{1}(\mathbf{\theta})
<\infty
\\
\label{nlu00001}
\mathbb{E}_{\mathbf{\theta}}\left[\left\|J(\mathbf{z}(i))-
PJ(\mathbf{z}(i))\right\|^{2}\right]
&
=
\kappa_{2}(\mathbf{\theta})
<\infty.
\end{align}

\item{\textbf{(D.5)}}: The weight sequences $\left\{\alpha(i)\right\}_{i\geq
0}$, and $\left\{\beta(i)\right\}_{i\geq 0}$ are given by
\begin{equation}
\label{nlu2}
\alpha(i)=\frac{a}{(i+1)^{\tau_{1}}},~~\beta(i)=\frac{b}{(i+1)^{\tau_{2}}}
\end{equation}
where $a,b>0$ are constants. We assume the following:
\begin{equation}
\label{nlu01} .5<\tau_{1},\tau_{2}\leq
1,~~\tau_{1}>\frac{1}{2+\epsilon_{1}}+\tau_{2},~~2\tau_{2}>\tau_{1}
\end{equation}
We note that, under Assumption~\textbf{(D.4)}, $\epsilon_{1}>0$,
such weight sequences always exist. As an example, if
$\frac{1}{2+\epsilon_{1}}=.49$, then the choice $\tau_{1}=1$ and
$\tau_{2}=.505$ satisfies the inequalities in
\eqref{nlu01}.
%
%
\end{itemize}
To write the $\mathcal{NLU}$ in a more compact form,  introduce
the \emph{transformed} state sequence,
$\left\{\widetilde{\mathbf{x}}(i)\right\}_{i\geq 0}$, where
$\widetilde{\mathbf{x}}(i)=[\widetilde{\mathbf{x}}^{T}_{1}(i)\cdots\widetilde{\mathbf{x}}^{T}_{N}(i)]^{T}\in\mathbb{R}^{N\overline{M}}$ and the iterations are
\begin{align}
\label{nlu3}
\hspace{-.2cm}
\widetilde{\mathbf{x}}(i+1)
\hspace{-.05cm}
&=\hspace{-.05cm}
\widetilde{\mathbf{x}}(i)-\beta(i)\left(L(i)\otimes
I_{M}\right)\widetilde{\mathbf{x}}(i)
-
\\
\nonumber
\hspace{-.05cm}&
-
\hspace{-.05cm}
\alpha(i)\left[\widetilde{\mathbf{x}}(i)-J(\mathbf{z}(i))\right]
-\beta(i)\left(\mathbf{\Upsilon}(i)+\mathbf{\Psi}(i)\right)
\\
\label{nlu4}
\mathbf{x}(i)
\hspace{-.05cm}
&
=
\hspace{-.05cm}
\left[\left(h^{-1}(\widetilde{\mathbf{x}}_{1}(i))\right)^{T}
\cdots\left(h^{-1}(\widetilde{\mathbf{x}}_{N}(i))\right)^{T}\right]^{T}
\end{align}
Here $\mathbf{\Upsilon}(i),\mathbf{\Psi}(i)\in\mathbb{R}^{N\overline{M}}$ model the dithered
quantization error effects as in algorithm $\mathcal{NU}$ resulting from the quantized transmissions in~\eqref{nlu1}. The update model in \eqref{nlu3} is a mixed time-scale procedure, where the consensus time-scale is determined by the weight sequence $\left\{\beta(i)\right\}_{i\geq 0}$. On the
other hand, the observation update time-scale is governed by the
weight sequence $\left\{\alpha(i)\right\}_{i\geq 0}$. It follows from
Assumption~\textbf{(D.5)} that $\tau_{1}>\tau_{2}$, which in turn
implies, $\frac{\beta(i)}{\alpha(i)}\rightarrow\infty$ as
$i\rightarrow\infty$. Thus, the consensus time-scale dominates the
observation update time-scale as the algorithm progresses making
it a mixed time-scale algorithm that does not directly fall under
the purview of stochastic approximation results like
Theorem~\ref{RM}. Also, the presence of the random link failures
and quantization noise (which operate at the same time-scale as
the consensus update) precludes standard approaches like
time-scale separation for the limiting system.
\begin{remark}
\label{NLU-imp} We comment on the distributed implementation of the $\mathcal{NLU}$.  Based on~\eqref{nlu1} and~\eqref{nlu3}-\eqref{nlu4}, the $\mathcal{NLU}$ may be implemented either in the estimate domain with $\{\mathbf{x}(i)\}$ as the algorithm state sequence or in the transformed domain with $\{\widetilde{\mathbf{x}}(i)\}$ as the state sequence. The implementation in the estimate domain,~\eqref{nlu1}, would require the sensors to store and transmit the instantaneous estimate $\mathbf{x}_{n}(i)$, however, implementation of the update involves computation of the functions $h(\mathbf{x}_{n}(i))$ followed by an inverse ($h^{-1}(\cdot)$) at every step. On the other hand, the implementation in the transformed domain, \eqref{nlu3}-\eqref{nlu4}, requires the sensors to store and transmit the states $\widetilde{\mathbf{x}}_{n}(i)$. The advantage in the latter implementation form is that the transformed state update rule,~\eqref{nlu3}, is linear in the state $\widetilde{\mathbf{x}}(i)$ and, in particular, does not require function computations and inverses at each step. (We note that the inverse in~\eqref{nlu4} may not be implemented at all iterations and does not affect the propagation of the transformed state sequence. In fact,~\eqref{nlu4} may be implemented only once to obtain the actual estimates from the transformed state sequence when the latter converge based on a suitable stopping criterion.) Hence, in practice, to simplify computations, the $\mathcal{NLU}$ may be implemented in the transformed domain with $\widetilde{\mathbf{x}}_{n}(i)$ being the state at a sensor $n$.
\end{remark}

\subsection{Algorithm $\mathcal{NLU}$: Discussions and Main
Results}\label{NLU_disres} We comment on the $\mathcal{NLU}$
algorithm. As is clear from \eqref{nlu3}-\eqref{nlu4}, the
$\mathcal{NLU}$ algorithm operates in a \emph{transformed} domain.
As a matter of fact, the function $h(\cdot)$ (c.f.
definition~\ref{sep_est_def}) can be viewed as an invertible
transformation on the parameter space $\mathcal{U}$. The
transformed state sequence, $\{\widetilde{\mathbf{x}}(i)\}_{i\geq
0}$, is then a transformation of the estimate sequence
$\{\mathbf{x}(i)\}_{i\geq 0}$, and, as seen from \eqref{nlu3},
the evolution of the sequence
$\{\widetilde{\mathbf{x}}(i)\}_{i\geq 0}$ is linear. This is an
important feature of the $\mathcal{NLU}$ algorithm, which is
linear in the transformed domain, although the underlying
observation model is nonlinear. Intuitively, this approach can be
thought of as a distributed stochastic version of homomorphic
filtering (see~\cite{Oppenheim}), where, by suitably transforming
the state space, linear filtering is performed on a certain
non-linear problem of filtering. In our case, for models of the
separably estimable type, the function $h(\cdot)$ then plays the
role of the analogous transformation in homomorphic filtering, and,
in this transformed space, one can design linear estimation
algorithms with desirable properties. This makes the
$\mathcal{NLU}$ algorithm significantly different from algorithm
$\mathcal{NU}$, with the latter operating on the untransformed
space and is non-linear. This linear property of the
$\mathcal{NLU}$ algorithm in the transformed domain leads to nice
statistical properties (for example, consistency asymptotic
unbiasedness) under much weaker assumptions on the observation
model than required by the nonlinear $\mathcal{NU}$ algorithm, but not asymptotic normality.

We now state the main results about the $\mathcal{NLU}$ algorithm
developed in the paper. We show that, if the observation
model is separably estimable, then, in the transformed domain, the
$\mathcal{NLU}$ algorithm is consistent. More specifically, if
$\mathbf{\theta}^{\ast}$ is the true (but unknown) parameter
value, then the transformed sequence
$\{\widetilde{\mathbf{x}}(i)\}_{i\geq 0}$ converges a.s. and in
mean-squared sense to $h(\mathbf{\theta}^{\ast})$. We note that,
unlike the $\mathcal{NU}$ algorithm, this only requires the
observation model to be separably estimable and no other
conditions on the functions $h_{n}(\cdot),h(\cdot)$. We summarize
these in the following theorem.
\begin{theorem}
\label{theorem_tilde}Consider the $\mathcal{NLU}$ algorithm under
the Assumptions~\textbf{(D.1)-(D.5)}, and the sequence
$\left\{\widetilde{\mathbf{x}}(i)\right\}_{i\geq 0}$ generated
according to \eqref{nlu3}. We then have
\begin{align}
\label{theorem_tilde1}
&
\hspace{-.35cm}
\mathbb{P}_{\mathbf{\theta}^{\ast}}\left[\lim_{i\rightarrow\infty} \widetilde{\mathbf{x}}_{n}(i)=
h(\mathbf{\theta}^{\ast}),\: \forall 1\leq n\leq N\right]=1
\\
\label{theorem_tilde2}
&
\hspace{-.4cm}
\lim_{i\rightarrow\infty}\mathbb{E}_{\mathbf{\theta}^{\ast}}
\left[\left\|\widetilde{\mathbf{x}}_{n}(i)-
h(\mathbf{\theta}^{\ast})\right\|^{2}\right]
=
0,\:
\forall 1\leq n\leq N
\end{align}
In particular,
\begin{equation}
\label{thorem_tilde3}
\lim_{i\rightarrow\infty}\mathbb{E}_{\mathbf{\theta}^{\ast}}
\left[\widetilde{\mathbf{x}}_{n}(i)\right]=h(\mathbf{\theta}^{\ast}),~~\forall
1\leq n\leq N
\end{equation}
In other words, in the transformed domain, the estimate sequence
$\{\widetilde{\mathbf{x}}_{n}(i)\}_{i\geq 0}$ at sensor $n$, is
consistent, asymptotically unbiased and converges in mean-squared
sense to $h(\mathbf{\theta}^{\ast})$.
\end{theorem}
As an immediate consequence of Theorem~\ref{theorem_tilde}, we
have the following result, which characterizes the statistical
properties of the untransformed state sequence
$\{\mathbf{x}(i)\}_{i\geq 0}$.
\begin{theorem}
\label{theorem_untrans}Consider the $\mathcal{NLU}$ algorithm
under the Assumptions~\textbf{(D.1)-(D.5)}. Let
$\left\{\mathbf{x}(i)\right\}_{i\geq 0}$ be the state sequence
generated, as given by \eqref{nlu3}-\eqref{nlu4}. We then have
\begin{equation}
\label{theorem_untrans1}
\mathbb{P}_{\mathbf{\theta}^{\ast}}\left[\lim_{i\rightarrow\infty}\mathbf{x}_{n}(i)=\mathbf{\theta}^{\ast},~\forall~1\leq
n\leq N\right]=1
\end{equation}
In other words, the $\mathcal{NLU}$ algorithm is consistent.

If in addition, the function $h^{-1}(\cdot)$ is Lipschitz
continuous, the $\mathcal{NLU}$ algorithm is asymptotically
unbiased, i.e.,
\begin{equation}
\label{theorem_untrans2}
\lim_{i\rightarrow\infty}\mathbb{E}_{\mathbf{\theta}^{\ast}}
\left[\mathbf{x}_{n}(i)\right]=\mathbf{\theta}^{\ast},~\forall~1\leq
n\leq N
\end{equation}
\end{theorem}
The next subsection is concerned with the proofs of
Theorems~\ref{theorem_tilde}, \ref{theorem_untrans}.
\subsection{Consistency and Asymptotic Unbiasedness of $\mathcal{NLU}$: Proofs of Theorems~\ref{theorem_tilde},\ref{theorem_untrans}}
\label{proofs_NLU} The present subsection is devoted to proving
the consistency and unbiasedness of the $\mathcal{NLU}$ algorithm
under the stated Assumptions. The proof is lengthy and we start by
explaining why standard stochastic approximation results like
Theorem~\ref{RM} do not apply directly. A careful inspection
shows that there are essentially two different time-scales
embedded in \eqref{nlu3}. The consensus time-scale is
determined by the weight sequence $\left\{\beta(i)\right\}_{i\geq
0}$, whereas the observation update time-scale is governed by the
weight sequence $\left\{\alpha(i)\right\}_{i\geq 0}$. It follows
from Assumption~\textbf{(D.5)} that $\tau_{1}>\tau_{2}$, which, in
turn, implies $\frac{\beta(i)}{\alpha(i)}\rightarrow\infty$ as
$i\rightarrow\infty$. Thus, the consensus time-scale dominates the
observation update time-scale as the algorithm progresses making
it a mixed time-scale algorithm that does not directly fall under
the purview of stochastic approximation results like
Theorem~\ref{RM}. Also, the presence of the random link failures
and quantization noise (which operate at the same time-scale as
the consensus update) precludes standard approaches like
time-scale separation for the limiting system.

Finally, we note that standard stochastic approximation assumes
that the state evolution follows a stable deterministic system
perturbed by \emph{zero-mean} stochastic noise. More specifically,
if $\{\mathbf{y}(i)\}_{i\geq 0}$ is the sequence of interest,
Theorem~\ref{RM} assumes that $\{\mathbf{y}(i)\}_{i\geq 0}$
evolves as
\begin{equation}
\label{proofs_NLU1}
\hspace{-.05cm}
\mathbf{y}(i+1)
\hspace{-.075cm}
=
\hspace{-.075cm}
\mathbf{y}(i)+\gamma(i)
\hspace{-.075cm}
\left[R(\mathbf{y}(i))
\hspace{-.05cm}
+
\hspace{-.05cm}
\Gamma(i+1,\omega,\mathbf{y}(i))\right]
\end{equation}
where $\{\gamma(i)\}_{i\geq 0}$ is the weight sequence,
$\Gamma(i+1,\omega,\mathbf{y}(i))$ is the \emph{zero-mean} noise.
If the sequence $\{\mathbf{y}(i)\}_{i\geq 0}$ is supposed to
converge to $\mathbf{y}_{0}$, it further assumes that
$R(\mathbf{y}_{0})=\mathbf{0}$ and $\mathbf{y}_{0}$ is a stable
equilibrium of the deterministic system
\begin{equation}
\label{proofs_NLU2}
\mathbf{y}_{d}(i+1)=\mathbf{y}_{d}(i)+\gamma(i)R(\mathbf{y}_{d}(i))
\end{equation}
The $\mathcal{NU}$ algorithm (and its linear version,
$\mathcal{LU}$) falls under the purview of this, and we can
establish convergence properties using standard stochastic
approximation (see Sections~\ref{lbm},\ref{nlprobform}.) However,
the $\mathcal{NLU}$ algorithm cannot be represented in the form of
\eqref{proofs_NLU1}, even ignoring the presence of multiple
time-scales. Indeed, as established by
Theorem~\ref{theorem_tilde}, the sequence
$\{\widetilde{\mathbf{x}}(i)\}_{i\geq 0}$ is supposed to converge
to $\mathbf{1}_{N}\otimes h(\mathbf{\theta}^{\ast})$ a.s. and
hence writing \eqref{nlu3} as a stochastically perturbed
system around $\mathbf{1}_{N}\otimes h(\mathbf{\theta}^{\ast})$ we
have
\begin{equation}
\label{proofs_NLU3}
\widetilde{\mathbf{x}}(i+1)=\widetilde{\mathbf{x}}(i)+\gamma(i)\left[R(\widetilde{\mathbf{x}}(i))+\Gamma(i+1,\omega,\widetilde{\mathbf{x}}(i))\right]
\end{equation}
where,
\begin{align*}
&
R(\widetilde{\mathbf{x}}(i))
=
-\beta(i)\left(\overline{L}\otimes
I_{M}\right)\left(\widetilde{\mathbf{x}}(i)-\mathbf{1}_{N}\otimes
h(\mathbf{\theta}^{\ast})\right)-
\\
&
\hspace{3cm}
-
\alpha(i)\left(\widetilde{\mathbf{x}}(i)-\mathbf{1}_{N}\otimes
h(\mathbf{\theta}^{\ast})\right)
\\
\nonumber
&
\Gamma(i+1,\omega,\widetilde{\mathbf{x}}(i))
=
-\beta(i)\left(\widetilde{L}(i)\otimes
I_{M}\right)\left(\widetilde{\mathbf{x}}(i)-
\right.
\\
&
\left.
\hspace{1.5cm}
-
\mathbf{1}_{N}\otimes
h(\mathbf{\theta}^{\ast})\right)-
\beta(i)\left(\mathbf{\Upsilon}(i)+
\mathbf{\Psi}(i)\right)+
\\
&
\hspace{1.5cm}
+
\alpha(i)\left(J(\mathbf{z}(i))-
\mathbf{1}_{N}\otimes
h(\mathbf{\theta}^{\ast})\right)
\end{align*}
Although, $R(\mathbf{1}_{N}\otimes
h(\mathbf{\theta}^{\ast}))=\mathbf{0}$ in the above decomposition,
the noise $\Gamma(i+1,\omega,\widetilde{\mathbf{x}}(i))$ is not
unbiased as the term $\left(J(\mathbf{z}(i))-\mathbf{1}_{N}\otimes
h(\mathbf{\theta}^{\ast})\right)$ is \emph{not} zero-mean.

With the above discussion in mind, we proceed to the proof of
Theorems~\ref{theorem_tilde},\ref{theorem_untrans}, which we
develop in stages. The detailed proofs of the intermediate results
are provided in Appendix~\ref{proof_theorem_tilde}.

In parallel to the evolution of the state sequence
$\left\{\mathbf{x}(i)\right\}_{i\geq 0}$, we consider the
following update of the auxiliary sequence,
$\left\{\widetilde{\mathbf{x}}^{\circ}(i)\right\}_{i\geq 0}$:
\begin{align}
\label{nlu04}
&\widetilde{\mathbf{x}}^{\circ}(i+1)=\widetilde{\mathbf{x}}^{\circ}(i)
-\beta(i)\left(\overline{L}\otimes
I_{M}\right)\widetilde{\mathbf{x}}^{\circ}(i)-
\\
\nonumber
&
\hspace{3.5cm}
-
\alpha(i)\left[
\widetilde{\mathbf{x}}^{\circ}(i)-J(\mathbf{z}(i))\right]
\end{align}
with
$\widetilde{\mathbf{x}}^{\circ}(0)=\widetilde{\mathbf{x}}(0)$.
Note that in~(\ref{nlu04}) the random Laplacian~$L$ is replaced by
the average Laplacian~$\overline{L}$ and the quantization
noises~$\mathbf{\Upsilon}(i)$ and~$\mathbf{\Psi}(i)$ are not
included. In other words, in the absence of link failures and
quantization, the recursion~(\ref{nlu3}) reduces to~(\ref{nlu04}),
i.e., the sequences
$\left\{\widetilde{\mathbf{x}}(i)\right\}_{i\geq 0}$ and
$\left\{\widetilde{\mathbf{x}}^{\circ}(i)\right\}_{i\geq 0}$ are
the same.

Now consider the sequence whose recursion adds as input to the
recursion in~(\ref{nlu04}) the quantization
noises~$\mathbf{\Upsilon}(i)$ and~$\mathbf{\Psi}(i)$. In other
words, in the absence of link failures, but with quantization
included, define similarly the sequence
$\left\{\widehat{\mathbf{x}}(i)\right\}_{i\geq 0}$ given by
\begin{align}
\label{hatcons}
&\widehat{\mathbf{x}}(i+1)=\widehat{\mathbf{x}}(i)
-\beta(i)\left(\overline{L}\otimes
I_{M}\right)\widehat{\mathbf{x}}(i)
-
\\
&
\nonumber
\hspace{1cm}
-
\alpha(i)\left[\widehat{\mathbf{x}}(i)-J(\mathbf{z}(i))\right]
-\beta(i)\left(\mathbf{\Upsilon}(i)+\mathbf{\Psi}(i)\right)
\end{align}
with $\widehat{\mathbf{x}}(0)=\widetilde{\mathbf{x}}(0)$. Like
before, the recursions~(\ref{nlu3},\ref{nlu4}) will reduce
to~(\ref{hatcons}) when there are no link failures. However,
notice that in~(\ref{hatcons}) the quantization noise sequences
$\mathbf{\Upsilon}(i)$ and $\mathbf{\Psi}(i)$ are the sequences
resulting from quantizing $\widetilde{\mathbf{x}}(i)$
in~(\ref{nlu3}) and not from quantizing $\widehat{\mathbf{x}}(i)$
in~(\ref{hatcons}).

Define the instantaneous averages over the network as
\begin{align}
\nonumber
\mathbf{x}_{\mbox{\scriptsize{avg}}}(i)
&=
\frac{1}{N}\sum_{n=1}^{N}\mathbf{x}_{n}(i)=\frac{1}{N}\left(\mathbf{1}_{N}\otimes
I_{M}\right)^{T}\mathbf{x}(i)
\\
\nonumber
\widetilde{\mathbf{x}}_{\mbox{\scriptsize{avg}}}(i)
&=
\frac{1}{N}\sum_{n=1}^{N}\widetilde{\mathbf{x}}_{n}(i)=\frac{1}{N}\left(\mathbf{1}_{N}\otimes
I_{M}\right)^{T}\widetilde{\mathbf{x}}(i)
\\
\label{nlu6}
\mathbf{x}^{\circ}_{\mbox{\scriptsize{avg}}}(i)
&=
\frac{1}{N}\sum_{n=1}^{N}\mathbf{x}^{\circ}_{n}(i)=\frac{1}{N}\left(\mathbf{1}_{N}\otimes
I_{M}\right)^{T}\mathbf{x}^{\circ}(i)
\\
\nonumber
\widetilde{\mathbf{x}}^{\circ}_{\mbox{\scriptsize{avg}}}(i)
&=
\frac{1}{N}\sum_{n=1}^{N}\widetilde{\mathbf{x}}^{\circ}_{n}(i)=\frac{1}{N}\left(\mathbf{1}_{N}\otimes
I_{M}\right)^{T}\widetilde{\mathbf{x}}^{\circ}(i)
\end{align}
We sketch the main steps of the proof here. While proving
consistency and mean-squared sense convergence, we first show that
the average sequence,
$\left\{\widetilde{\mathbf{x}}^{\circ}_{\mbox{\scriptsize{avg}}}(i)
\right\}_{i\geq 0}$, converges a.s. to
$h(\mathbf{\theta}^{\ast})$. This can be done by invoking standard
stochastic approximation arguments. Then we show that the sequence
$\left\{\widetilde{\mathbf{x}}^{\circ}(i)\right\}_{i\geq 0}$
reaches consensus a.s., and clearly the limiting consensus value
must be $h(\mathbf{\theta}^{\ast})$. Intuitively, the a.s.
consensus comes from the fact that, after a sufficiently large
number of iterations, the consensus effect dominates over the
observation update effect, thus asymptotically leading to
consensus. The final step in the proof uses a series of comparison
arguments to show that the sequence
$\left\{\widetilde{\mathbf{x}}(i)\right\}_{i\geq 0}$ also reaches
consensus a.s. with $h(\mathbf{\theta}^{\ast})$ as the limiting
consensus value.

We now detail the proofs of
Theorems~\ref{theorem_tilde},\ref{theorem_untrans} in the
following steps.
\begin{itemize}[\setlabelwidth{I}]
\item{\textbf{(I)}}: The first step consists of studying the
convergence properties of the sequence
$\left\{\widetilde{\mathbf{x}}^{\circ}_{\mbox{\scriptsize{avg}}}(i)
\right\}_{i\geq 0}$, see \eqref{nlu04}, for which we
establish the following result.
\begin{lemma}
\label{circcons} Consider the sequence,
$\left\{\widetilde{\mathbf{x}}^{\circ}(i)\right\}_{i\geq 0}$,
given by \eqref{nlu04}, under the
Assumptions~\textbf{(D.1)-(D.5)}. Then,
\begin{eqnarray}
\label{circcons1}
\mathbb{P}_{\mathbf{\theta}^{\ast}}\left[\lim_{i\rightarrow\infty}\widetilde{\mathbf{x}}^{\circ}(i)=\mathbf{1}_{N}\otimes
h(\mathbf{\theta}^{\ast})\right]&=&1
\\
\label{circcons200}
\lim_{i\rightarrow\infty}\mathbb{E}_{\mathbf{\theta}^{\ast}}
\left[\left\|\widetilde{\mathbf{x}}^{\circ}(i)-\mathbf{1}_{N}\otimes
h(\mathbf{\theta}^{\ast})\right\|^{2}\right]&=&0
\end{eqnarray}
\end{lemma}
Lemma~\ref{circcons} says that the sequence
$\left\{\widetilde{\mathbf{x}}^{\circ}(i)\right\}_{i\geq 0}$
converges a.s. and in $\mathcal{L}_{2}$ to $\mathbf{1}_{N}\otimes
h(\mathbf{\theta}^{\ast})$. For proving Lemma~\ref{circcons} we
first consider the corresponding average sequence
$\{\widetilde{\mathbf{x}}^{\circ}_{\mbox{\scriptsize{avg}}}(i)\}_{i\geq
0}$, see \eqref{nlu6}. For the sequence
$\{\widetilde{\mathbf{x}}^{\circ}_{\mbox{\scriptsize{avg}}}(i)\}_{i\geq
0}$, we can invoke stochastic approximation algorithms to prove
that it converges a.s. and in $\mathcal{L}_{2}$ to
$h(\mathbf{\theta}^{\ast})$. This is carried out in
Lemma~\ref{nlulm}, which we state now.
\begin{lemma}
\label{nlulm} Consider the sequence,
$\left\{\widetilde{\mathbf{x}}^{\circ}_{\mbox{\scriptsize{avg}}}(i)\right\}_{i\geq
0}$, given by \eqref{nlu6}, under the
Assumptions~\textbf{(D.1)-(D.5)}. Then,
\begin{eqnarray}
\label{nlulm1}
\mathbb{P}_{\mathbf{\theta}^{\ast}}\left[\lim_{i\rightarrow\infty} \widetilde{\mathbf{x}}^{\circ}_{\mbox{\scriptsize{avg}}}(i)= h(\mathbf{\theta}^{\ast})\right]&=&1
\\
\label{nlulm2}
\lim_{i\rightarrow\infty}\mathbb{E}_{\mathbf{\theta}^{\ast}}
\left[\left\|\widetilde{\mathbf{x}}^{\circ}_{\mbox{\scriptsize{avg}}}(i)
-h(\mathbf{\theta}^{\ast})\right\|^{2}\right]&=&0
\end{eqnarray}
\end{lemma}
%
The arguments in Lemmas~\ref{nlulm},\ref{circcons} and subsequent
results require the following property of real number sequences,
which we state here (see Appendix~\ref{proof_prop_alpha} for proof.)
\begin{lemma}
\label{prop_alpha} Let the sequences
$\left\{r_{1}(i)\right\}_{i\geq 0}$ and
$\left\{r_{2}(i)\right\}_{i\geq 0}$ be given by
\begin{equation}
\label{prop_alpha01}
r_{1}(i)=\frac{a_{1}}{(i+1)^{\delta_{1}}},~~r_{2}(i)=\frac{a_{2}}{(i+1)^{\delta_{2}}}
\end{equation}
where $a_{1},a_{2},\delta_{2}\geq 0$ and $0\leq\delta_{1}\leq 1$.
Then, if $\delta_{1}=\delta_{2}$, there exists $B>0$, such that,
for sufficiently large non-negative integers, $j<i$,
\begin{equation}
\label{prop_alpha1} 0\leq \sum_{k=j}^{i-1}\left[\left(
\prod_{l=k+1}^{i-1}\left(1-r_{1}(l)\right)\right)r_{2}(k)\right]\leq
B
\end{equation}
Moreover, the constant $B$ can be chosen independently of $i,j$.
Also, if $\delta_{1}<\delta_{2}$, then, for arbitrary fixed $j$,
\begin{equation}
\label{prop_alpha2} \lim_{i\rightarrow\infty}\sum_{k=j}^{i-1}
\left[\left(\prod_{l=k+1}^{i-1}
\left(1-r_{1}(l)\right)\right)r_{2}(k)\right]=0
\end{equation}

%
(We use the convention that,
$\prod_{l=k+1}^{i-1}\left(1-r_{l}\right)=1$, for $k=i-1$.)
\end{lemma}
We note that Lemma~\ref{prop_alpha} essentially studies stability
of time-varying deterministic scalar recursions of the form:
\begin{equation}
\label{step1.1} y(i+1)=r_{1}(i)y(i)+r_{2}(i)
\end{equation}
where $\{y(i)\}_{i\geq 0}$ is a scalar sequence evolving according
to \eqref{step1.1} with $y(0)=0$, and the sequences
$\left\{r_{1}(i)\right\}_{i\geq 0}$ and
$\left\{r_{2}(i)\right\}_{i\geq 0}$ are given by
\eqref{prop_alpha01}.\\

\item{\textbf{(II)}}: In this step, we study the convergence
properties of the sequence $\left\{\widehat{\mathbf{x}}(i)
\right\}_{i\geq 0}$, see \eqref{hatcons}, for which we
establish the following result.
\begin{lemma}
\label{hatconslemma} Consider the sequence
$\left\{\widehat{\mathbf{x}}(i)\right\}_{i\geq 0}$ given by
\eqref{hatcons} under the Assumptions~\textbf{(D.1)-(D.5)}. We
have
\begin{eqnarray}
\label{hatcons1}
\mathbb{P}_{\mathbf{\theta}^{\ast}}\left[\lim_{i\rightarrow\infty}\widehat{\mathbf{x}}(i)=\mathbf{1}_{N}\otimes
h(\mathbf{\theta}^{\ast})\right]&=&1
\\
\label{hatcons200}
\lim_{i\rightarrow\infty}\mathbb{E}_{\mathbf{\theta}^{\ast}}
\left[\left\|\widehat{\mathbf{x}}(i)-\mathbf{1}_{N}\otimes
h(\mathbf{\theta}^{\ast})\right\|^{2}\right]&=&0
\end{eqnarray}
\end{lemma}
The proof of Lemma~\ref{hatconslemma} is given in Appendix~\ref{proof_hatconslemma},
and mainly consists of a comparison argument involving the
sequences
$\left\{\widetilde{\mathbf{x}}^{\circ}_{\mbox{\scriptsize{avg}}}(i)
\right\}_{i\geq 0}$ and $\left\{\widehat{\mathbf{x}}(i)
\right\}_{i\geq 0}$.\\
\item{\textbf{(III)}}: This is the final step in the proofs of
Theorems~\ref{theorem_tilde},\ref{theorem_untrans}. The proof of
Theorem~\ref{theorem_tilde} consists of a comparison argument
between the sequences $\left\{\widehat{\mathbf{x}}(i)
\right\}_{i\geq 0}$ and $\left\{\widetilde{\mathbf{x}}(i)
\right\}_{i\geq 0}$, which is detailed in Appendix~\ref{proof_theorem_tilde}. The proof of Theorem~\ref{theorem_untrans}, also detailed in Appendix~\ref{proof_theorem_tilde}, is a consequence of Theorem~\ref{theorem_tilde} and the Assumptions.
\end{itemize}

\subsection{Application: Distributed Static Phase Estimation in Smart Grids}
\label{NLU-power} In this subsection we show that our development of the $\mathcal{NLU}$ for separably estimable observation models leads to a completely distributed solution of the static phase estimation problem in smart grids. We briefly review the application scenario in the following, for a more complete treatment of the classical problem of static phase estimation in power grids the reader is referred to one of the many existing excellent textbooks, for example,~\cite{Ilic-Zaborsky}. For our purpose, we may assume the power grid to be a physical network of $N$ generators and loads (hereafter called nodes), interconnected through transmission lines. The physical grid may then be modeled as a network $G_{p}=(V,E_{p})$, where\footnote{Note that the physical connections $E_{p}$ and the inter-node communication links $E$ (to be used by the distributed information processing algorithms) are, in general, different. This is a common feature of cyberphysical architectures, where a sensor network is instrumented on top of an existing physical infrastructure. In the following, we will assume that each physical node is equipped with a sensor for information processing, although the inter-sensor communication topology may be different from the physical inter-node connections. Also, the terms nodes and sensors will be used synonymously in the sequel.} $E_{p}$ denotes the set of transmission lines or interconnections. The physical state of a node $n$ consists of the pair $(\mathcal{V}_{n},\theta_{n})$, denoting the voltage magnitude and the phase angle respectively. The real power flowing through the transmission line connecting nodes $n$ and $l$ is then given by (see~\cite{Thorp})
\begin{equation}
\label{power-flow}
\hspace{-.025cm}
\mathcal{P}_{nl}
\hspace{-.075cm}
=
\hspace{-.075cm}
\mathcal{V}_{n}^{2}a_{nl}
\hspace{-.075cm}
-
\hspace{-.075cm}
\mathcal{V}_{n}\mathcal{V}_{l}a_{nl}\cos(\theta_{nl})+\mathcal{V}_{n}\mathcal{V}_{l}b_{nl}\sin(\theta_{nl})
\end{equation}
where $a_{nl}+jb_{nl}$ is the complex line admittance and $\theta_{nl}=\theta_{n}-\theta_{l}$. In view of the physical network structure, the following assumptions on the physical grid are supposed to hold:
\begin{itemize}
\item[\textbf{(P.1)}] The physical grid, represented by the graph $(V,E_{p})$, is connected. This is reasonable, as the physical power grid is often studied with aggregated models that have dense interconnections, see, for example, the benchmark IEEE 30 bus and IEEE 118 bus systems (\cite{Thorp}).
\item[\textbf{(P.2)}] The real and imaginary parts, $a_{nl},b_{nl}$, of the line admittance are taken to be zero, if no direct physical connection (link) exists between the nodes $n$ and $l$. Similarly, if nodes $n$ and $l$ are connected by a transmission line, we assume both the components $a_{nl},b_{nl}$, of the line admittance to be non-zero. In particular, from~\eqref{power-flow}, the real power flow $\mathcal{P}_{nl}$ between nodes $n$ and $l$ is non-zero \emph{iff} there exists a physical transmission line connecting the nodes. Also, $\mathcal{V}_{n}\neq 0$ for all $n$.
\end{itemize}
In the following, we will assume that the node voltages are known constants and the unknown parameter of interest is the vector of node phases $\mathbf{\theta}=[\theta_{1},\cdots,\theta_{N}]^{T}$. This is justified by the common assumption of phase-voltage decoupling in power grids, where the voltage magnitude is generally seen to fluctuate at a much slower time-scale than the phase (see~\cite{Ilic-Zaborsky}). Also, to keep the exposition simple, we assume that the node phase differences are small, i.e., $\cos(\theta_{nl})\approx 1$, commonly used in the steady state grid operating regime (\cite{Ilic-Zaborsky}). With these simplifications, the real power flow in a transmission line connecting nodes $n$ and $l$ is approximately given by
\begin{equation}
\label{power-flow1}
\mathcal{P}_{nl}(\mathbf{\theta})=\mathcal{V}_{n}\mathcal{V}_{l}b_{nl}\sin(\theta_{nl})
\end{equation}
The goal of centralized static phase estimation is to estimate the unknown vector $\mathbf{\theta}$ of phases by using line flow data. Since, only relative quantities (phase differences) are involved in this problem, it is customary to assume (see~\cite{Thorp}) that one of the nodes is a slack (or reference) bus, whose phase is a known constant. W.l.o.g. we assume that node $N$ is the slack bus in our system, whose phase angle $\theta_{N}$ is a known constant. Hence, the effective parameter vector is $\mathbf{\theta}=[\theta_{1},\cdots,\theta_{N-1}]^{T}\in\mathbb{R}^{N-1}$. We now provide conditions for the distributed observation model~\eqref{power-flow4} to be separably estimable. We show that our $\mathcal{NLU}$ algorithm can be used to obtain a distributed solution to this problem, leading to a consistent estimate of $\mathbf{\theta}$ at each sensor. To setup the distributed observation model, let $E_{m}\subset E_{p}$ denote the set of physical transmission lines equipped with power flow measuring devices (usually some form of relays, see~\cite{Thorp}.) The successive power flow measurements, $\{z_{nl}(i)\}$ at the physical line $(n,l)$ (assuming $(n,l)\in E_{m}$) are then noisy versions of the power flow $\mathcal{P}_{nl}$, i.e.,
\begin{align}
\label{power-flow2}
z_{nl}(i)
&=\mathcal{P}_{nl}(\mathbf{\theta})+\zeta_{nl}(i)
\\
\nonumber
&=\mathcal{V}_{n}\mathcal{V}_{l}b_{nl}\sin(\theta_{nl})+\zeta_{nl}(i)
\end{align}
where, $\{\zeta_{nl}(i)\}$ is the zero mean i.i.d. measurement noise. For distributed information processing, we assume that the measurement sequence $\{z_{nl}(i)\}$ is forwarded to one of the adjacent nodes $n$ or $l$. As will be seen, the particular choice of $n$ or $l$ is not important, as long as it stays constant for all $i$. In general, denoting by $\Omega^{m}_{n}$ to be the set of physical neighbors of node $n$ w.r.t. the physical graph $(V,E_{m})$, let $\mathcal{M}_{n}\subset\Omega^{m}_{n}$ denote the set such that $l\in\mathcal{M}_{n}$ if the line flow data $\{z_{nl}(i)\}$ is forwarded to node $n$. The observation sequence, $\{\mathbf{z}_{n}(i)\}$, is then
\begin{equation}
\label{power-flow3}
\mathbf{z}_{n}(i)=\{z_{nl}(i),~\forall~l\in\mathcal{M}_{n}\}
\end{equation}
By~\eqref{power-flow2}-\eqref{power-flow3} and the development in Section~\ref{nonlin-obs-mod}, the above distributed observation model corresponds to the signal in additive noise type. Following the notation in Section~\ref{nonlin-obs-mod}, the observation process $\{\mathbf{z}_{n}(i)\}$ may be written as
\begin{equation}
\label{power-flow4}
\mathbf{z}_{n}(i)=f_{n}(\mathbf{\theta})+\mathbf{\zeta}_{n}(i)
\end{equation}
where $f_{n}:\mathcal{U}\longmapsto\mathbb{R}^{|\mathcal{M}_{n}|}$ and $\mathbf{\zeta}_{n}(i)$ are given by
\begin{align}
\label{power-flow5}
f_{n}(\mathbf{\theta})&=[\mathcal{V}_{n}\mathcal{V}_{l}b_{nl}
\sin(\theta_{nl}),~l\in\mathcal{M}_{n}]^{T}
\\
\nonumber
\mathbf{\zeta}_{n}(i)&=[\zeta_{nl},~l\in\mathcal{M}_{n}]^{T}.
\end{align}
We now provide conditions for the distributed observation model~\eqref{power-flow4} to be separably estimable.
\begin{proposition}
\label{prop-power-flow} Depending on the phase $\theta_{N}$ of the reference bus $N$ (as to whether $\theta_{N}\in [0,\pi/2)$ or $\theta_{N}\in [-\pi/4,\pi/4)$), let the parameter domain $\mathcal{U}\subset\mathbb{R}^{N-1}$ be either $\times_{n=1}^{N-1}[0,\pi/2)$ or $\times_{n=1}^{N-1}[-\pi/4,\pi/4)$. Define, $\forall\mathbf{\theta}\in\mathcal{U}$,  $f:\mathcal{U}\longmapsto\mathbb{R}^{\sum_{n=1}^{N}|\mathcal{M}_{n}|}$ by
\begin{equation}
\label{power-flow7}
f(\mathbf{\theta})
\hspace{-.1cm}
=
\hspace{-.1cm}
[f_{1}^{T}(\mathbf{\theta})\cdots f_{N}^{T}(\mathbf{\theta})]^{T}
\hspace{-.2cm}
=
\hspace{-.1cm}
[\mathcal{P}_{nl}(\mathbf{\theta}),(n,l)
\hspace{-.05cm}
\in
\hspace{-.05cm}
E_{m}]^{T}
\hspace{-.1cm}
.
\hspace{-.05cm}
\end{equation}
Then, if the graph $(V,E_{m})$ (with $E_{m}$ as the edge set) is connected and assumption~\textbf{(P.2)} holds, $f(\cdot)$ is invertible on $\mathcal{U}$ and the observation model~\eqref{power-flow4} is separably estimable.
\end{proposition}
Before proceeding to the proof we comment on Proposition~\ref{prop-power-flow}. The observation model in~\eqref{power-flow4} is of the signal in additive noise type and hence, by the development in Section~\ref{nonlin-obs-mod}, the invertibility of $f(\cdot)$ is equivalent to separable estimability and necessary for the consistency or observability of the centralized estimator also (see the text following Proposition~\ref{prop-sigadd}.) Proposition~\ref{prop-power-flow} shows that the invertibility of $f(\cdot)$ holds if the graph formed by the physical transmission links equipped with power flow measuring devices is connected.
\begin{proof}
The proof is based on a simple inductive argument. First, we note that the continuity of $f(\cdot)$ on $\mathcal{U}$ holds trivially. To establish the invertibility of $f(\cdot)$ on $\mathcal{U}$, it suffices to show that we can uniquely recover the value of $\mathbf{\theta}\in\mathcal{U}$ given the  value $f(\mathbf{\theta})$. To this end, assume that $f(\mathbf{\theta})$ is given. Recall the form of $f(\mathbf{\theta})$, as in~\eqref{power-flow7}. Since, $\theta_{N}$ is known, given $f(\mathbf{\theta})$, the components $\theta_{n}$, $n\in\Omega^{m}_{N}$ may be uniquely determined. Indeed, knowing $f(\mathbf{\theta})$ amounts to knowing the values of the quantities $\mathcal{P}_{n,N}$, $n\in\Omega^{m}_{N}$. Hence, under assumption~\textbf{(P.2)}, and the fact that $\mathbf{\theta}\in\mathcal{U}$, we can uniquely determine $\theta_{n}$, $n\in\Omega^{m}_{N}$. To continue the induction, define $\mathcal{J}_{1}\subset V$ by $\mathcal{J}_{1}=\Omega^{m}_{N}$. Once the components $\theta_{n}$, $n\in\mathcal{J}_{1}$ are known, by using similar reasoning, the components $\theta_{l}$, $l\in\Omega^{m}_{n}$ for each $n\in\mathcal{J}_{1}$ may be uniquely determined. Hence, in the second step, the set of known components $\mathcal{J}_{2}\subset V$ is
\[
\mathcal{J}_{2}=\{l\in\Omega_{n}^{m}~|~n\in\mathcal{J}_{1}\}
\]
Continuing the same recursion, the set of components known at the $k$-th step, $\mathcal{J}_{k}$ is given by
\[
\mathcal{J}_{k}=\{l\in\Omega_{n}^{m}~|~n\in\mathcal{J}_{k-1}\}
\]
 Note that $\mathcal{J}_{1}\subset\mathcal{J}_{2}\subset\cdots \subset\mathcal{J}_{k}\subset\cdots$. However, the number of nodes is finite; hence, the sets $\mathcal{J}_{k}$ cannot increase forever. Due to the connectivity of the graph $(V,E_{m})$ in a finite number of steps $k_{0}$ (at most equal to the diameter of the graph), the process will converge with $\mathcal{J}_{k_{0}}=V$. Hence, all components of $\mathbf{\theta}$ will be uniquely determined, establishing the invertibility of $f(\cdot)$.
\end{proof}
The following result demonstrates the applicability of the $\mathcal{NLU}$ to distributed consistent phase estimation in power grids. It follows from Theorem~\ref{theorem_untrans}.
\begin{theorem}
\label{th:power-flow}
Consider the power grid described above and let the observation process $\{\mathbf{z}_{n}(i)\}$ at the $n$-th node (sensor) be given by~\eqref{power-flow4}. Let the measurement noise process $\{\mathbf{\zeta}(i)\}$ satisfy the moment conditions~\textbf{(D.4)} and the physical grid conditions in the hypothesis of Proposition~\ref{prop-power-flow} hold. Suppose the nodes are instrumented with a communication architecture satisfying assumptions~\textbf{(D.2)-(D.3)}. Let $\{\mathbf{x}_{n}(i)\}$ be the estimate sequence of the vector $\mathbf{\theta}^{\ast}$ of phase angles generated at node $n$ by an instantiation of the $\mathcal{NLU}$ distributed parameter estimation algorithm under assumption~\textbf{(D.5)}. Then,
\begin{equation}
\label{power-flow8}
\mathbb{P}_{\mathbf{\theta}^{\ast}}\left(\mathbf{x}_{n}(i)=\mathbf{\theta}^{\ast},~~\forall n\right)=1,~~\forall~\mathbf{\theta}^{\ast}\in\mathcal{U}
\end{equation}
\end{theorem}

\section{Conclusion}
\label{conclusion}
%
\label{dis_future}
This paper studies linear and nonlinear \emph{distributed} (vector) parameter estimation problems as may arise in constrained sensor networks. Our problem statement is quite general, including communication among sensors that is quantized, noisy, and with channels that fail at random times. These are characteristic of packet communication in wireless sensor networks. We introduce a generic observability condition, the separable estimability condition, that generalizes to distributed estimation the general observability condition of centralized parameter estimation. We study three recursive distributed estimators, $\mathcal{ALU}$, $\mathcal{NU}$, and $\mathcal{NLU}$. We study their asymptotic properties, namely: consistency, asymptotic unbiasedness, and for the $\mathcal{ALU}$ and $\mathcal{NU}$ algorithms their asymptotic normality. The $\mathcal{NLU}$ works in a transformed domain where the recursion is actually linear, and a final nonlinear transformation, justified by the separable estimability condition, recovers the parameter estimate (a stochastic generalization of homeomorphic filtering.)  For example, Theorem~\ref{theorem_tilde} shows that, in the
transformed domain, the $\mathcal{NLU}$ leads to consistent and
asymptotically unbiased estimators at every sensor for all
separably estimable observation models satisfying~\textbf{(D.4)}\footnote{The $\mathcal{NLU}$ requires a slightly stronger moment condition, Assumption~\textbf{(D.4)}. However, for reasonable observation noise statistics arising in practice, this assumption is justified.}. Since, the function
$h(\cdot)$ is invertible, for practical purposes, a knowledge of
$h(\mathbf{\theta}^{\ast})$ is sufficient for knowing
$\mathbf{\theta}^{\ast}$. In that respect, the algorithm
$\mathcal{NLU}$ is much more applicable than the algorithm
$\mathcal{NU}$, which requires further assumptions on the
observation model for the existence of consistent and
asymptotically unbiased estimators. However, in case, the
algorithm $\mathcal{NU}$ is applicable, it provides convergence
rate guarantees (for example, asymptotic normality) that follow
from standard stochastic approximation theory. On the other hand,
the algorithm $\mathcal{NLU}$ does not fall under the purview of
standard stochastic approximation theory (see
Section~\ref{proofs_NLU}) and hence does not inherit these
convergence rate properties. In this paper, we presented a
convergence theory of the three algorithms
 under broad conditions. An interesting
future research direction is to establish a convergence rate
theory for the $\mathcal{NLU}$ algorithm (and in general,
distributed stochastic algorithms of this form, which involve
mixed time-scale behavior and biased perturbations.)
\appendices

\renewcommand{\baselinestretch}{1.3}

\section{Some Results on Stochastic Approximation}
\label{res_stoch_app}

{\small

We present some classical results on stochastic
approximation from~\cite{Nevelson} regarding the convergence
properties of generic stochastic recursive procedures, which will
be used to characterize the convergence properties (consistency, convergence rate) of the $\mathcal{LU}$ algorithm.
\begin{theorem}
\label{RM} Let $\left\{\mathbf{x}(i)\in\mathbb{R}^{l}\right\}_{i\geq
0}$ be a random  sequence:
\begin{equation}
\label{RM:1}
\mathbf{x}(i+1)=\mathbf{x}(i)+\alpha(i)\left[R(\mathbf{x}(i))+\Gamma
\left(i+1,\mathbf{x}(i),\omega\right)\right]
\end{equation}
where, $R(\cdot):\mathbb{R}^{l}\longmapsto\mathbb{R}^{l}$ is Borel measurable and
$\left\{\Gamma(i,\mathbf{x},\omega)\right\}_{i\geq
0,~\mathbf{x}\in\mathbb{R}^{l}}$ is a family of random
vectors in $\mathbb{R}^{l}$, defined on a probability
space $(\Omega,\mathcal{F},\mathcal{P})$, and $\omega\in\Omega$ is
a canonical element. Let the following sets of
assumptions hold:
\begin{itemize}[\setlabelwidth{B.1)}]
\item{\textbf{(B.1)}}: The function $\Gamma(i,\cdot,
\cdot):\mathbb{R}^{l}\times\Omega\longrightarrow\mathbb{R}^{l}$ is
$\mathcal{B}^{l}\otimes\mathcal{F}$
measurable for every $i$; $\mathcal{B}^{l}$ is the Borel algebra of $\mathbb{R}^{l}$.
\item{\textbf{(B.2)}}: There exists a filtration
$\left\{\mathcal{F}_{i}\right\}_{i\geq 0}$ of $\mathcal{F}$, such that, for
each $i$, the family of random vectors
$\left\{\Gamma\left(i,\mathbf{x},\omega\right)
\right\}_{\mathbf{x}\in\mathbb{R}^{l}}$ is $\mathcal{F}_{i}$ measurable, zero-mean and independent of
$\mathcal{F}_{i-1}$.

(If Assumptions~\textbf{(B.1)-(B.2)} hold, $\left\{\mathbf{x}(i)\right\}_{i\geq 0}$, is Markov.)
\item{\textbf{(B.3)}}: There exists a function
$V\left(\mathbf{x}\right)\in\mathbb{C}_{2}$ with bounded second order partial derivatives and a point $\mathbf{x}^{\ast}\in\mathbb{R}^{l}$ satisfying:
\begin{align*}
&V\left(\mathbf{x}^{\ast}\right)=0,
\:\:V\left(\mathbf{x}\right)>0,\,\mathbf{x}\neq\mathbf{x}^{\ast},
\:\:\lim_{\left\|\mathbf{x}\right\|\rightarrow\infty}V\left(\mathbf{x}\right)=\infty,
\\
&
\sup_{\epsilon<\left\|\mathbf{x}-\mathbf{x}^{\ast}\right\|<
\frac{1}{\epsilon}}\left(R\left(\mathbf{x}\right),
V_{\mathbf{x}}\left(\mathbf{x}\right)\right)<0,\:\forall\epsilon>0
\end{align*}
where $V_{\mathbf{x}}\left(\mathbf{x}\right)$ denotes the gradient (vector) of $V(\cdot)$ at $\mathbf{x}$.

\item{\textbf{(B.4)}}: There exist constants $k_{1},k_{2}>0$, such
that,
\begin{align*}
&
\hspace{-.6cm}
\left\|R\left(\mathbf{x}\right)\right\|^{2}+\mathbb{E}\left[\left\|
\Gamma\left(i+1,\mathbf{x},\omega\right)\right\|^{2}\right]
\leq
k_{1}\left(1+V\left(\mathbf{x}\right)\right)-
\\
&
\hspace{3.5cm}
-
k_{2}\left(R\left(\mathbf{x}\right),
V_{\mathbf{x}}\left(\mathbf{x}\right)\right)
\end{align*}
\item{\textbf{(B.5)}}: The weight sequence $\left\{\alpha(i)\right\}_{i\geq
0}$ satisfies
\begin{equation}
\label{alphacond-b} \alpha(i)>0,\:\sum_{i\geq 0}\alpha_{i}=\infty, \: \sum_{i\geq 0}\alpha^{2}(i)<\infty
\end{equation}
\end{itemize}
\begin{itemize}[\setlabelwidth{C.1)}]
\item{\textbf{(C.1)}}: The function $R\left(\mathbf{x}\right)$ admits the
representation
\begin{equation}
\label{assc1.1}
R\left(\mathbf{x}\right)=B\left(\mathbf{x}-\mathbf{x}^{\ast}\right)
+\delta\left(\mathbf{x}\right)
\end{equation}
where
\begin{equation}
\label{assc1.2}
\lim_{\mathbf{x}\rightarrow\mathbf{x}^{\ast}}
\frac{\left\|\delta\left(\mathbf{x}\right)\right\|}{
\left\|\mathbf{x}-\mathbf{x}^{\ast}\right\|}=0
\end{equation}
(Note, in particular, if $\delta\left(\mathbf{x}\right)\equiv 0$, then
\eqref{assc1.2} is satisfied.)
\item{\textbf{(C.2)}}: The weight sequence, $\left\{\alpha(i)\right\}_{i\geq
0}$ is of the form,
\begin{equation}
\label{assc2.1} \alpha(i)=\frac{a}{i+1},~~\forall i\geq 0
\end{equation}
where $a>0$ is a constant (note that \textbf{(C.2)} implies
\textbf{(B.5)}).
\item{\textbf{(C.3)}}: Let $I$ be the $l\times l$ identity matrix and $a,B$ as in~\eqref{assc2.1} and~\eqref{assc1.1}, respectively. Then, the matrix $\Sigma= aB+\frac{1}{2}I$ is stable.
%
\item{\textbf{(C.4)}}: The entries of the matrices, $\forall
i\geq 0, x\in\mathbf{R}^{l}$,
\[
A\left(i,\mathbf{x}\right)=\mathbb{E}\left[\Gamma\left(i+1,\mathbf{x},\omega\right)
\Gamma^{T}\left(i+1,\mathbf{x},\omega\right)\right],
\]
are finite, and the following limit exists:
\[
\lim_{i\rightarrow\infty,~\mathbf{x}\rightarrow\mathbf{x}^{\ast}}
A\left(i,\mathbf{x}\right)
=S_{0}
\]
\item{\textbf{(C.5)}}: There exists $\epsilon>0$, such that
\begin{equation}
\label{assc5.1}
\hspace{-.1cm}
\lim_{R\rightarrow\infty}\sup_{\left\|\mathbf{x}-\mathbf{x}^{\ast}\right\|
<\epsilon}\sup_{i\geq 0} \int_{\left\|\Gamma\left(i+1,\mathbf{x},\omega\right)\right\|>R}
\hspace{-1.4cm}
\left\|\Gamma\left(i+1,\mathbf{x},\omega\right)\right\|^{2}dP=0
\end{equation}
\end{itemize}
Then we have the following:

Let Assumptions~\textbf{(B.1)-(B.5)} hold for
$\left\{\mathbf{x}(i)\right\}_{i\geq 0}$ in \eqref{RM:1}. Then,
starting from an arbitrary initial state, the Markov process,
$\left\{\mathbf{x}(i)\right\}_{i\geq 0}$, converges a.s. to
$\mathbf{x}^{\ast}$. In other words,
\begin{equation}
\label{RMres1}
\mathbb{P}\left[\lim_{i\rightarrow\infty}\mathbf{x}(i)=\mathbf{x}^{\ast}\right]=1
\end{equation}
The normalized process, $\left\{\sqrt{i}\left(\mathbf{x}(i)- \mathbf{x}^{\ast}\right)\right\}_{i\geq 0}$, is asymptotically normal if, besides Assumptions~\textbf{(B.1)-(B.5)}, Assumptions~\textbf{(C.1)-(C.5)} are also satisfied. In particular, as $i\rightarrow\infty$
\begin{equation}
\label{RMres2}
\sqrt{i}\left(\mathbf{x}(i)-\mathbf{x}^{\ast}\right)\Longrightarrow\, \mathcal{N}(\mathbf{0},S)
\end{equation}
where $\Longrightarrow$ denotes convergence in distribution
(weak convergence.) Also, the asymptotic variance, $S$, in
\eqref{RMres2} is
\begin{equation}
\label{RMres3} S=a^{2}\int_{0}^{\infty}e^{\Sigma v}S_{0}e^{\Sigma^{T}v}dv
\end{equation}
%
\end{theorem}
\begin{proof}
For a proof see~\cite{Nevelson} (c.f. Theorems 4.4.4, 6.6.1).
\end{proof}
}
\vspace*{-1.0cm}
\section{Proofs of Theorems~\ref{nlth},\ref{thnl}}
\label{app:NU}

\textbf{Proof of Theorem~\ref{nlth}}
{\small
\begin{proof}
 We noted the recursive scheme~\eqref{nalg2} satisfies Assumptions~\textbf{(B.1)-(B.2)} of
Theorem~\ref{RM}. To prove consistency, we verify
Assumptions~\textbf{(B.3)-(B.4)}. Let the Lyapunov function
\begin{equation}
\label{nlth5}
V\left(\mathbf{x}\right)=\|\mathbf{x}-\mathbf{1}_{N}\otimes\mathbf{\theta}^{\ast}\|^{2}
\end{equation}
Clearly,
\[
V(\mathbf{1}_{N}\otimes\mathbf{\theta}^{\ast})=0,V\left(\mathbf{x}\right)>0,
\mathbf{x}\neq\mathbf{1}_{N}\otimes\mathbf{\theta}^{\ast},
\hspace{-.2cm}
\lim_{\|\mathbf{x}\|\rightarrow\infty}V\left(\mathbf{x}\right)=\infty
\]
By Assumptions~\eqref{nlth1}-\eqref{nlth2},
$h(\cdot)$ is Lipschitz continuous and
\begin{equation}
\label{nlth7}
\left(\mathbf{\theta}-\widetilde{\mathbf{\theta}}\right)^{T}\left(h(\mathbf{\theta})-h(\widetilde{\theta})\right)>
0,~~\forall~\mathbf{\theta}\neq\widetilde{\mathbf{\theta}}\in\mathbb{R}^{M}
\end{equation}
where \eqref{nlth7} follows from the invertibility of
$h(\cdot)$ and
\begin{equation}
\label{nlth8}
h\left(\mathbf{\theta}\right)=\frac{1}{N}\sum_{n=1}^{N}h_{n}
\left(\mathbf{\theta}\right),~~\forall~\mathbf{\theta}\in\mathbb{R}^{M}
\end{equation}
Recall $R\left(\mathbf{x}\right), \Gamma\left(i+1,\mathbf{x},\omega\right)$ in
\eqref{nalg6}-\eqref{nalg7}. Then
\begin{align}
\label{nlth9}
&\left(R\left(\mathbf{x}\right),V_{\mathbf{x}}\left(\mathbf{x}\right)\right)
=
\\
&
\nonumber
\hspace{1.5cm}
-2\beta\left(\mathbf{x}-\mathbf{1}_{N}\otimes\mathbf{\theta}^{\ast}\right)^{T}
\left(\overline{L}\otimes
I_{M}\right)
\left(\mathbf{x}-\mathbf{1}_{N}
\otimes\mathbf{\theta}^{\ast}\right)
\\
&
\nonumber
\hspace{1.5cm}
-2\left(\mathbf{x}-\mathbf{1}_{N}\otimes\mathbf{\theta}^{\ast}\right)^{T}
\left[M\left(\mathbf{x}\right)-M(\mathbf{1}_{N}
\otimes\mathbf{\theta}^{\ast})\right]\nonumber
\\
\nonumber
&
=
 -2\beta\left(\mathbf{x}-\mathbf{1}_{N}
\otimes\mathbf{\theta}^{\ast}\right)^{T}\left(\overline{L}
\otimes I_{M}\right)\left(\mathbf{x}-\mathbf{1}_{N}
\otimes\mathbf{\theta}^{\ast}\right)
-
\\
\nonumber
&
\hspace{1.5cm}
-
2\sum_{n=1}^{N}\left[\left(\mathbf{x}_{n}
-\mathbf{\theta}^{\ast}\right)^{T}
\left(h_{n}(\mathbf{x}_{n})-h_{n}(\mathbf{\theta}^{\ast})\right)\right]
\leq
0
\end{align}
where the last step follows from the positive-semidefiniteness of
$\overline{L}\otimes I_{M}$ and \eqref{nlth2}. To verify
Assumption~\textbf{(B.3)}, show
\begin{equation}
\label{nlth10}
\sup_{\epsilon<\|\mathbf{x}-\mathbf{1}_{N}\mathbf{\theta}^{\ast}\|<\frac{1}{\epsilon}}\left(R\left(\mathbf{x}\right),V_{\mathbf{x}}\left(\mathbf{x}\right)\right)<0,~~\forall\epsilon>0
\end{equation}
Assume, on the contrary, \eqref{nlth10} not
satisfied. Then from \eqref{nlth9}
\begin{equation}
\label{nlth12}
\sup_{\epsilon<\|\mathbf{x}-\mathbf{1}_{N}\mathbf{\theta}^{\ast}\|<\frac{1}{\epsilon}}\left(R\left(\mathbf{x}\right),V_{\mathbf{x}}\left(\mathbf{x}\right)\right)=
0,~~\forall\epsilon>0
\end{equation}
Then, there exists a sequence, $\left\{\mathbf{x}^{k}\right\}_{k\geq 0}$ in
$\left\{\mathbf{x}\in\mathbb{R}^{NM}\left|\rule{0cm}{.35cm}\right.\epsilon<\|\mathbf{x}-\mathbf{1}_{N}\mathbf{\theta}^{\ast}\|<\frac{1}{\epsilon}\right\}$,
such that
\begin{equation}
\label{nlth13}
\lim_{k\rightarrow\infty}\left(R(\mathbf{x}^{k}),V_{\mathbf{x}}(\mathbf{x}^{k})\right)=0
\end{equation}
Since $\left\{\mathbf{x}\in\mathbb{R}^{NM}\left|\rule{0cm}{.25cm}\right.\epsilon<\|\mathbf{x}-\mathbf{1}_{N}\mathbf{\theta}^{\ast}\|<\frac{1}{\epsilon}\right\}$
is relatively compact, $\left\{\mathbf{x}^{k}\right\}_{k\geq 0}$
has a limit point, $\widehat{\mathbf{x}}$, such that
$\epsilon\leq\|\widetilde{\mathbf{x}}-\mathbf{1}_{N}\mathbf{\theta}^{\ast}\|\leq\frac{1}{\epsilon}$,
and, by continuity of
$\left(R\left(\mathbf{x}\right),V_{\mathbf{x}}\left(\mathbf{x}\right)\right)$:
\begin{equation}
\label{nlth14}
\left(R(\widetilde{\mathbf{x}}),V_{\mathbf{x}}(\widetilde{\mathbf{x}})\right)=0
\end{equation}
From \eqref{nlth2} and \eqref{nlth9}, we then have
\begin{align}
\label{nlth15}
\left(\widetilde{\mathbf{x}}-\mathbf{1}_{N}
\otimes\mathbf{\theta}^{\ast}\right)^{T}\left(\overline{L}\otimes
I_{M}\right)\left(\widetilde{\mathbf{x}}-\mathbf{1}_{N}
\otimes\mathbf{\theta}^{\ast}\right)
&=
0,
\\
\label{nlth15-b}
\hspace{-1cm}
\left(\widetilde{\mathbf{x}}_{n}-\mathbf{\theta}^{\ast}\right)^{T}
\left(h_{n}(\widetilde{\mathbf{x}}_{n})-h_{n}(\mathbf{\theta}^{\ast})\right)
&
=
0,\:\forall n
\end{align}
The equality in \eqref{nlth15} and the properties of
the Laplacian imply that $\widetilde{\mathbf{x}}\in\mathcal{C}$
and hence there exists $\mathbf{a}\in\mathbb{R}^{M}$, such
that,
\begin{equation}
\label{nlth16} \widetilde{\mathbf{x}}_{n} = \mathbf{a},~~\forall n
\end{equation}
The  inequalities in \eqref{nlth15-b} then imply
\begin{equation}
\label{nlth17}
\left(\mathbf{a}-\mathbf{\theta}^{\ast}\right)^{T}\left(h(\mathbf{a})-h(\mathbf{\theta}^{\ast})\right)=0
\end{equation}
which is a contradiction by \eqref{nlth7} since
$\mathbf{a}\neq\mathbf{\theta}^{\ast}$. Thus, we have
\eqref{nlth10} that verifies Assumption~\textbf{(B.3)}.
Finally, we note that,
\begin{align}
\label{nlth19}
&
\|R\left(\mathbf{x}\right)\|^{2}=
\left\|\beta\left(\overline{L}\otimes
I_{M}\right)\left(\mathbf{x}-\mathbf{1}_{N}
\otimes\mathbf{\theta}^{\ast}\right)+\right.
\\
&
\nonumber
\left.
\hspace{3cm}
\phantom{\overline{L}}
+
\left(M\left(\mathbf{x}\right)-M(\mathbf{1}_{N}
\otimes\mathbf{\theta}^{\ast})\right)\right\|^{2}
\\
\nonumber
&
\leq
4\beta^{2}\left\|\left(\overline{L}\otimes I_{M}\right)\left(\mathbf{x}-\mathbf{1}_{N}
\otimes\mathbf{\theta}^{\ast}\right)\right\|^{2}+
\\
&
\nonumber
\hspace{3cm}
+
4\left\|M\left(\mathbf{x}\right)-M(\mathbf{1}_{N}\otimes\mathbf{\theta}^{\ast})\right\|^{2}\nonumber
\\
\nonumber
&
\leq
4\beta^{2}\lambda_{N}(\overline{L})\|\mathbf{x}-\mathbf{1}_{N}\otimes\mathbf{\theta}^{\ast}\|^{2}+4K^{2}\|\mathbf{x}-\mathbf{1}_{N}\otimes\mathbf{\theta}^{\ast}\|^{2}
\end{align}
where the second step follows from the Lipschitz continuity of
$h_{n}(\cdot)$ and $K$ is defined in \eqref{nlth3}. To verify
Assumption~\textbf{(B.4)}, we have then along similar lines as in
Theorem~\ref{REas}
\begin{align}
\label{nlth18}
&
\|R\left(\mathbf{x}\right)\|^{2}+\mathbb{E}
\left[\left\|\Gamma\left(i+1,\mathbf{x},\omega\right)\right\|^{2}\right]
\leq
 k_{1}(1+V\left(\mathbf{x}\right))
\\
&
\hspace{3cm}
\leq
\nonumber
k_{1}(1+V\left(\mathbf{x}\right))-\left(R\left(\mathbf{x}\right),V_{\mathbf{x}}\left(\mathbf{x}\right)\right)
\end{align}
for a constant $k_{1}>0$ (the last step follows from
\eqref{nlth9}.) Hence, the assumptions are satisfied
and the claim follows.
\end{proof}
}
\textbf{Proof of Theorem~\ref{thnl}}
{\small
\begin{proof}
The recursive scheme in
\eqref{nalg2} satisfies Assumptions~\textbf{(B.1)-(B.2)} of
Theorem~\ref{RM}. To prove consistency, we verify
Assumptions~\textbf{(B.3)-(B.4)}. Consider the
Lyapunov function
\begin{equation}
\label{thnl4}
V\left(\mathbf{x}\right)=\|\mathbf{x}-\mathbf{1}_{N}\otimes\mathbf{\theta}^{\ast}\|^{2}
\end{equation}
Clearly,
\begin{align*}
V(\mathbf{1}_{N}\otimes\mathbf{\theta}^{\ast})=0,
V\left(\mathbf{x}\right)>0,
\mathbf{x}\neq\mathbf{1}_{N}\otimes\mathbf{\theta}^{\ast},
\hspace{-.2cm}
\lim_{\|\mathbf{x}\|\rightarrow\infty}
\hspace{-.2cm}
V\left(\mathbf{x}\right)=\infty
\end{align*}
Recall the definitions of
$R\left(\mathbf{x}\right),\Gamma\left(i+1,\mathbf{x},\omega\right)$ in
\eqref{nalg6}-\eqref{nalg7}, and the consensus
subspace in \eqref{nalg9}. We then have
\begin{align}
\label{thnl6}
&
\hspace{2cm}
\hspace{-.2cm}
\left(R\left(\mathbf{x}\right),V_{\mathbf{x}}\left(\mathbf{x}\right)\right) =
\\
&
\nonumber
-2\beta\left(\mathbf{x}-\mathbf{1}_{N}\otimes\mathbf{\theta}^{\ast}\right)^{T}\left(\overline{L}\otimes
I_{M}\right)\left(\mathbf{x}-\mathbf{1}_{N}
\otimes\mathbf{\theta}^{\ast}\right)-
\\
\nonumber
&
\hspace{1.5cm}
-
2\left(\mathbf{x}-\mathbf{1}_{N}\otimes\mathbf{\theta}^{\ast}\right)^{T}\left[M\left(\mathbf{x}\right)-M(\mathbf{1}_{N}\otimes\mathbf{\theta}^{\ast})\right]
\\
\nonumber
&
\leq
-2\beta\lambda_{2}(\overline{L})\|\mathbf{x}_{\mathcal{C}^{\perp}}\|^{2}
-
\\
&
\nonumber
\hspace{2.25cm}
-
2\left(\mathbf{x}-\mathbf{1}_{N}\otimes\mathbf{\theta}^{\ast}\right)^{T}
\left[M\left(\mathbf{x}\right)-M(\mathbf{x}_{\mathcal{C}})\right]
\nonumber
\\
&
\hspace{1.375cm}
-2\left(\mathbf{x}-\mathbf{1}_{N}\otimes
\mathbf{\theta}^{\ast}\right)^{T}
\left[M(\mathbf{x}_{\mathcal{C}})-
M(\mathbf{1}_{N}\otimes\mathbf{\theta}^{\ast})
\right]\nonumber
\\
\nonumber
&
\leq
-2\beta\lambda_{2}(\overline{L})
\|\mathbf{x}_{\mathcal{C}^{\perp}}\|^{2}
+
\\
&
\nonumber
\hspace{1.9cm}
+
2\left\|\left(\mathbf{x}-\mathbf{1}_{N}\otimes
\mathbf{\theta}^{\ast}\right)^{T}
\left[M\left(\mathbf{x}\right)-M(\mathbf{x}_{\mathcal{C}})\right]
\right\|
\nonumber
\\
&
\hspace{1.35cm}
-2\left(\mathbf{x}-\mathbf{1}_{N}\otimes
\mathbf{\theta}^{\ast}\right)^{T}
\left[M(\mathbf{x}_{\mathcal{C}})-M(\mathbf{1}_{N}
\otimes\mathbf{\theta}^{\ast})\right]\nonumber
\\
&
\leq
-2\beta\lambda_{2}(\overline{L})
\|\mathbf{x}_{\mathcal{C}^{\perp}}\|^{2}
+2K\|\mathbf{x}_{\mathcal{C}^{\perp}}\|\|\mathbf{x}-\mathbf{1}_{N}
\otimes\mathbf{\theta}^{\ast}\|
\nonumber
\\
&
\hspace{1.35cm}
-2\left(\mathbf{x}-\mathbf{1}_{N}
\otimes\mathbf{\theta}^{\ast}\right)^{T}
\left[M(\mathbf{x}_{\mathcal{C}})
-M(\mathbf{1}_{N}\otimes\mathbf{\theta}^{\ast})\right]
\nonumber
\\
\nonumber
&
=
-2\beta\lambda_{2}(\overline{L})
\|\mathbf{x}_{\mathcal{C}^{\perp}}\|^{2}+
2K\|\mathbf{x}_{\mathcal{C}^{\perp}}\|\|\mathbf{x}-
\mathbf{1}_{N}\otimes\mathbf{\theta}^{\ast}\|-
\\
&
\nonumber
-
2\mathbf{x}_{\mathcal{C}^{\perp}}^{T}
\left[M(\mathbf{x}_{\mathcal{C}})-M(\mathbf{1}_{N}\otimes
\mathbf{\theta}^{\ast})\right]
\nonumber
\\
&
\hspace{1.15cm}
 -2\left(\mathbf{x}_{\mathcal{C}}-
\mathbf{1}_{N}\otimes\mathbf{\theta}^{\ast}\right)^{T}
\left[M(\mathbf{x}_{\mathcal{C}})-M(\mathbf{1}_{N}\otimes
\mathbf{\theta}^{\ast})\right]
\nonumber
\\
\nonumber
&
\leq
-2\beta\lambda_{2}(\overline{L})
\|\mathbf{x}_{\mathcal{C}^{\perp}}\|^{2}+2K
\|\mathbf{x}_{\mathcal{C}^{\perp}}\|
\|\mathbf{x}-\mathbf{1}_{N}\otimes\mathbf{\theta}^{\ast}\|+
\\
\nonumber
&
\hspace{2.9cm}
2\left\|\mathbf{x}_{\mathcal{C}^{\perp}}^{T}
\left[M(\mathbf{x}_{\mathcal{C}})-M(\mathbf{1}_{N}
\otimes\mathbf{\theta}^{\ast})\right]\right\|
\nonumber
\\
&
\hspace{1.05cm}
-2\left(\mathbf{x}_{\mathcal{C}}-
\mathbf{1}_{N}\otimes\mathbf{\theta}^{\ast}\right)^{T}
\left[M(\mathbf{x}_{\mathcal{C}})-M(\mathbf{1}_{N}\otimes
\mathbf{\theta}^{\ast})\right]
\nonumber
\\
\nonumber
&
\leq
-2\beta\lambda_{2}(\overline{L})
\|\mathbf{x}_{\mathcal{C}^{\perp}}\|^{2}
+2K\|\mathbf{x}_{\mathcal{C}^{\perp}}\|\|\mathbf{x}
-\mathbf{1}_{N}\otimes\mathbf{\theta}^{\ast}\|
+
\\
\nonumber
&
\hspace{.4cm}
+
2K\|\mathbf{x}_{\mathcal{C}^{\perp}}\|\|\mathbf{x}_{\mathcal{C}}
-\mathbf{1}_{N}\otimes\mathbf{\theta}^{\ast}\|
-2\gamma\left\|\mathbf{x}_{\mathcal{C}}
-\mathbf{1}_{N}\otimes\mathbf{\theta}^{\ast}\right\|^{2}\nonumber
\\
&
\nonumber
= \left(-2\beta\lambda_{2}(\overline{L})+2K\right)
\|\mathbf{x}_{\mathcal{C}^{\perp}}\|^{2}
+
\\
&
\nonumber
\hspace{.4cm}
+
4K\|\mathbf{x}_{\mathcal{C}^{\perp}}\|\|\mathbf{x}_{\mathcal{C}}-\mathbf{1}_{N}\otimes\mathbf{\theta}^{\ast}\|-2\gamma\left\|\mathbf{x}_{\mathcal{C}}-\mathbf{1}_{N}\otimes\mathbf{\theta}^{\ast}\right\|^{2}
\end{align}
where the second to last step is justified because
$\mathbf{x}_{\mathcal{C}}=\mathbf{1}_{N}\otimes\widetilde{\mathbf{y}}$
for some $\widetilde{\mathbf{y}}\in\mathbb{R}^{M}$ and
\begin{align}
\label{thnl7}
&
\left(\mathbf{x}_{\mathcal{C}}-\mathbf{1}_{N}\otimes\mathbf{\theta}^{\ast}\right)^{T}\left[M(\mathbf{x}_{\mathcal{C}})-M(\mathbf{1}_{N}\otimes\mathbf{\theta}^{\ast})\right]
 =
 \\
 &
 \nonumber
 \hspace{3.5cm}
\sum_{n=1}^{N}\left(\widetilde{\mathbf{y}}-\mathbf{\theta}^{\ast}\right)^{T}\left[h_{n}(\widetilde{\mathbf{y}})-h_{n}(\mathbf{\theta}^{\ast})\right]\nonumber
\\ & = \left(\widetilde{\mathbf{y}}-\mathbf{\theta}^{\ast}\right)^{T}\sum_{n=1}^{N}\left[h_{n}(\widetilde{\mathbf{y}})-h_{n}(\mathbf{\theta}^{\ast})\right]\nonumber
\\ & = N\left(\widetilde{\mathbf{y}}-\mathbf{\theta}^{\ast}\right)^{T}\left[h(\widetilde{\mathbf{y}})-h(\mathbf{\theta}^{\ast})\right]\nonumber
\geq
N\gamma\left\|\widetilde{\mathbf{y}}-\mathbf{\theta}^{\ast}\right\|^{2}\nonumber
\\
\nonumber
&
=
\gamma\left\|\mathbf{x}_{\mathcal{C}}-\mathbf{1}_{N}\otimes\mathbf{\theta}^{\ast}\right\|^{2}
\end{align}
It can be shown that, if
$\beta>\frac{K^{2}+K\gamma}{\gamma\lambda_{2}{\overline{L}}}$, the
term on the R.H.S. of \eqref{thnl6} is always non-positive.
We thus have
\begin{equation}
\label{thnl8}
\left(R\left(\mathbf{x}\right),V_{\mathbf{x}}\left(\mathbf{x}\right)\right)\leq
0,~~\forall\mathbf{x}\in\mathbb{R}^{MN}
\end{equation}
By the continuity of $\left(R\left(\mathbf{x}\right),V_{\mathbf{x}}
\left(\mathbf{x}\right)\right)$ and the relative compactness of
$\left\{\mathbf{x}\in\mathbb{R}^{NM}\left|\rule{0cm}{.35cm}\right.
\epsilon<\left\|\mathbf{x}-\mathbf{1}_{N}\mathbf{\theta}^{\ast}
\right\|<\frac{1}{\epsilon}\right\}$,
we can show along similar lines as in Theorem~\ref{nlth} that
\begin{equation}
\label{thnl9}
\sup_{\epsilon<\|\mathbf{x}-\mathbf{1}_{N}\mathbf{\theta}^{\ast}\|<\frac{1}{\epsilon}}\left(R\left(\mathbf{x}\right),V_{\mathbf{x}}\left(\mathbf{x}\right)\right)<0,~~\forall\epsilon>0
\end{equation}
verifying Assumption~\textbf{(B.3)}. Assumption~\textbf{(B.4)} is verified in similar manner to Theorem~\ref{nlth}
and the result follows.
\end{proof}
}
\vspace{-.4cm}
\section{Proof of Lemma~\ref{prop_alpha}}
\label{proof_prop_alpha}
%
%
%
{\small
\begin{proof}[Proof of Lemma~\ref{prop_alpha}]
We prove for the case $\delta_{1}< 1$ first. Consider $j$
sufficiently large, such that,
\begin{equation}
\label{prop_alpha04} r_{1}(i)\leq 1,~~\forall i\geq j
\end{equation}
Then, for $k\geq j$, using  $1-a\leq e^{-a}$,
$0\leq a\leq 1$:
\begin{equation}
\label{prop_alpha3} \prod_{l=k+1}^{i-1}\left(1-r_{1}(l)\right)\leq
e^{-\sum_{l=k+1}^{i-1}r_{1}(l)}
\end{equation}
It follows from the properties of the Riemann integral that
\begin{eqnarray}
\label{prop_alpha4} \sum_{l=k+1}^{i-1}r_{1}(l) & = &
\sum_{l=k+1}^{i-1}\frac{a_{1}}{(l+1)^{\delta_{1}}}
\\ & \geq & a_{1}\int_{k+2}^{i+1}\frac{1}{t^{\delta_{1}}}dt\nonumber \\
& = &
\nonumber
\frac{a_{1}}{1-\delta_{1}}\left[(i+1)^{1-\delta_{1}}-(k+2)^{1-\delta_{1}}\right]
\end{eqnarray}
We thus have from \eqref{prop_alpha3}-\eqref{prop_alpha4}
\begin{align}
\label{prop_alpha5}
&\sum_{k=j}^{i-1}\left[\left(\prod_{l=k+1}^{i-1}\left(1-r_{1}(l)\right)\right)r_{2}(l)\right]
\leq
\\
&
\nonumber
\sum_{k=j}^{i-1}\left[e^{-\frac{a_{1}}{1-\delta_{1}}(i+1)^{1-\delta_{1}}}
e^{\frac{a_{1}}{1-\delta_{1}}(k+2)^{1-\delta_{1}}}\right]
\frac{a_{2}}{(k+1)^{\delta_{2}}}
=
\\
\nonumber
&
a_{2}e^{-\frac{a_{1}}{1-\delta_{1}}(i+1)^{1-\delta_{1}}}\sum_{k=j}^{i-1}\left[e^{\frac{a_{1}}{1-\delta_{1}}(k+2)^{1-\delta_{1}}}\frac{1}{(k+1)^{\delta_{2}}}\right]
\end{align}
%
From properties of Riemann integration, for
 $j$ large enough:
\begin{align}
\label{prop_alpha07}
&
\hspace{-.5cm}
\sum_{k=j}^{i-1}\left[e^{\frac{a_{1}}{1-\delta_{1}}(k+2)^{1-\delta_{1}}}\frac{1}{(k+1)^{\delta_{2}}}\right]
\leq
\\
\nonumber
&
\hspace{+2.5cm}
\sum_{k=j}^{i-1}\left[e^{\frac{a_{1}}{1-\delta_{1}}(k+2)^{1-\delta_{1}}}\frac{1}{(\frac{k}{2}+1)^{\delta_{2}}}\right]\nonumber
\\
\nonumber
&
=
2^{\delta_{2}}\sum_{k=j}^{i-1}\left[e^{\frac{a_{1}}{1-\delta_{1}}(k+2)^{1-\delta_{1}}}\frac{1}{(k+2)^{\delta_{2}}}\right]\nonumber
\\
\nonumber
&
=
2^{\delta_{2}}\sum_{k=j+2}^{i+1}\left[e^{\frac{a_{1}}{1-\delta_{1}}k^{1-\delta_{1}}}\frac{1}{k^{\delta_{2}}}\right]\nonumber
\\
\nonumber
&
=
2^{\delta_{2}}e^{\frac{a_{1}}{1-\delta_{1}}(i+1)^{1-\delta_{1}}}
\frac{1}{(i+1)^{\delta_{2}}}+
\\
\nonumber
&
\hspace{2.75cm}
+
2^{\delta_{2}}\sum_{k=j+2}^{i}\left[e^{\frac{a_{1}}{1-\delta_{1}}k^{1-\delta_{1}}}\frac{1}{k^{\delta_{2}}}\right]\nonumber
\end{align}
\begin{align}
\nonumber
&
\leq
2^{\delta_{2}}e^{\frac{a_{1}}{1-\delta_{1}}(i+1)^{1-\delta_{1}}}
\frac{1}{(i+1)^{\delta_{2}}}+
\\
\nonumber
&
\hspace{2.75cm}
+
2^{\delta_{2}}\int_{j+2}^{i+1}\left[e^{\frac{a_{1}}{1-\delta_{1}}t^{1-\delta_{1}}}\frac{1}{t^{\delta_{2}}}\right]dt
\end{align}
Again by the fundamental theorem of calculus,
\begin{align}
\nonumber
&
e^{\frac{a_{1}}{1-\delta_{1}}(i+1)^{1-\delta_{1}}}
=
a_{1}\int_{j+2}^{i+1}\left[e^{\frac{a_{1}}{1-\delta_{1}}t^{1-\delta_{1}}}\frac{1}{t^{\delta_{1}}}\right]dt+C_{1}
\\
\label{prop_alpha08}
&
\hspace{1cm}
= a_{1}\int_{j+2}^{i+1}\left[e^{\frac{a_{1}}{1-\delta_{1}}t^{1-\delta_{1}}}\frac{1}{t^{\delta_{2}}}t^{\delta_{2}-\delta_{1}}\right]dt+C_{1}
\end{align}
where $C_{1}=C_{1}(j)>0$ for sufficiently large $j$. From
\eqref{prop_alpha07}-\eqref{prop_alpha08} we have
\begin{align}
\label{prop_alpha09}
&\sum_{k=j}^{i-1}\left[\left(\prod_{l=k+1}^{i-1}\left(1-r_{1}(l)\right)\right)r_{2}(i)\right]
 =
 \\
 \nonumber
 &
 \hspace{.5cm}
a_{2}e^{-\frac{a_{1}}{1-\delta_{1}}(i+1)^{1-\delta_{1}}}
\sum_{k=j}^{i-1}
\left[e^{\frac{a_{1}}{1-\delta_{1}}(k+2)^{1-\delta_{1}}}
\frac{1}{(k+1)^{\delta_{2}}}
\right]
\leq
\nonumber
\\
&
\nonumber
\leq
\frac{2^{\delta_{2}}a_{2}
e^{\frac{a_{1}}{1-\delta_{1}}(i+1)^{1-\delta_{1}}}
\hspace{-.3cm}
\frac{1}{(i+1)^{\delta_{2}}}
\hspace{-.1cm}
+
\hspace{-.1cm}
2^{\delta_{2}}a_{2}
\hspace{-.1cm}
\int_{j+2}^{i+1}
\hspace{-.15cm}
\left[e^{\frac{a_{1}}{1-\delta_{1}}t^{1-\delta_{1}}}
\frac{1}{t^{\delta_{2}}}
\right]
\hspace{-.1cm}
dt}{e^{\frac{a_{1}}{1-\delta_{1}}(i+1)^{1-\delta_{1}}}}\nonumber
\\
&
\nonumber
= \frac{2^{\delta_{2}}a_{2}}{(i+1)^{\delta_{2}}}
+\frac{2^{\delta_{2}}a_{2}\int_{j+2}^{i+1}\left[e^{\frac{a_{1}}{1-\delta_{1}}t^{1-\delta_{1}}}\frac{1}{t^{\delta_{2}}}\right]dt}{e^{\frac{a_{1}}{1-\delta_{1}}(i+1)^{1-\delta_{1}}}}\nonumber
\\
&
\nonumber
\leq
\frac{2^{\delta_{2}}a_{2}}{(i+1)^{\delta_{2}}}+\frac{2^{\delta_{2}}a_{2}\int_{j+2}^{i+1}\left[e^{\frac{a_{1}}{1-\delta_{1}}t^{1-\delta_{1}}}\frac{1}{t^{\delta_{2}}}\right]dt}{a_{1}\int_{j+2}^{i+1}\left[e^{\frac{a_{1}}{1-\delta_{1}}t^{1-\delta_{1}}}\frac{1}{t^{\delta_{2}}}t^{\delta_{2}-\delta_{1}}\right]dt+C_{1}}
\end{align}
The second term stays bounded if
$\delta_{1}=\delta_{2}$ and goes to zero as $i\rightarrow\infty$
if $\delta_{1}<\delta_{2}$, thus establishing the Lemma for the
case $\delta_{1}<1$. Also, in the case $\delta_{1}=\delta_{2}$, we
have from \eqref{prop_alpha09}:
\begin{align}
\label{prop_alpha0100}
&\sum_{k=j}^{i-1}\left[\left(\prod_{l=k+1}^{i-1}\left(1-r_{1}(l)\right)\right)r_{2}(i)\right]
 \leq
\frac{2^{\delta_{2}}a_{2}}{(i+1)^{\delta_{2}}}+
\\
&
\nonumber
\hspace{1.5cm}
+\frac{2^{\delta_{2}}a_{2}}{a_{1}+
C_{1}\left[\int_{j+2}^{i+1}\left[e^{\frac{a_{1}}{1-\delta_{1}}t^{1-\delta_{1}}}\frac{1}{t^{\delta_{2}}}\right]dt\right]^{-1}}
\nonumber
\\
&
\nonumber
\hspace{1.5cm}
\leq 2^{\delta_{2}}a_{2}+\frac{2^{\delta_{2}}a_{2}}{a_{1}}
\end{align}
thus making the choice of $B$ in \eqref{prop_alpha1}
independent of $i,j$.

Now consider the case $\delta_{1}=1$. Consider $j$ sufficiently
large, such that,
\begin{equation}
\label{prop_alpha010} r_{1}(i)\leq 1,~~\forall i\geq j
\end{equation}
Using a similar set of manipulations for $k\geq j$, we have
\begin{eqnarray}
 \prod_{l=k+1}^{i-1}\left(1-r_{1}(l)\right) & \leq &
 e^{-a_{1}\sum_{l=k+1}^{i-1}\frac{1}{l+1}}\nonumber \\ & \leq &
 e^{-a_{1}\int_{k+2}^{i+1}\frac{1}{t}dt}\nonumber \\ & = &
 e^{-a_{1}\ln\left(\frac{i+1}{k+2}\right)}\nonumber \\ & = &
 \nonumber
 \frac{(k+2)^{a_{1}}}{(i+1)^{a_{1}}}
 \end{eqnarray}
We thus have
\begin{align}
&\sum_{k=j}^{i-1}
\hspace{-.075cm}
\left[
\hspace{-.075cm}
\left(\prod_{l=k+1}^{i-1}\left(1-r_{1}(l)\right)
\hspace{-.05cm}
\right)
\hspace{-.05cm}
r_{2}(i)
\hspace{-.05cm}
\right]
 \leq
\frac{a_{2}}{(i+1)^{a_{1}}}\sum_{k=j}^{i-1}\frac{(k+2)^{a_{1}}}{(k+1)^{\delta_{2}}}
\nonumber
\\
&
\hspace{4.2cm}
\leq
\frac{2^{\delta_{2}}a_{2}}{(i+1)^{a_{1}}}\sum_{k=j}^{i-1}\frac{(k+2)^{a_{1}}}{(k+2)^{\delta_{2}}}
\nonumber
\\
&
\nonumber
\hspace{4.2cm}
=\frac{2^{\delta_{2}}a_{2}}{(i+1)^{a_{1}}}\sum_{k=j+2}^{i+1}\frac{k^{a_{1}}}{k^{\delta_{2}}}
\end{align}
Now, if $a_{1}\geq\delta_{2}$, then
\begin{align}
\label{prop_alpha013}
&\sum_{k=j}^{i-1}\left[\left(\prod_{l=k+1}^{i-1}\left(1-r_{1}(l)\right)\right)r_{2}(i)\right]
\leq
\\
&
\nonumber
\hspace{.25cm}
\leq
\frac{2^{\delta_{2}}a_{2}}{(i+1)^{a_{1}}}\sum_{k=j+2}^{i+1}k^{a_{1}-\delta_{2}}\nonumber
\\
&
\nonumber
\hspace{.25cm}
=
\frac{2^{\delta_{2}}a_{2}}{(i+1)^{a_{1}}}\left[(i+1)^{a_{1}-\delta_{2}}+\sum_{k=j+2}^{i}k^{a_{1}-\delta_{2}}\right]\nonumber
\\
&
\nonumber
\hspace{.25cm}
\leq
\frac{2^{\delta_{2}}a_{2}}{(i+1)^{a_{1}}}\left[(i+1)^{a_{1}-\delta_{2}}+\int_{j+2}^{i+1}t^{a_{1}-\delta_{2}}dt\right]\nonumber
\\
&
\nonumber
\hspace{.25cm}
=
\hspace{-.01cm}
\frac{2^{\delta_{2}}a_{2}}{(i+1)^{\delta_{2}}}
\hspace{-.01cm}
+
\hspace{-.01cm}
\frac{2^{\delta_{2}}a_{2}}{a-\delta_{2}+1}
\frac{(i+1)^{a-\delta_{2}+1}
\hspace{-.02cm}
-
\hspace{-.02cm}
(j+2)^{a-\delta_{2}+1}}{(i+1)^{a_{1}}}
\end{align}
The second term is bounded if $\delta_{2}=1$
and vanishes if $\delta_{2}>1$. If $a_{1}<\delta_{2}$
is resolved similarly.
%
%
\end{proof}
}
\vspace*{-.5cm}
\section{Proofs of Lemmas~\ref{nlulm},\ref{circcons}}
\label{proof_circcons}
\label{Appen}

\textbf{Proof of Lemma~\ref{nlulm}}

{\small

\begin{proof}
It follows from \eqref{nlu04} and \eqref{nlu6}, and the fact that
\begin{equation}
\label{nlulm3} \left(\mathbf{1}_{N}\otimes
I_{M}\right)^{T}\left(\overline{L}\otimes I_{M}\right)=\mathbf{0}
\end{equation}
that the evolution of the sequence,
$\left\{\widetilde{\mathbf{x}}^{\circ}_{\mbox{\scriptsize{avg}}}(i)\right\}_{i\geq
0}$ is given by
\begin{equation}
\label{nlulm4}
\hspace{-.05cm}
\widetilde{\mathbf{x}}^{\circ}_{\mbox{\scriptsize{avg}}}(i+1)
\hspace{-.045cm}
=
\hspace{-.045cm}
\widetilde{\mathbf{x}}^{\circ}_{\mbox{\scriptsize{avg}}}(i)
\hspace{-.045cm}
-
\hspace{-.045cm}
\alpha(i)
\hspace{-.06cm}
\left[
\hspace{-.035cm}
\widetilde{\mathbf{x}}^{\circ}_{\mbox{\scriptsize{avg}}}(i)
\hspace{-.035cm}
-
\hspace{-.045cm}
\frac{1}{N}
\hspace{-.095cm}
\sum_{n=1}^{N}g_{n}(\mathbf{z}_{n}(i))
\hspace{-.035cm}
\right]
\end{equation}
We note that \eqref{nlulm4} can be written as
\[
\widetilde{\mathbf{x}}^{\circ}_{\mbox{\scriptsize{avg}}}
(i
\hspace{-.05cm}
+
\hspace{-.05cm}
1)
\hspace{-.05cm}
=
\hspace{-.05cm}
\widetilde{\mathbf{x}}^{\circ}_{\mbox{\scriptsize{avg}}}(i)
\hspace{-.05cm}
+
\hspace{-.05cm}
\alpha(i)
\hspace{-.1cm}
\left[R(\widetilde{\mathbf{x}}^{\circ}_{\mbox{\scriptsize{avg}}}(i))
\hspace{-.05cm}
+
\hspace{-.05cm}
\Gamma(i
\hspace{-.05cm}
+
\hspace{-.05cm}
1,\widetilde{\mathbf{x}}^{\circ}_{\mbox{\scriptsize{avg}}}(i),\omega)\right]
\]
where
\begin{align}
\label{nlulm7}
R(\mathbf{y})&=-\left(\mathbf{y}-h(\mathbf{\theta}^{\ast})\right),
\\
\label{nlulm7-b}
\Gamma(i+1,\mathbf{y},\omega)&=
\frac{1}{N}\sum_{n=1}^{N}g_{n}(\mathbf{z}_{n}(i))-
h(\mathbf{\theta}^{\ast}), \mathbf{y}\in\mathbb{R}^{M}
\end{align}
Such a definition of $R(\cdot),\Gamma(\cdot)$ clearly satisfies
Assumptions~\textbf{(B.1)-(B.2)} of Theorem~\ref{RM}. Now, defining
\begin{equation}
\label{nlulm8}
V(\mathbf{y})=\|\mathbf{y}-h(\mathbf{\theta}^{\ast})\|^{2}
\end{equation}
we have
\begin{equation}
\label{nlulm9}
\hspace{-.05cm}
V\left(h(\mathbf{\theta}^{\ast})\right)=0,
\hspace{-.025cm}
V(\mathbf{y})>0,
\hspace{-.025cm}
\mathbf{y}\neq h(\mathbf{\theta}^{\ast}),
\hspace{-.2cm}
\lim_{\|\mathbf{y}\|\rightarrow\infty}
\hspace{-.1cm}
V(\mathbf{y})=\infty
\end{equation}
Also, we have for $\epsilon>0$
\begin{align}
\label{nlulm10}
&\sup_{\epsilon<\|\mathbf{y}-h(\mathbf{\theta}^{\ast})\|<\frac{1}{\epsilon}}\left(R(\mathbf{y}),V_{\mathbf{y}}(\mathbf{y})\right)
=
\\
&
\nonumber
\hspace{2cm}
=\sup_{\epsilon<\|\mathbf{y}-h(\mathbf{\theta}^{\ast})\|<\frac{1}{\epsilon}}\left(-2\|\mathbf{y}-h(\mathbf{\theta}^{\ast})\|^{2}\right)\nonumber
\\
&
\nonumber
\hspace{2cm}
\leq
 -2\epsilon^{2}\nonumber
 \\
&
\nonumber
\hspace{2cm}
<
0
\end{align}
thus verifying Assumption~\textbf{(B.3)}. Finally from
\eqref{nlprobform4} and \eqref{nlulm7}-\eqref{nlulm7-b}, we have
\begin{align}
\label{nlulm11}
&\|R(\mathbf{y})\|^{2}+\mathbb{E}_{\mathbf{\theta}^{\ast}}
\left[\left\|\Gamma(i+1,\mathbf{y},\omega)\right\|^{2}\right]
=
\\
&
\nonumber
\hspace{3cm}
=
\|\mathbf{y}-h(\mathbf{\theta}^{\ast})\|^{2}+\eta(\mathbf{\theta}^{\ast})\nonumber
\\
&
\nonumber
\hspace{3cm}
\leq
k_{1}(1+V(\mathbf{y}))\nonumber
\\
&
\nonumber
\hspace{3cm}
\leq
k_{1}(1+V(\mathbf{y}))-\left(R(\mathbf{y}),V_{\mathbf{y}}(\mathbf{y})\right)
\end{align}
for $k_{1}=\max(1,\eta(\mathbf{\theta}^{\ast}))$. Thus the
Assumptions~\textbf{(B.1)-(B.4)} are satisfied, and we have the claim
in \eqref{nlulm1}.


To establish \eqref{nlulm2}, we note that, for sufficiently large
$i$,
\begin{align}
\label{nlulm12}
&
\mathbb{E}_{\mathbf{\theta}^{\ast}}
\left[\left\|\widetilde{\mathbf{x}}^{\circ}_{\mbox{\scriptsize{avg}}}(i)
-h(\mathbf{\theta}^{\ast})\right\|^{2}\right]
=
\\
&
\nonumber
\hspace{1cm}
=
(1-\alpha(i-1))^{2}\mathbb{E}_{\mathbf{\theta}^{\ast}}
\left[\left\|\widetilde{\mathbf{x}}^{\circ}_{\mbox{\scriptsize{avg}}}(i-1)
-h(\mathbf{\theta}^{\ast})\right\|^{2}\right]
+
\\
&
\nonumber
\hspace{3cm}
+
\alpha^{2}(i-1)\eta(\mathbf{\theta}^{\ast})
\\
&
\nonumber
\hspace{1cm}
\leq
 (1-\alpha(i-1))\mathbb{E}_{\mathbf{\theta}^{\ast}}
\left[\left\|\widetilde{\mathbf{x}}^{\circ}_{\mbox{\scriptsize{avg}}}(i-1)
-h(\mathbf{\theta}^{\ast})\right\|^{2}\right]+
\\
&
\nonumber
\hspace{3cm}
+\alpha^{2}(i-1)\eta(\mathbf{\theta}^{\ast})
\end{align}
where the last step follows from the fact that
$0\leq(1-\alpha(i))\leq 1$ for sufficiently large $i$. Continuing
the recursion in \eqref{nlulm12}, we have for sufficiently
large $j\leq i$
\begin{align}
\label{nlulm13}
&\mathbb{E}_{\mathbf{\theta}^{\ast}}
\left[\left\|\widetilde{\mathbf{x}}^{\circ}_{\mbox{\scriptsize{avg}}}(i)
-h(\mathbf{\theta}^{\ast})\right\|^{2}\right]
\leq
\\
&
\nonumber
\hspace{1cm}
\leq
\left(\prod_{k=j}^{i-1}(1-\alpha(k))\right)
\left\|\widetilde{\mathbf{x}}^{\circ}_{\mbox{\scriptsize{avg}}}(0)
-h(\mathbf{\theta}^{\ast})\right\|^{2}
+
\\
&
\nonumber
\hspace{2cm}
+
\eta(\mathbf{\theta}^{\ast})\sum_{k=j}^{i-1}
\left[\left(\prod_{l=k+1}^{i-1}
\left(1-\alpha(l)\right)\right)\alpha^{2}(k)\right]\nonumber
\\
&
\nonumber
\hspace{1cm}
\leq
\left(e^{-\sum_{k=j}^{i-1}\alpha(k)}\right)\left\|\widetilde{\mathbf{x}}^{\circ}_{\mbox{\scriptsize{avg}}}(0)-h(\mathbf{\theta}^{\ast})\right\|^{2} +
\\
&
\nonumber
\hspace{2cm}
+
\eta(\mathbf{\theta}^{\ast})\sum_{k=0}^{i-1}\left[\left(\prod_{l=k+1}^{i-1}\left(1-\alpha(l)\right)\right)\alpha^{2}(k)\right]
\end{align}
From Assumption~\textbf{(D.5)}, we note that
$\sum_{k=j}^{i-1}\alpha(k)\rightarrow\infty$ as
$i\rightarrow\infty$ because $0.5<\tau_{1}\leq 1$. Thus, the first
term in \eqref{nlulm13} goes to zero as $i\rightarrow\infty$.
The second term in \eqref{nlulm13} falls under the purview of
Lemma~\ref{prop_alpha} with $\delta_{1}=\tau_{1}$ and
$\delta_{2}=2\tau_{1}$ and hence goes to zero as
$i\rightarrow\infty$. We thus have
\begin{equation}
\label{nlulm14}
\lim_{i\rightarrow\infty}\mathbb{E}_{\mathbf{\theta}^{\ast}}
\left[\left\|\widetilde{\mathbf{x}}^{\circ}_{\mbox{\scriptsize{avg}}}(i)
-h(\mathbf{\theta}^{\ast})\right\|^{2}\right]=0
\end{equation}
\end{proof}
}



\textbf{Proof of Lemma~\ref{circcons}}

{\small

\begin{proof}
Recall from \eqref{nlu04} and \eqref{nlulm4} that the evolution of
the sequences
$\left\{\widetilde{\mathbf{x}}^{\circ}(i)\right\}_{i\geq 0}$ and
$\left\{\widetilde{\mathbf{x}}^{\circ}_{\mbox{\scriptsize{avg}}}(i)\right\}_{i\geq 0}$ are
given by
\begin{align}
\label{circcons2}
\widetilde{\mathbf{x}}^{\circ}(i+1)
&
=
\widetilde{\mathbf{x}}^{\circ}(i)-\beta(i)\left(\overline{L}\otimes
I_{M}\right)\widetilde{\mathbf{x}}^{\circ}(i)
-
\\
&
\nonumber
\hspace{1cm}
-
\alpha(i)\left[\widetilde{\mathbf{x}}(i)-J(\mathbf{z}(i))\right]
\label{circcons3}
\\
\widetilde{\mathbf{x}}^{\circ}_{\mbox{\scriptsize{avg}}}(i+1)
&
=
\widetilde{\mathbf{x}}^{\circ}_{\mbox{\scriptsize{avg}}}(i)
-
\\
&
\nonumber
\hspace{1cm}
-
\alpha(i)\left[\widetilde{\mathbf{x}}^{\circ}_{\mbox{\scriptsize{avg}}}(i)-\frac{1}{N}\sum_{n=1}^{N}g_{n}(\mathbf{z}_{n}(i))\right]
\end{align}
To establish the claim~\eqref{circcons1},
Lemma~\ref{circcons}, we prove
\begin{equation}
\label{circcons4}
\mathbb{P}_{\mathbf{\theta}^{\ast}}\left[\lim_{i\rightarrow\infty}\left\|\widetilde{\mathbf{x}}^{\circ}(i)-\left(\mathbf{1}_{N}\otimes\widetilde{\mathbf{x}}^{\circ}_{\mbox{\scriptsize{avg}}}(i)\right)\right\|=0\right]=1
\end{equation}
Recall the matrix
\begin{equation}
\label{circcons5}
P=\frac{1}{N}\left(\mathbf{1}_{N}\otimes
I_{M}\right)\left(\mathbf{1}_{N}\otimes I_{M}\right)^{T}
\end{equation}
and note that
\[
P\widetilde{\mathbf{x}}^{\circ}(i)=
\mathbf{1}_{N}\otimes\widetilde{\mathbf{x}}^{\circ}_{\mbox{\scriptsize{avg}}}(i),~~P\mathbf{1}_{N}\otimes\widetilde{\mathbf{x}}^{\circ}_{\mbox{\scriptsize{avg}}}(i)=\mathbf{1}_{N}\otimes\widetilde{\mathbf{x}}^{\circ}_{\mbox{\scriptsize{avg}}}(i), \forall i
\]
From \eqref{circcons2}-\eqref{circcons3}, we then have
\begin{align}
\label{circcons7}
&
\widetilde{\mathbf{x}}^{\circ}(i+1)
-\left(\mathbf{1}_{N}
\otimes\widetilde{\mathbf{x}}^{\circ}_{\mbox{\scriptsize{avg}}}(i+1)
\right)
=
\\
&
\nonumber
=
\left[I_{NM}-\beta(i)\left(\overline{L}\otimes
I_{M}\right)-\alpha(i)I_{NM}-P\right]
\left[\widetilde{\mathbf{x}}^{\circ}(i)
-
\right.
\\
&
\nonumber
\hspace{5cm}
\left.
-\left(\mathbf{1}_{N}\otimes
\widetilde{\mathbf{x}}^{\circ}_{\mbox{\scriptsize{avg}}}(i)\right)
\right]
\nonumber
\\
&
\nonumber
\hspace{3.75cm}
+\alpha(i)\left[J(\mathbf{z}(i))-PJ(\mathbf{z}(i))\right]
\end{align}
Choose $\delta$ satisfying
\begin{equation}
\label{circcons8}
0<\delta<\tau_{1}-\frac{1}{2+\epsilon_{1}}-\tau_{2}.
\end{equation}
Such a choice exists by Assumption~\textbf{(D.5)}.
Now claim:
\begin{equation}
\label{circcons9}
\mathbb{P}_{\mathbf{\theta}^{\ast}}\left[\lim_{i\rightarrow\infty}\frac{1}{(i+1)^{\frac{1}{2+\epsilon_{1}}+\delta}}\left\|J(\mathbf{z}(i))-PJ(\mathbf{z}(i))\right\|=0\right]=1
\end{equation}
Indeed, consider any $\epsilon>0$. We then have from
Assumption~\textbf{(D.4)} and Chebyshev's inequality
\begin{align}
\nonumber
&\sum_{i\geq 0}
\mathbb{P}_{\mathbf{\theta}^{\ast}}
\left[\frac{1}{(i+1)^{\frac{1}{2+\epsilon_{1}}+\delta}}
\left\|J(\mathbf{z}(i))-PJ(\mathbf{z}(i))
\right\|>\epsilon
\right]
\leq
\\
&
\nonumber
\hspace{4.15cm}
\leq
\hspace{-.125cm}
\sum_{i\geq
0}\frac{1}{(i+1)^{1+\delta(2+\epsilon_{1})}\epsilon^{2+\epsilon_{1}}}
\nonumber
\\
&
\mathbb{E}_{\mathbf{\theta}}
\left[\left\|J(\mathbf{z}(i))
-PJ(\mathbf{z}(i))\right\|^{2+\epsilon_{1}}\right]
\nonumber
\hspace{-.1cm}
=
\hspace{-.1cm}
\frac{\kappa(\mathbf{\theta}^{\ast})}{\epsilon^{2+\epsilon_{1}}}
\hspace{-.125cm}
\sum_{i\geq
0}\frac{1}{(i+1)^{1+\delta(2+\epsilon_{1})}}\nonumber
\\
&
\hspace{4.15cm}
<
\nonumber
\infty
\end{align}
It then follows from the Borel-Cantelli Lemma
(see~\cite{Kallenberg}) that for arbitrary $\epsilon>0$
\begin{equation}
\label{circcons11}
\mathbb{P}_{\mathbf{\theta}^{\ast}}\left[\frac{1}{(i+1)^{\frac{1}{2+\epsilon_{1}}+\delta}}\left\|J(\mathbf{z}(i))-PJ(\mathbf{z}(i))\right\|>\epsilon~\mbox{i.o.}~\right]=0
\end{equation}
where i.o. stands for infinitely often. Since the above holds for
$\epsilon$ arbitrarily small, we have (see~\cite{Kallenberg}) the
a.s. claim in \eqref{circcons9}.

Consider the set $\Omega_{1}\subset\Omega$ with
$\mathbb{P}_{\mathbf{\theta}^{\ast}}\left[\Omega_{1}\right]=1$,
where the a.s. property in \eqref{circcons9} holds. Also,
consider the set $\Omega_{2}\subset\Omega$ with
$\mathbb{P}_{\mathbf{\theta}^{\ast}}\left[\Omega_{2}\right]=1$,
where the sequence
$\left\{\widetilde{\mathbf{x}}^{\circ}_{\mbox{\scriptsize{avg}}}(i)\right\}_{i\geq
0}$ converges to $h(\mathbf{\theta}^{\ast})$. Let
$\Omega_{3}=\Omega_{1}\cap\Omega_{2}$. It is clear that
$\mathbb{P}_{\mathbf{\theta}^{\ast}}\left[\Omega_{3}\right]=1$. We
will now show that, on $\Omega_{3}$, the sample paths of the
sequence $\left\{\widetilde{\mathbf{x}}^{\circ}(i)\right\}_{i\geq 0}$
converge to $\left(\mathbf{1}_{N}\otimes
h(\mathbf{\theta}^{\ast})\right)$, thus proving the Lemma. In the
following we index the sample paths by $\omega$ to emphasize the
fact that we are establishing properties pathwise.

From \eqref{circcons7}, we have on $\omega\in\Omega_{3}$
\begin{align*}
&
\left\|\widetilde{\mathbf{x}}^{\circ}(i+1,\omega)
-\left(\mathbf{1}_{N}
\otimes\widetilde{\mathbf{x}}^{\circ}_{\mbox{\scriptsize{avg}}}
(i+1,\omega)\right)\right\|
\leq
\\
&
\hspace{.5cm}
\leq
\left\|I-\beta(i)\left(\overline{L}\otimes
I_{M}\right)-\alpha(i)I_{NM}-P\right\|
\\
&
\hspace{3.5cm}
\left\|\widetilde{\mathbf{x}}^{\circ}(i,\omega)
-\left(\mathbf{1}_{N}\otimes
\widetilde{\mathbf{x}}^{\circ}_{\mbox{\scriptsize{avg}}}(i,\omega)\right)\right\|
\nonumber
\\
\nonumber
&
\hspace{1.0cm}
+\frac{a}{(i+1)^{\tau_{1}-\frac{1}{2+\epsilon_{1}}-\delta}}
\\
&
\hspace{1.95cm}
\left\|\frac{1}{(i+1)^{\frac{1}{2+\epsilon_{1}}+\delta}}
\left[J(\mathbf{z}(i,\omega))-PJ(\mathbf{z}(i,\omega))\right]\right\|
\end{align*}
For sufficiently large $i$, we have
\begin{equation}
\label{circcons13}
\hspace{-.05cm}
\left\|I
\hspace{-.1cm}
-
\hspace{-.1cm}
\beta(i)\left(\overline{L}\otimes
I_{M}\right)
\hspace{-.1cm}
-
\hspace{-.1cm}
\alpha(i)I_{NM}-P\right\|\leq
1-\beta(i)\lambda_{2}(\overline{L})
\end{equation}
From \eqref{circcons11} for $\omega\in\Omega_{3}$ we can
choose $\epsilon>0$ and $j(\omega)$ such that $\forall i\geq j(\omega)$
\begin{equation}
\label{circcons14}
\left\|\frac{1}{(i+1)^{\frac{1}{2+\epsilon_{1}}
+\delta}}
\left[J(\mathbf{z}(i,\omega))-
PJ(\mathbf{z}(i,\omega))\right]
\right\|\leq\epsilon.
\end{equation}
Let $j(\omega)$ be sufficiently large such that
\eqref{circcons13} is also satisfied in addition to
\eqref{circcons14}. We then have for
$\omega\in\Omega_{3},~i\geq j(\omega)$
\begin{align}
\label{circcons15}
&\left\|\widetilde{\mathbf{x}}^{\circ}(i,\omega)
-\left(\mathbf{1}_{N}\otimes
\widetilde{\mathbf{x}}^{\circ}_{\mbox{\scriptsize{avg}}}(i,\omega)\right)\right\|
\leq
\\
&
\nonumber
\leq
\left(\prod_{k=j(\omega)}^{i-1}\left(1-\beta(k)\lambda_{2}(\overline{L})
\right)\right)
\left\|\widetilde{\mathbf{x}}^{\circ}(j(\omega),\omega)
-
\right.
\\
&
\nonumber
\hspace{4.0cm}
-
\left.
\left(\mathbf{1}_{N}
\otimes\widetilde{\mathbf{x}}^{\circ}_{\mbox{\scriptsize{avg}}}
(j(\omega),\omega)\right)\right\|
+
\nonumber
\\
&
\nonumber
+
a\epsilon
\hspace{-.25cm}
\sum_{k=j(\omega)}^{i-1}
\left[\left(\prod_{l=k+1}^{i-1}\left(1-\beta(l)\lambda_{2}(\overline{L})
\right)\right)\frac{1}{(k+1)^{\tau_{1}-\frac{1}{2+\epsilon_{1}}-\delta}}
\right]
\nonumber
\end{align}
For the first term on the R.H.S. of \eqref{circcons15} we
note that
\begin{eqnarray}
\label{circcons16}
\prod_{k=j(\omega)}^{i-1}\left(1-\beta(k)\lambda_{2}(\overline{L})\right)
& \leq &
e^{-\lambda_{2}(\overline{L})\sum_{k=j(\omega)}^{i-1}\beta(k)}
\\
\nonumber
& = &
e^{-b\lambda_{2}(\overline{L})\sum_{k=j(\omega)}^{i-1}\frac{1}{(k+1)^{\tau_{2}}}}
\end{eqnarray}
which goes to zero as $i\rightarrow\infty$ since $\tau_{2}<1$ by
Assumption~\textbf{(D.5)}. Hence the first term on the R.H.S. of
\eqref{circcons15} goes to zero as $i\rightarrow\infty$. The
summation in the second term on the R.H.S. of
\eqref{circcons15} falls under the purview of
Lemma~\ref{prop_alpha} with $\delta_{1}=\tau_{2}$ and
$\delta_{2}=\tau_{1}-\frac{1}{2+\epsilon_{1}}-\delta$. It follows
from the choice of $\delta$ in \eqref{circcons8} and
Assumption~\textbf{(D.5)} that $\delta_{1}<\delta_{2}$ and hence
the term
$\sum_{k=j(\omega)}^{i-1}\left[\left(\prod_{l=k+1}^{i-1}\left(1-\beta(l)\lambda_{2}(\overline{L})\right)\right)\frac{1}{(k+1)^{\tau_{1}-\frac{1}{2+\epsilon_{1}}-\delta}}\right]\rightarrow
0$ as $i\rightarrow\infty$. We then conclude from
\eqref{circcons15} that, for $\omega\in\Omega_{3}$
\begin{equation}
\label{circcons17}
\lim_{i\rightarrow\infty}\left\|\widetilde{\mathbf{x}}^{\circ}(i,\omega)-\left(\mathbf{1}_{N}\otimes\widetilde{\mathbf{x}}^{\circ}_{\mbox{\scriptsize{avg}}}(i,\omega)\right)\right\|=0
\end{equation}
The Lemma then follows from the fact that
$\mathbb{P}_{\mathbf{\theta}^{\ast}}\left[\Omega_{3}\right]=1$.

To establish \eqref{circcons200}, we have from
\eqref{circcons7}
%
\begin{align*}
\nonumber
&
\left\|\widetilde{\mathbf{x}}^{\circ}(i+1)-\left(\mathbf{1}_{N}
\otimes\widetilde{\mathbf{x}}^{\circ}_{\mbox{\scriptsize{avg}}}(i+1)
\right)\right\|^{2}
\leq
\\
&
\nonumber
\hspace{1cm}
\leq
\left\|I-\beta(i)\left(\overline{L}\otimes
I_{M}\right)-\alpha(i)I_{NM}-P\right\|^{2}
\\
&
\nonumber
\hspace{3cm}
\left\|\widetilde{\mathbf{x}}^{\circ}(i)-\left(\mathbf{1}_{N}
\otimes\widetilde{\mathbf{x}}^{\circ}_{\mbox{\scriptsize{avg}}}(i)
\right)\right\|^{2}
+
\\
\nonumber
&
\hspace{.65cm}
+ 2\alpha(i)\left\|I-\beta(i)\left(\overline{L}\otimes
I_{M}\right)-\alpha(i)I_{NM}-P\right\|
\\
&
\nonumber
\hspace{.15cm}
\left\| \widetilde{\mathbf{x}}^{\circ}(i)-\left(\mathbf{1}_{N}
\otimes\widetilde{\mathbf{x}}^{\circ}_{\mbox{\scriptsize{avg}}}(i)
\right)\right\|\left\|J(\mathbf{z}(i))-PJ(\mathbf{z}(i))\right\|
+
\\
&
\hspace{2.75cm}
+\alpha^{2}(i)\left\|J(\mathbf{z}(i))-PJ(\mathbf{z}(i))\right\|^{2}
\end{align*}
Taking expectations on both sides and from \eqref{nlu0001}
\begin{align}
\nonumber
&
\mathbb{E}_{\mathbf{\theta}^{\ast}}
\left[\left\|\widetilde{\mathbf{x}}^{\circ}(i+1)
-\left(\mathbf{1}_{N}\otimes
\widetilde{\mathbf{x}}^{\circ}_{\mbox{\scriptsize{avg}}}(i+1)\right)
\right\|^{2}\right]
\leq
\\
&
\nonumber
\hspace{1cm}
\leq\left\|I-\beta(i)\left(\overline{L}\otimes
I_{M}\right)-\alpha(i)I_{NM}-P\right\|^{2}
\\
&
\nonumber
\hspace{3.0cm}
\mathbb{E}_{\mathbf{\theta}^{\ast}} \left[
\left\|\widetilde{\mathbf{x}}^{\circ}(i)
-\left(\mathbf{1}_{N}\otimes
\widetilde{\mathbf{x}}^{\circ}_{\mbox{\scriptsize{avg}}}(i)\right)
\right\|^{2}\right]+
\\
\nonumber
&
\hspace{.35cm}
+2\alpha(i)\left\|I-\beta(i)\left(\overline{L}\otimes
I_{M}\right)-\alpha(i)I_{NM}-P\right\|\kappa_{1} \left(\mathbf{\theta}^{\ast}\right)
\\
&
\nonumber
\hspace{3cm}
\mathbb{E}_{\mathbf{\theta}^{\ast}}
\left[ \left\|\widetilde{\mathbf{x}}^{\circ}(i)
-\left(\mathbf{1}_{N}\otimes
\widetilde{\mathbf{x}}^{\circ}_{\mbox{\scriptsize{avg}}}(i)\right)
\right\|^{2}\right]
+
\\
\nonumber
&
\hspace{.2cm}
+2\alpha(i)\left\|I-\beta(i)\left(\overline{L}\otimes
I_{M}\right)-\alpha(i)I_{NM}-P\right\|
\kappa_{1}\left(\mathbf{\theta}^{\ast}\right)
+
\\
&
\nonumber
\hspace{5.75cm}
+\alpha^{2}(i)\kappa_{2}(\mathbf{\theta}^{\ast})
\end{align}
where we used the inequality that $\forall i$
\[
 \left\|\widetilde{\mathbf{x}}^{\circ}(i)
\hspace{-.1cm}
-
\hspace{-.1cm}
\left(\mathbf{1}_{N}\otimes
\widetilde{\mathbf{x}}^{\circ}_{\mbox{\scriptsize{avg}}}(i)\right)
\hspace{-.05cm}
\right\|
\hspace{-.05cm}
\leq
\hspace{-.1cm}
\left\|
\widetilde{\mathbf{x}}^{\circ}(i)
\hspace{-.05cm}
-
\hspace{-.075cm}
\left(\mathbf{1}_{N}\otimes
\widetilde{\mathbf{x}}^{\circ}_{\mbox{\scriptsize{avg}}}(i)\right)
\right\|^{2}
\hspace{-.1cm}
+
\hspace{-.1cm}
1.
\]
 Choose $j$ sufficiently large such that $\forall i\geq j$
\[
 \left\|I-\beta(i)\left(\overline{L}\otimes
I_{M}\right)-\alpha(i)I_{NM}-P\right\|
1-\beta(i)\lambda_{2}(\overline{L}).
\]
For $i\geq j$, it can be shown that, for $c_{1}>0$ a constant:
\begin{align}
\label{circcons21}
&
\mathbb{E}_{\mathbf{\theta}^{\ast}}
\left[\left\|\widetilde{\mathbf{x}}^{\circ}(i+1)
-\left(\mathbf{1}_{N}\otimes
\widetilde{\mathbf{x}}^{\circ}_{\mbox{\scriptsize{avg}}}(i+1)\right)
\right\|^{2}\right]
\leq
\\
&
\nonumber
\leq
\left[1-\beta(i)\lambda_{2}(\overline{L})+2\alpha(i)
\kappa_{1}(\mathbf{\theta}^{\ast})\right]
\\
&
\nonumber
\hspace{1cm}
\mathbb{E}_{\mathbf{\theta}^{\ast}}
\left[ \left\|\widetilde{\mathbf{x}}^{\circ}(i)
-\left(\mathbf{1}_{N}\otimes
\widetilde{\mathbf{x}}^{\circ}_{\mbox{\scriptsize{avg}}}(i)\right)
\right\|^{2}\right]
+\alpha(i)c_{1}.
\end{align}
 Now choose $j_{1}\geq j$ and
$0<c_{2}<\lambda_{2}(\overline{L})$\footnote{Such a choice exists
because $\tau_{1}>\tau_{2}$.} such that,
\begin{equation}
\label{circcons21.1}
1-\beta(i)\lambda_{2}(\overline{L})+2\alpha(i)\kappa_{1}(\mathbf{\theta}^{\ast})\leq
1-\beta(i)c_{2},~~\forall i\geq j_{1}
\end{equation}
Then the claim in~\eqref{circcons200} follows because  for $i\geq j_{1}$
{
\small
\begin{align*}
&
\mathbb{E}_{\mathbf{\theta}^{\ast}} \left[
\hspace{-.08cm}\rule{0cm}{.325cm}
\left\|\widetilde{\mathbf{x}}^{\circ}(i)
-\left(\mathbf{1}_{N}\otimes
\widetilde{\mathbf{x}}^{\circ}_{\mbox{\scriptsize{avg}}}(i)\right)
\right\|^{2}\right]
\leq
\\
&
\nonumber
\left(\prod_{k=j_{1}}^{i-1}\left(1-\beta(k)c_{2}\right)\right)
\hspace{-.075cm}
\mathbb{E}_{\mathbf{\theta}^{\ast}}
\hspace{-.075cm}
\left[
\hspace{-.08cm}\rule{0cm}{.325cm} \left\|
\widetilde{\mathbf{x}}^{\circ}(j_{1})-\left(\mathbf{1}_{N}\otimes
\widetilde{\mathbf{x}}^{\circ}_{\mbox{\scriptsize{avg}}}(j)\right)
\right\|^{2}\right]
\hspace{-.075cm}
+
\\
&
\nonumber
\hspace{2cm}
+c_{1}
\sum_{k=j_{1}}^{i-1}\left[\left(\prod_{l=k+1}^{i-1}\left(1-
\beta(l)c_{2}\right)\right) \alpha(k)\right]
\end{align*}
}
and the first and second terms on the R.H.S.~of \eqref{circcons21} vanish as $i\rightarrow\infty$ by the argument in
\eqref{circcons16} and Lemma~\ref{prop_alpha}, respectively.
\end{proof}
}
\vspace*{-.75cm}
\section{Proof of Lemma~\ref{hatconslemma}}
\label{proof_hatconslemma}
{\small
\begin{proof}[Proof of Lemma~\ref{hatconslemma}]
From \eqref{nlu04} and \eqref{hatcons} we have
\begin{align}
\label{hatcons2}
&
\widehat{\mathbf{x}}(i+1)-\widetilde{\mathbf{x}}^{\circ}(i+1)=
\left[I_{NM}-\beta(i)\left(\overline{L}\otimes
I_{M}\right)-
\right.
\\
&
\nonumber
\left.
\phantom{\overline{L}\otimes
I_{M}}
-\alpha(i)I_{NM}\right]
\left[\widehat{\mathbf{x}}(i)-\widetilde{\mathbf{x}}^{\circ}(i)\right]
-\beta(i)\left(\mathbf{\Upsilon}(i)+\mathbf{\Psi}(i)\right)
\end{align}
For sufficiently large $j$, we have
\begin{equation}
\label{hatcons3} \left\|I-\beta(i)\left(\overline{L}\otimes
I_{M}\right)-\alpha(i)I_{NM}\right\|\leq 1-\alpha(i), \forall
i\geq j
\end{equation}
We then have from \eqref{hatcons2}, for $i\geq j$,
\begin{align}
\label{hatcons5}
&
\mathbb{E}_{\mathbf{\theta}^{\ast}}
\left[\left\|\widehat{\mathbf{x}}(i+1)
-\widetilde{\mathbf{x}}^{\circ}(i+1)\right\|^{2}\right]
\leq
\\
&
\nonumber
\leq
\left(1-\alpha(i)\right)^{2}\mathbb{E}_{\mathbf{\theta}^{\ast}}
\left[\left\|\widehat{\mathbf{x}}(i)
-\widetilde{\mathbf{x}}^{\circ}(i)\right\|^{2}\right]
+
\\
&
\nonumber
\hspace{2.5cm}
+\beta^{2}(i)\mathbb{E}_{\mathbf{\theta}^{\ast}}
\left[\left\|\mathbf{\Upsilon}(i)
+\mathbf{\Psi}(i)\right\|^{2}\right]\nonumber
\\
&
\leq
\left(1-\alpha(i)\right)\mathbb{E}_{\mathbf{\theta}^{\ast}}
\left[\left\|\widehat{\mathbf{x}}(i)
-\widetilde{\mathbf{x}}^{\circ}(i)\right\|^{2}\right]+\eta_{q}\beta^{2}(i)
\end{align}
where the last step follows from the fact that
$0\leq(1-\alpha(i))\leq 1$ for $i\geq j$ and \eqref{dith40}.
Continuing the recursion, we have
\begin{align}
\nonumber
&
\hspace{-.325cm}
\mathbb{E}_{\mathbf{\theta}^{\ast}}
\hspace{-.075cm}
\left[\left\|\widehat{\mathbf{x}}(i)
-\widetilde{\mathbf{x}}^{\circ}(i)\right\|^{2}\right]
\hspace{-.075cm}
\leq
\hspace{-.075cm}
\left(\prod_{k=j}^{i-1}(1-\alpha(k))\right)
\hspace{-.125cm}
\left\|\widehat{\mathbf{x}}(j)-
\widetilde{\mathbf{x}}^{\circ}(j)\right\|^{2}
\hspace{-.075cm}
+
\\
&
\label{hatcons6}
\hspace{1.5cm}
+
\eta_{q}\sum_{k=j}^{i-1}\left[\left(\prod_{l=k+1}^{i-1}\left(1-\alpha(l)\right)\right)\beta^{2}(k)\right]
\end{align}
The first and second terms on the R.H.S.~of \eqref{hatcons6} vanish as $i\rightarrow\infty$, respectively because~1) of an argument similar to the proof of Lemma~\ref{nlulm}, and~2) by  Lemma~\ref{prop_alpha}, with $\delta_{1}=\tau_{1}, \delta_{2}=2\tau_{2}$,
since by Assumption~\textbf{(D.5)}, $2\tau_{2}>\tau_{1}$. Thus:
\begin{equation}
\label{hatcons7}
\lim_{i\rightarrow\infty}\mathbb{E}_{\mathbf{\theta}^{\ast}}
\left[\left\|\widehat{\mathbf{x}}(i)
-\widetilde{\mathbf{x}}^{\circ}(i)\right\|^{2}\right]=0
\end{equation}
which shows that the sequence
$\left\{\left\|\widehat{\mathbf{x}}(i)-\widetilde{\mathbf{x}}^{\circ}(i)\right\|\right\}_{i\geq
0}$ converges to 0 in $\mathcal{L}_{2}$ (mean-squared sense).
We then have from Lemma~\ref{circcons}
\begin{align}
\label{hatcons202}
&
\lim_{i\rightarrow\infty}\mathbb{E}_{\mathbf{\theta}^{\ast}}
\left[\left\|\widehat{\mathbf{x}}(i)-\mathbf{1}_{N}\otimes
h(\mathbf{\theta}^{\ast})\right\|^{2}\right]
\leq
\\
&
\nonumber
\hspace{1cm}
\leq
2\lim_{i\rightarrow\infty}\mathbb{E}_{\mathbf{\theta}^{\ast}}
\left[\left\|\widehat{\mathbf{x}}(i)
-\widetilde{\mathbf{x}}^{\circ}(i)\right\|^{2}\right]+
\\
&
\nonumber
\hspace{2cm}
+2\lim_{i\rightarrow\infty}\mathbb{E}_{\mathbf{\theta}^{\ast}}
\left[\left\|\widetilde{\mathbf{x}}^{\circ}(i)-\mathbf{1}_{N}\otimes
h(\mathbf{\theta}^{\ast})\right\|^{2}\right]
\nonumber
= 0
\end{align}
thus establishing the claim in \eqref{hatcons200}.

We now show that the sequence
$\left\{\left\|\widehat{\mathbf{x}}(i)-\widetilde{\mathbf{x}}^{\circ}(i)\right\|\right\}_{i\geq
0}$ also converges a.s. to a finite random variable. Choose $j$
sufficiently large as in \eqref{hatcons3}. We then have from
\eqref{hatcons2}
\begin{align}
\label{hatcons8}
&
\widehat{\mathbf{x}}(i)-\widetilde{\mathbf{x}}^{\circ}(i)
=
\\
&
\nonumber
\left(\prod_{k=j}^{i-1}\left(I_{NM}-\beta(k)\left(\overline{L}\otimes
I_{M}\right)-\alpha(k)I\right)
\hspace{-.1cm}
\right)
\hspace{-.1cm}
\left(\widehat{\mathbf{x}}(j)-\widetilde{\mathbf{x}}^{\circ}(j)\right)
-
\\
&
\nonumber
\hspace{.25cm}
-\sum_{k=j}^{i-1}\left[\left(\prod_{l=k+1}^{i-1}
\left(I_{NM}-\beta(l)\left(\overline{L}\otimes
I_{M}\right)-\alpha(l)I\right)\right)
\right.
\\
&
\nonumber
\hspace{5.15cm}
\left.
\phantom{\prod_{l=k+1}^{i-1}}
\beta(k)\mathbf{\Upsilon}(k)\right]
-
\nonumber
\\
&
\nonumber
\hspace{.25cm}
-\sum_{k=j}^{i-1}\left[\left(\prod_{l=k+1}^{i-1}\left(I_{NM}-\beta(l)\left(\overline{L}\otimes
I_{M}\right)-\alpha(l)I\right)\right)
\right.
\\
&
\nonumber
\hspace{5.15cm}
\left.
\phantom{\prod_{l=k+1}^{i-1}}
\beta(k)\mathbf{\Psi}(k)\right]
\end{align}
The first term on the R.H.S. of \eqref{hatcons8} converges
a.s. to zero as $i\rightarrow\infty$ by a similar argument as in
the proof of Lemma~\ref{nlulm}. Since the sequence
$\left\{\mathbf{\Upsilon}(i)\right\}_{i\geq 0}$ is i.i.d., the second term is
a weighted summation of independent random vectors. Define the
triangular array of weight matrices, $\left\{A_{i,k},~j\leq k\leq
i-1\right\}_{i>j}$, by
\begin{equation}
\label{hatcons9}
A_{i,k}=\prod_{l=k+1}^{i-1}\left(I_{NM}-\beta(l)\left(\overline{L}\otimes
I_{M}\right)-\alpha(l)I\right)\beta(k)
\end{equation}
We then have
\begin{align}
\label{hatcons10}
&
\sum_{k=j}^{i-1}\left[\left(\prod_{l=k+1}^{i-1}\left(I_{NM}-\beta(l)\left(\overline{L}\otimes
I_{M}\right)-\alpha(l)I\right)\right)
\right.
\\
&
\nonumber
\hspace{3.30cm}
\left.
\phantom{\prod_{l=k+1}^{i-1}}
\beta(k)\mathbf{\Upsilon}(k)\right]
= \sum_{k=j}^{i-1}A_{i,k}\mathbf{\Upsilon}(k)
\end{align}
By Lemma~\ref{prop_alpha} and Assumption~\textbf{(D.5)} we note
that
\begin{align}
\label{hatcons11}
&
\limsup_{i\rightarrow\infty}\sum_{k=j}^{i-1}\left\|A_{i,k}\right\|^{2}
\leq
\\
&
\nonumber
\hspace{2cm}
\leq
\limsup_{i\rightarrow\infty}\sum_{k=j}^{i-1}\left[\left(\prod_{l=k+1}^{i-1}\left(1-\alpha(l)\right)\right)\beta^{2}(k)\right]\nonumber
\\
&
\nonumber
\hspace{2cm}
=
 0
\end{align}
It then follows that
\begin{equation}
\label{hatcons12}
\sup_{i>j}\sum_{k=j}^{i-1}\left\|A_{i,k}\right\|^{2}=C_{3}<\infty
\end{equation}
The sequence
$\left\{\sum_{k=j}^{i-1}A_{i,k}\mathbf{\Upsilon}(k)\right\}_{i>j}$
then converges a.s. to a finite random vector by standard results
from the limit theory of weighted summations of independent random
vectors (see~\cite{Chow,ChowLai,Stout}).

In a similar way, the last term on the R.H.S of
\eqref{hatcons8} converges a.s. to a finite random vector
since by the properties of dither the sequence
$\left\{\mathbf{\Psi}(i)\right\}_{i\geq 0}$ is i.i.d. It then follows from
\eqref{hatcons8} that the sequence
$\left\{\widehat{\mathbf{x}}(i)-\widetilde{\mathbf{x}}^{\circ}(i)\right\}_{i\geq
0}$ converges a.s. to a finite random vector, which in turn
implies that the sequence
$\left\{\left\|\widehat{\mathbf{x}}(i)-\widetilde{\mathbf{x}}^{\circ}(i)\right\|\right\}_{i\geq
0}$ converges a.s. to a finite random variable. However, we have
already shown that the sequence
$\left\{\left\|\widehat{\mathbf{x}}(i)-\widetilde{\mathbf{x}}^{\circ}(i)\right\|\right\}_{i\geq
0}$ converges in mean-squared sense to 0. It then follows from the
uniqueness of the mean-squared and a.s. limit, that the sequence
$\left\{\left\|\widehat{\mathbf{x}}(i)-\widetilde{\mathbf{x}}^{\circ}(i)\right\|\right\}_{i\geq
0}$ converges a.s. to 0. In other words,
\begin{equation}
\label{hatcons13}
\mathbb{P}_{\mathbf{\theta}^{\ast}}\left[\lim_{i\rightarrow\infty}\left\|\widehat{\mathbf{x}}(i)-\widetilde{\mathbf{x}}^{\circ}(i)\right\|=0\right]=1
\end{equation}
The claim in \eqref{hatcons1} then follows from
\eqref{hatcons13} and Lemma~\ref{circcons}.

\end{proof}
}

\section{Proofs of Theorems~\ref{theorem_tilde},\ref{theorem_untrans}}
\label{proof_theorem_tilde}

\textbf{Proof of Theorem~\ref{theorem_tilde}}

{\small

\begin{proof}
Recall the evolution of the sequences
$\left\{\widetilde{\mathbf{x}}(i)\right\}_{i\geq 0}$,
$\left\{\widehat{\mathbf{x}}(i)\right\}_{i\geq 0}$ in
\eqref{nlu3} and \eqref{hatcons}.

Then writing $L(i)=\overline{L}+\widetilde{L}(i)$ and using the
fact that
\begin{equation}
\label{tildecons8} \left(\widetilde{L}(i)\otimes
I_{M}\right)\widehat{\mathbf{x}}(i) =
\left(\widetilde{L}(i)\otimes
I_{M}\right)\widehat{\mathbf{x}}_{\mathcal{C}^{\perp}}(i),~\forall
i
\end{equation}
 we have from \eqref{nlu3} and \eqref{hatcons}
\begin{align}
\label{tildecons6}
&
\hspace{-.175cm}
\widetilde{\mathbf{x}}(i+1)-\widehat{\mathbf{x}}(i+1)=
\left[I_{NM}-\beta(i)\left(L(i)\otimes
I_{M}\right)-
\right.
\\
&
\nonumber
\hspace{.5cm}
\left.
-
\alpha(i)I_{NM}\right]
\left(\widetilde{\mathbf{x}}(i)-\widehat{\mathbf{x}}(i)\right)-\beta(i)\left(\widetilde{L}(i)\otimes
I_{M}\right)\widehat{\mathbf{x}}_{\mathcal{C}^{\perp}}(i)
\end{align}
For ease of notation, introduce the sequence
$\left\{\mathbf{y}(i)\right\}_{i\geq 0}$, given by
\begin{equation}
\label{tildecons9}
\mathbf{y}(i)=\widetilde{\mathbf{x}}(i)-\widehat{\mathbf{x}}(i)
\end{equation}
To prove \eqref{theorem_tilde1}, it clearly suffices (from
Lemma~\ref{hatconslemma}) to prove
\begin{equation}
\label{tildecons10}
\mathbb{P}_{\mathbf{\theta}^{\ast}}\left[\lim_{i\rightarrow\infty}\mathbf{y}(i)=\mathbf{0}\right]=1
\end{equation}
From \eqref{tildecons6}, the evolution of the sequence $\left\{\mathbf{y}(i)\right\}_{i\geq 0}$ is:
\begin{align}
\nonumber
&
\hspace{-.1cm}
\mathbf{y}(i+1)=\left[I_{NM}-\beta(i)\left(\overline{L}\otimes
I_{M}\right)-\alpha(i)I_{NM}\right]\mathbf{y}(i)
-
\\
&
\label{tildecons11}
\hspace{-.225cm}
-
\beta(i)\left(\widetilde{L}(i)\otimes
I_{M}\right)\mathbf{y}(i)-\beta(i)\left(\widetilde{L}(i)\otimes
I_{M}\right)\widehat{\mathbf{x}}_{\mathcal{C}^{\perp}}(i)
\end{align}
The sequence $\left\{\mathbf{y}(i)\right\}_{i\geq 0}$ is not uniformly bounded, in general,  because of
$\beta(i)\left(\widetilde{L}(i)\otimes
I_{M}\right)\widehat{\mathbf{x}}_{\mathcal{C}^{\perp}}(i)$. However, from Lemma~\ref{hatconslemma}:
\begin{equation}
\label{tildecons7}
\mathbb{P}_{\mathbf{\theta}^{\ast}}\left[\lim_{i\rightarrow\infty}\widehat{\mathbf{x}}_{\mathcal{C}^{\perp}}(i)=\mathbf{0}\right]=1
\end{equation}
and, hence, asymptotically, its effect diminishes. However, $\left\{\widehat{\mathbf{x}}_{\mathcal{C}^{\perp}}(i)\right\}_{i\geq
0}$ is not uniformly bounded over sample paths and, hence, we use
truncation arguments (see, e.g., \cite{Nevelson}). For a
scalar $a$, define its truncation $(a)^{R}$ at level $R>0$ by
\begin{equation}
\label{tildecons12} (a)^{R}=\left\{\begin{array}{ll}
                                \frac{a}{|a|}\min(|a|,R) & \mbox{if
                                $a\neq 0$}\\
                                0 & \mbox{if $a=0$}
                                \end{array}
                                \right.
\end{equation}
For a vector, the truncation operation applies componentwise. For
$R>0$, we also consider the sequences,
$\left\{\mathbf{y}_{R}(i)\right\}_{i\geq 0}$:
\begin{align}
\nonumber&
\hspace{-.1cm}
\mathbf{y}_{R}(i+1)=\left[I_{NM}-\beta(i)\left(\overline{L}\otimes
I_{M}\right)-\alpha(i)I_{NM}\right]\mathbf{y}_{R}(i)
-
\\
&
\nonumber
\hspace{2.0cm}
-
\beta(i)\left(\widetilde{L}(i)\otimes
I_{M}\right)\mathbf{y}_{R}(i)
-
\\
&
\label{tildecons13}
\hspace{2.0cm}
-
\beta(i)\left(\widetilde{L}(i)\otimes
I_{M}\right)\left(\widehat{\mathbf{x}}_{\mathcal{C}^{\perp}}(i)\right)^{R}
\end{align}
We will show that for every $R>0$
\begin{equation}
\label{tildecons14}
\mathbb{P}_{\mathbf{\theta}^{\ast}}\left[\lim_{i\rightarrow\infty}\mathbf{y}_{R}(i)=\mathbf{0}\right]=1
\end{equation}
Now, the sequence
$\left\{\widehat{\mathbf{x}}_{\mathcal{C}^{\perp}}(i)\right\}_{i\geq 0}$
converges a.s. to zero, and, hence, for every $\epsilon>0$, there
exists $R(\epsilon)>0$ (see~\cite{Kallenberg}), such that
\begin{equation}
\label{tildecons15}
\mathbb{P}_{\mathbf{\theta}^{\ast}}\left[\sup_{i\geq
0}\left\|\widehat{\mathbf{x}}_{\mathcal{C}^{\perp}}(i)-\left(\widehat{\mathbf{x}}_{\mathcal{C}^{\perp}}(i)\right)^{R(\epsilon)}\right\|=0\right]>1-\epsilon
\end{equation}
and, hence, from \eqref{tildecons11}-\eqref{tildecons13}
\begin{equation}
\label{tildecons16}
\mathbb{P}_{\mathbf{\theta}^{\ast}}\left[\sup_{i\geq
0}\left\|\mathbf{y}(i)-\mathbf{y}^{R(\epsilon)}(i)\right\|=0\right]>1-\epsilon
\end{equation}
This, together with \eqref{tildecons14}, will then imply
\begin{equation}
\label{tildecons17}
\mathbb{P}_{\mathbf{\theta}^{\ast}}\left[\lim_{i\rightarrow\infty}\mathbf{y}(i)=\mathbf{0}\right]>1-\epsilon
\end{equation}
Since $\epsilon>0$ is arbitrary in \eqref{tildecons17}, we
will be able to conclude \eqref{theorem_tilde1}. Thus, the proof
reduces to establishing \eqref{tildecons14} for every $R>0$,
which is carried out in the following.

For a given $R>0$ consider the recursion given in
\eqref{tildecons13}. Choose $\varepsilon_{1}>0$ and
$\varepsilon_{2}<0$ such that
\begin{equation}
\label{tildecons17-b} 1-\varepsilon_{2}<2\tau_{2}-\varepsilon_{1}.
\end{equation}
Because $\tau_{2}>.5$ in
Assumption~\textbf{(D.5)} permits such choice of
$\varepsilon_{1},\varepsilon_{2}$. Let $\rho>0$ be constant and define
$V:\mathbb{N}\times\mathbb{R}^{NM}\longmapsto\mathbb{R}^{+}$
\begin{equation}
\label{tildecons18}
V\left(i,\mathbf{x}\right)=i^{\varepsilon_{1}}\mathbf{x}^{T}\left(\overline{L}\otimes
I_{M}\right)\mathbf{x}+\rho i^{\varepsilon_{2}}.
\end{equation}
 Recall the filtration
$\left\{\mathcal{F}_{i}\right\}_{i\geq 0}$  in \eqref{nalg8}
\[
\mathcal{F}_{i}=\sigma\left(\mathbf{x}(0),\left\{L(j),\left\{\mathbf{z}_{n}(j)\right\}_{1\leq
N},~\mathbf{\Upsilon}(j),\mathbf{\Psi}(j)\right\}_{0\leq
j<i}\right)
\]
to which all the processes of interest are adapted. We now show
that there exists an integer $i_{R}>0$ sufficiently large, such
that the process $\left\{V(i,\mathbf{y}_{R}(i))\right\}_{i\geq i_{R}}$ is a
non-negative supermartingale w.r.t. the filtration
$\left\{\mathcal{F}_{i}\right\}_{i\geq i_{R}}$. To this end, we note that, using the recursion~\eqref{tildecons13}:
\begin{align}
\label{tildecons20}
&\mathbb{E}_{\mathbf{\theta}^{\ast}}
\left[V(i+1,\mathbf{y}_{R}(i+1))\,|\,\mathcal{F}_{i}\right]-V(i,\mathbf{y}_{R}(i))
=
\\
&
\nonumber
\hspace{.5cm}
(i+1)^{\varepsilon_{1}}\mathbf{y}_{R}^{T}(i+1)\left(\overline{L}\otimes
I_{M}\right)\mathbf{y}_{R}(i+1)+\rho
(i+1)^{\varepsilon_{2}}
\\
\nonumber
&
\hspace{3.15cm}
 -i^{\varepsilon_{1}}\mathbf{y}_{R}^{T}(i)\left(\overline{L}\otimes I_{M}\right)\mathbf{y}_{R}(i)-\rho i^{\varepsilon_{2}}
\\
&
\nonumber
\hspace{0.0cm}
=
(i+1)^{\varepsilon_{1}}\left[\mathbf{y}_{R,\mathcal{C}^{\perp}}^{T}(i)\left(\overline{L}\otimes
I_{M}\right)\mathbf{y}_{R,\mathcal{C}^{\perp}}(i)
-
\right.
\\
&
\nonumber
\hspace{2.4cm}
\left.
-2\beta(i)\mathbf{y}_{R,\mathcal{C}^{\perp}}^{T}(i)\left(\overline{L}\otimes I_{M}\right)^{2}\mathbf{y}_{R,\mathcal{C}^{\perp}}(i)
-
\right.
\\
&
\nonumber
\hspace{0.0cm}
 \left.
-2\alpha(i)\mathbf{y}_{R,\mathcal{C}^{\perp}}^{T}(i)
\left(\overline{L}\otimes
I_{M}\right)\mathbf{y}_{R,\mathcal{C}^{\perp}}(i)
+
\right.
\\
&
\nonumber
\hspace{1.8cm}
\left.
+
2\beta(i)\alpha(i)\mathbf{y}_{R,\mathcal{C}^{\perp}}^{T}(i)\left(\overline{L}\otimes I_{M}\right)^{2}\mathbf{y}_{R,\mathcal{C}^{\perp}}(i)
+
\right.
\\
\nonumber
&
\hspace{0.0cm}
\left.
+\beta^{2}(i)\mathbf{y}_{R,\mathcal{C}^{\perp}}^{T}(i)\left(\overline{L}\otimes
I_{M}\right)^{3}\mathbf{y}_{R,\mathcal{C}^{\perp}}(i)+
\right.
\\
&
\nonumber
\hspace{2.55cm}
\left.
+\alpha^{2}(i)\mathbf{y}_{R,\mathcal{C}^{\perp}}^{T}(i)\left(\overline{L}\otimes I_{M}\right)\mathbf{y}_{R,\mathcal{C}^{\perp}}(i)
+
\right.
\\
 \nonumber
&
\hspace{0.0cm}
\left.
+
\beta^{2}(i)\mathbb{E}_{\mathbf{\theta}^{\ast}}
\left[\mathbf{y}_{R,\mathcal{C}^{\perp}}^{T}(i)
\left(\widetilde{L}(i)\otimes
I_{M}\right)\left(\overline{L}\otimes
I_{M}\right)
\right.
\right.
\\
&
\nonumber
\hspace{3.45cm}
\left.
\left.
\left(\widetilde{L}(i)\otimes I_{M}\right)\mathbf{y}_{R,\mathcal{C}^{\perp}}(i)
\left|\rule{0cm}{.35cm}\right.\mathcal{F}_{i}\right]
+
\right.
\\
\nonumber
&
\hspace{0.0cm}
 \left.
+2\beta^{2}(i)\mathbb{E}_{\mathbf{\theta}^{\ast}}
\left[\mathbf{y}_{R,\mathcal{C}^{\perp}}^{T}(i)\left(\widetilde{L}(i)\otimes
I_{M}\right)\left(\overline{L}\otimes
I_{M}\right)
\right.
\right.
\\
&
\nonumber
\hspace{3.3cm}
\left.
\left.
\left(\widetilde{L}(i)\otimes
I_{M}\right)\left(\widehat{\mathbf{x}}_{\mathcal{C}^{\perp}}(i)\right)^{R}\left|\rule{0cm}{.35cm}\right.
\mathcal{F}_{i}\right]
+
\right.
\\
\nonumber
&
\hspace{0.0cm}
 \left.
 +\beta^{2}(i)\mathbb{E}_{\mathbf{\theta}^{\ast}}
\left[\left(\widehat{\mathbf{x}}_{\mathcal{C}^{\perp}}^{T}(i)\right)^{R}
\left(\widetilde{L}(i)\otimes I_{M}\right)\left(\overline{L}\otimes I_{M}\right)
\right.
\right.
\\
&
\nonumber
\hspace{3.3cm}
\left.
\left.
\left(\widetilde{L}(i)\otimes I_{M}\right)\left(\widehat{\mathbf{x}}_{\mathcal{C}^{\perp}}(i)\right)^{R}
|\mathcal{F}_{i}\right]\right]
+
\\
\nonumber
&
\hspace{0.0cm}
 +(i+1)^{\varepsilon_{2}}
 -i^{\varepsilon_{1}}\mathbf{y}_{R,\mathcal{C}^{\perp}}^{T}(i)
\left(\overline{L}\otimes I_{M}\right)\mathbf{y}_{R,\mathcal{C}^{\perp}}(i)
-\rho i^{\varepsilon_{2}}
\nonumber
\end{align}
where we repeatedly used the fact that
\begin{align*}
\left(\overline{L}\otimes
I_{M}\right)\mathbf{y}_{R}(i)
&=
\left(\overline{L}\otimes
I_{M}\right)\mathbf{y}_{R,\mathcal{C}^{\perp}}(i)
\\
\left(\widetilde{L}(i)\otimes I_{M}\right)\mathbf{y}_{R}(i)
&=
\left(\widetilde{L}(i)\otimes
I_{M}\right)\mathbf{y}_{R,\mathcal{C}^{\perp}}(i)
\end{align*}
and $\widetilde{L}(i)$ is independent of $\mathcal{F}_{i}$.

In going to the next step we use the following inequalities, where $c_{1}>0$ is a constant:
\begin{align}
\label{tildecons22}
&
\hspace{-.2cm}
\mathbf{y}_{R,\mathcal{C}^{\perp}}^{T}(i)\left(\overline{L}\otimes
I_{M}\right)^{2}\mathbf{y}_{R,\mathcal{C}^{\perp}}(i)
\geq
\lambda_{2}^{2}(\overline{L})
\left\|\mathbf{y}_{R,\mathcal{C}^{\perp}}(i)\right\|^{2}
\\
&
\nonumber
=
\frac{\lambda_{2}^{2}(\overline{L})}{\lambda_{N}(\overline{L})}
\lambda_{N}(\overline{L})\left\|\mathbf{y}_{R,\mathcal{C}^{\perp}}(i)
\right\|^{2}
\nonumber
\\
&
\nonumber
\geq \frac{\lambda_{2}^{2}(\overline{L})}{\lambda_{N}(\overline{L})}
\mathbf{y}_{R,\mathcal{C}^{\perp}}^{T}(i)
\left(\overline{L}\otimes I_{M}\right)
\mathbf{y}_{R,\mathcal{C}^{\perp}}(i)
\\
\label{tildecons23}
&
\mathbf{y}_{R,\mathcal{C}^{\perp}}^{T}(i)
\left(\overline{L}\otimes
I_{M}\right)^{2}
\mathbf{y}_{R,\mathcal{C}^{\perp}}(i)
\leq
\lambda_{N}^{2}(\overline{L})
\left\|\mathbf{y}_{R,\mathcal{C}^{\perp}}(i)\right\|^{2}
\nonumber
\\
&
=
\frac{\lambda_{N}^{2}(\overline{L})}{\lambda_{2}(\overline{L})}
\lambda_{2}(\overline{L})
\left\|\mathbf{y}_{R,\mathcal{C}^{\perp}}(i)\right\|^{2}
\nonumber
\\
&
\leq \frac{\lambda_{N}^{2}(\overline{L})}{\lambda_{2}(\overline{L})}
\mathbf{y}_{R,\mathcal{C}^{\perp}}^{T}(i)\left(\overline{L}
\otimes I_{M}\right)\mathbf{y}_{R,\mathcal{C}^{\perp}}(i)
\\
\label{tildecons24}
&
\mathbf{y}_{R,\mathcal{C}^{\perp}}^{T}(i)\left(\overline{L}\otimes
I_{M}\right)^{3}\mathbf{y}_{R,\mathcal{C}^{\perp}}(i)
\leq
\lambda_{N}^{3}(\overline{L})
\left\|\mathbf{y}_{R,\mathcal{C}^{\perp}}(i)\right\|^{2}
\nonumber
\\
&
=
\frac{\lambda_{N}^{3}(\overline{L})}{\lambda_{2}(\overline{L})}
\lambda_{2}(\overline{L})
\left\|\mathbf{y}_{R,\mathcal{C}^{\perp}}(i)\right\|^{2}\nonumber
\\
&
\leq \frac{\lambda_{N}^{3}(\overline{L})}{\lambda_{2}(\overline{L})}\mathbf{y}_{R,\mathcal{C}^{\perp}}^{T}(i)\left(\overline{L}\otimes I_{M}\right)\mathbf{y}_{R,\mathcal{C}^{\perp}}(i)
\end{align}
\begin{align}
\nonumber
&\mathbb{E}_{\mathbf{\theta}^{\ast}}
\left[\mathbf{y}_{R,\mathcal{C}^{\perp}}^{T}(i)\left(\widetilde{L}(i)\otimes
I_{M}\right)\left(\overline{L}\otimes
I_{M}\right)
\right.
\\
&
\nonumber
\hspace{2cm}
\left.
\left(\widetilde{L}(i)\otimes I_{M}\right)
\mathbf{y}_{R,\mathcal{C}^{\perp}}(i)\left|\rule{0cm}{.35cm}\right.
\mathcal{F}_{i}\right]
\leq
\\
&
\nonumber
\leq
\lambda_{N}(\overline{L})\mathbb{E}_{\mathbf{\theta}^{\ast}}
\left[\left\|\left(\widetilde{L}(i)\otimes
I_{M}\right)\mathbf{y}_{R,\mathcal{C}^{\perp}}(i)
\right\|^{2}\,|\,\mathcal{F}_{i}\right]
\\
\nonumber
&
\leq
c_{1}\lambda_{N}(\overline{L})\mathbb{E}_{\mathbf{\theta}^{\ast}}
\left[\left\|\mathbf{y}_{R,\mathcal{C}^{\perp}}(i)
\right\|^{2}\left|\rule{0cm}{.35cm}\right.\mathcal{F}_{i}\right]\nonumber
\\
&
\nonumber
=
c_{1}\lambda_{N}(\overline{L})\left\|\mathbf{y}_{R,\mathcal{C}^{\perp}}(i)\right\|^{2}\nonumber
\\
&
\label{tildecons25}
\leq \frac{c_{1}\lambda_{N}(\overline{L})}{\lambda_{2}}\mathbf{y}_{R,\mathcal{C}^{\perp}}^{T}(i)\left(\overline{L}\otimes I_{M}\right)\mathbf{y}_{R,\mathcal{C}^{\perp}}(i)
\\
&
\label{tildecons26}
\mathbb{E}_{\mathbf{\theta}^{\ast}}
\left[\mathbf{y}_{R,\mathcal{C}^{\perp}}^{T}(i)
\left(\widetilde{L}(i)\otimes I_{M}\right)
\left(\overline{L}\otimes I_{M}\right)
\right.
\\
&
\nonumber
\hspace{2cm}
\left.
\left(\widetilde{L}(i)\otimes I_{M}\right)
\left(\widehat{\mathbf{x}}_{\mathcal{C}^{\perp}}(i) \right)^{R}
\left|\rule{0cm}{.35cm}\right.\mathcal{F}_{i}\right]
\\
&
\nonumber
\leq
\mathbb{E}_{\mathbf{\theta}^{\ast}}
\left[
\rule{0cm}{.4cm}
\left\|
\mathbf{y}_{R,\mathcal{C}^{\perp}}^{T}(i)
\right\|
\,\left\|
\left(\widetilde{L}(i)\otimes I_{M}\right)
\right\|
\right.
\\
&
\label{eq:aux1}
\left.
\left\|
\left(\overline{L}\otimes I_{M}\right)
\right\|
\,
\left\|
\left(\widetilde{L}(i)\otimes I_{M}\right)
\right\|
\,
\left\|
\left(\widehat{\mathbf{x}}_{\mathcal{C}^{\perp}}(i)\right)^{R}
\right\|
\left|\rule{0cm}{.3cm}
\right.
\mathcal{F}_{i}\rule{0cm}{.4cm}\right]
\\
&
\label{eq:aux2}
\leq
Rc_{1}\lambda_{N}(\overline{L})
\left\|
\mathbf{y}_{R,\mathcal{C}^{\perp}}(i)
\right\|
\\
&
\nonumber
\leq
Rc_{1}\lambda_{N}(\overline{L})+
Rc_{1}\lambda_{N}(\overline{L})
\left\|
\mathbf{y}_{R,\mathcal{C}^{\perp}}(i)
\right\|^{2}
\\
&
\nonumber
\leq
Rc_{1}\lambda_{N}(\overline{L})
+\frac{Rc_{1}\lambda_{N}(\overline{L})}{\lambda_{2}(\overline{L})}
\mathbf{y}_{R,\mathcal{C}^{\perp}}^{T}(i)
\left(\overline{L}\otimes I_{M}\right) \mathbf{y}_{R,\mathcal{C}^{\perp}}(i)
\\
&
\label{tildecons27}
\mathbb{E}_{\mathbf{\theta}^{\ast}}
\left[
\left(\widehat{\mathbf{x}}_{\mathcal{C}^{\perp}}^{T}(i)\right)^{R}
\left(\widetilde{L}(i)\otimes I_{M}\right) \left(\overline{L}\otimes I_{M}\right)
\right.
\\
&
\hspace{2cm}
\nonumber
\left.
\left(\widetilde{L}(i)\otimes I_{M}\right)
\left(\widehat{\mathbf{x}}_{\mathcal{C}^{\perp}}(i)\right)^{R}
\,\left|\rule{0cm}{.35cm}\right.
\,\mathcal{F}_{i}\rule{0cm}{.45cm}\right]
\\
&
\leq
R^{2}c_{1}\lambda_{N}(\overline{L})
\\
&
\label{tildecons28}
(i+1)^{\varepsilon_{1}}-i^{\varepsilon_{1}}
\leq
\varepsilon_{1}(i+1)^{\varepsilon_{1}-1}
\\
&
\label{tildecons29}
\rho(i+1)^{\varepsilon_{2}}-\rho
i^{\varepsilon_{2}}
\leq
\rho\varepsilon_{2}i^{\varepsilon_{2}-1}.
\end{align}
 We go from \eqref{eq:aux1} to \eqref{eq:aux2} because $\left\|\left(\widehat{\mathbf{x}}_{\mathcal{C}^{\perp}}(i)
\right)^{R}\right\|\leq R$.
Using inequalities~\eqref{tildecons22}-\eqref{tildecons29}, we have
from \eqref{tildecons20}
\begin{align}
\label{tildecons30}
&\mathbb{E}_{\mathbf{\theta}^{\ast}}
\left[V(i+1,\mathbf{y}_{R}(i+1))\left|\rule{0cm}{.35cm}\right.\mathcal{F}_{i}
\right]-V(i,\mathbf{y}_{R}(i))
\leq
\\
&
\nonumber
(i+1)^{\varepsilon_{1}}\left[\frac{\varepsilon_{1}}{(i+1)^{1}}-2\beta(i)
\frac{\lambda_{2}^{2}(\overline{L})}{\lambda_{N}(\overline{L})}
-2\alpha(i)
+
\right.
\\
&
\nonumber
\left.
+2\beta(i)\alpha(i)\frac{\lambda_{N}^{2}(\overline{L})}{
\lambda_{2}(\overline{L})}
\hspace{-.055cm}
+
\hspace{-.055cm}
\beta^{2}
\hspace{-.05cm}
(i)\frac{\lambda_{N}^{3}(\overline{L})}{
\lambda_{2}(\overline{L})}
\hspace{-.055cm}
+
\hspace{-.055cm}
\alpha^{2}
\hspace{-.05cm}
(i)
\hspace{-.055cm}
+
\hspace{-.055cm}
\beta^{2}
\hspace{-.05cm}
(i)\frac{c_{1}\lambda_{N}(\overline{L})}{
\lambda_{2}}
\hspace{-.03cm}
+
\right.
\\
&
\nonumber
\left.
+
2\beta^{2}(i)\frac{Rc_{1}\lambda_{N}(\overline{L})}{
\lambda_{2}(\overline{L})}\right]
\mathbf{y}_{R,\mathcal{C}^{\perp}}^{T}(i)\left(\overline{L}\otimes
I_{M}\right)\mathbf{y}_{R,\mathcal{C}^{\perp}}(i)
+
\\
&
\nonumber
+\left[\frac{1}{2\tau_{2}-\varepsilon_{1}}
\left(2Rc_{1}\lambda_{N}(\overline{L})
+R^{2}c_{1}\lambda_{N}(\overline{L})\right)
+\rho\varepsilon_{2}i^{\varepsilon_{2}-1}\right]
\end{align}
For the first term on the R.H.S.~of \eqref{tildecons30}
involving
$\mathbf{y}_{R,\mathcal{C}^{\perp}}^{T}(i)\left(\overline{L}\otimes
I_{M}\right)\mathbf{y}_{R,\mathcal{C}^{\perp}}(i)$, the
coefficient $-2\beta(i)(i+1)^{\varepsilon_{1}}$ dominates all
other coefficients eventually ($\tau_{2}<1$ by
Assumption~\textbf{(D.5)}); hence, the first term on the R.H.S.~of \eqref{tildecons30} becomes negative eventually (for
sufficiently large $i$). The second term on the R.H.S.~of
\eqref{tildecons30} becomes negative eventually because
$\rho\varepsilon_{2}<0$ and
$1-\varepsilon_{2}<2\tau_{2}-\varepsilon_{1}$ by assumption. Hence
there exists sufficiently large $i$, say $i_{R}$, such that,
\[
\mathbb{E}_{\mathbf{\theta}^{\ast}}\left[V(i+1,\mathbf{y}_{R}(i+1))\left|\rule{0cm}{.35cm}\right.\mathcal{F}_{i}\right]-V(i,\mathbf{y}_{R}(i))\leq
0, \forall i\geq i_{R}.
\]
This shows  $\left\{V(i,\mathbf{y}_{R}(i))\right\}_{i\geq
i_{R}}$ is a non-negative supermartingale w.r.t.~the filtration
$\left\{\mathcal{F}_{i}\right\}_{i\geq i_{R}}$. Thus,
$\left\{V(i,\mathbf{y}_{R}(i))\right\}_{i\geq i_{R}}$ converges a.s. to a
finite random variable (see~\cite{Kallenberg}). Clearly,
the sequence $\rho i^{\varepsilon_{2}}$ goes to zero as
$\varepsilon_{2}<0$. Then:
\[
\mathbb{P}_{\mathbf{\theta}^{\ast}}\left[\lim_{i\rightarrow\infty}i^{\varepsilon_{1}}\mathbf{y}_{R}^{T}(i)\left(\overline{L}\otimes
I_{M}\right)\mathbf{y}_{R}(i)~\mbox{exists and is finite}\right]=1
\]
Since $i^{\varepsilon_{1}}\rightarrow\infty$ as
$i\rightarrow\infty$, it follows
\begin{equation}
\label{tildecons33}
\mathbb{P}_{\mathbf{\theta}^{\ast}}\left[\lim_{i\rightarrow\infty}\mathbf{y}_{R}^{T}(i)\left(\overline{L}\otimes
I_{M}\right)\mathbf{y}_{R}(i)=0\right]=1
\end{equation}
Since $\mathbf{y}_{R}^{T}(i)\left(\overline{L}\otimes
I_{M}\right)\mathbf{y}_{R}(i)\geq\lambda_{2}(\overline{L})\left\|\mathbf{y}_{R,\mathcal{C}^{\perp}}(i)\right\|^{2}$,
from \eqref{tildecons33} we have
\begin{equation}
\label{tildecons34}
\mathbb{P}_{\mathbf{\theta}^{\ast}}\left[\lim_{i\rightarrow\infty}\mathbf{y}_{R,\mathcal{C}^{\perp}}(i)=0\right]=1
\end{equation}
To establish \eqref{tildecons14} we note that
\begin{equation}
\label{tildecons35}
\mathbf{y}_{R,\mathcal{C}}(i)=\mathbf{1}_{N}\otimes\mathbf{y}_{R,\mbox{\scriptsize{avg}}}(i)
\end{equation}
where
\begin{equation}
\label{tildecons36}
\mathbf{y}_{R,\mbox{\scriptsize{avg}}}(i+1)=\left(1-\alpha(i)\right)\mathbf{y}_{R,\mbox{\scriptsize{avg}}}(i)
\end{equation}
Since $\sum_{i\geq 0}\alpha(i)=\infty$, it follows from standard
arguments that
$\mathbf{y}_{R,\mbox{\scriptsize{avg}}}(i)\rightarrow 0$ as
$i\rightarrow\infty$. We then have from \eqref{tildecons35}
\begin{equation}
\label{tildecons37}
\mathbb{P}_{\mathbf{\theta}^{\ast}}\left[\lim_{i\rightarrow\infty}\mathbf{y}_{R,\mathcal{C}}(i)=0\right]=1
\end{equation}
which together with \eqref{tildecons34} establishes
\eqref{tildecons14}. The claim in \eqref{theorem_tilde1}
then follows from the arguments above.

We now prove the claim in \eqref{theorem_tilde2}. Recall the 
 matrix $P$ in \eqref{circcons5}. Using the fact,
\begin{equation}
\label{tildecons39} P\left(L(i)\otimes
I_{M}\right)=P\left(\overline{L}\otimes
I_{M}\right)=\mathbf{0},~~\forall i
\end{equation}
we have
\begin{align}
\nonumber
&
 P\widetilde{\mathbf{x}}(i+1) =
P\widetilde{\mathbf{x}}(i)-\alpha(i)
\left[P\widetilde{\mathbf{x}}(i)-PJ(\mathbf{z}(i))\right]
\\
&
\label{tildecons40}
\hspace{3cm}
-\beta(i)P\left(\mathbf{\Upsilon}(i)+\mathbf{\Psi}(i)\right)
\end{align}
and similarly
\begin{align}
\nonumber
&
P\widehat{\mathbf{x}}(i+1) =
P\widehat{\mathbf{x}}(i)-\alpha(i)
\left[P\widehat{\mathbf{x}}(i)-PJ(\mathbf{z}(i))\right]
\\
&
\label{tildecons41}
\hspace{3.0cm}
-\beta(i)P\left(\mathbf{\Upsilon}(i)+\mathbf{\Psi}(i)\right)
\end{align}
Since the sequences $\left\{P\widetilde{\mathbf{x}}(i)\right\}_{i\geq 0}$ and
$\left\{P\widehat{\mathbf{x}}(i)\right\}_{i\geq 0}$ follow the same recursion
and start with the same initial state
$P\widetilde{\mathbf{x}}(0)$, they are equal, and we have~$\forall
i$
\begin{eqnarray}
\label{tildecons42} P\mathbf{y}(i) & = &
P\left(\widetilde{\mathbf{x}}(i)-\widehat{\mathbf{x}}(i)\right)\nonumber
\\ &=&0
\end{eqnarray}
From \eqref{tildecons11} we then have
\begin{align*}
&
\mathbf{y}(i+1)=\left[I_{NM}-\beta(i)\left(\overline{L}\otimes
I_{M}\right)-\alpha(i)I_{NM}-P\right]\mathbf{y}(i)
-
\\
&
\hspace{1.15cm}
-
\beta(i)\left(\widetilde{L}(i)\otimes
I_{M}\right)\mathbf{y}(i)
-
\beta(i)\left(\widetilde{L}(i)\otimes
I_{M}\right)\widehat{\mathbf{x}}(i)
\end{align*}
By Lemma~\ref{hatconslemma}, to prove the claim in
\eqref{theorem_tilde1}, it suffices to prove 
\begin{equation}
\label{tildecons44}
\lim_{i\rightarrow\infty}\mathbb{E}_{\mathbf{\theta}^{\ast}}
\left[\left\|\mathbf{y}(i)\right\|^{2}\right]=0
\end{equation}

From Lemma~\ref{hatconslemma}, we note that the sequence
$\{\widehat{\mathbf{x}}(i)\}_{i\geq 0}$ converges in
$\mathcal{L}_{2}$ to $\mathbf{1}_{N}\otimes
h(\mathbf{\theta}^{\ast})$ and hence $\mathcal{L}_{2}$ bounded,
i.e., there exists constant $c_{3}>0$, such that,
\begin{equation}
\label{tildecons48} \sup_{i\geq
0}\mathbb{E}_{\mathbf{\theta}^{\ast}}\left[\left\|\widehat{\mathbf{x}}(i)\right\|^{2}\right]\leq
c_{3}<\infty
\end{equation}
Choose $j$ large enough, such that, for $i\geq j$
\[
\left\|I_{NM}-\beta(i)\left(\overline{L}\otimes
I_{M}\right)-\alpha(i)I_{NM}-P\right\|\leq
1-\beta(i)\lambda_{2}(\overline{L})
\]
Noting that $\widetilde{L}(i)$ is independent of $\mathcal{F}_{i}$
and $\left\|\widetilde{L}(i)\right\|\leq c_{2}$ for some constant
$c_{2}>0$, we have for $i\geq j$,
\begin{align}
\label{tildecons50}
&\mathbb{E}_{\mathbf{\theta}^{\ast}}
\left[\left\|\mathbf{y}(i+1)\right\|^{2}\right]
=
\\
&
\nonumber
\mathbb{E}_{\mathbf{\theta}^{\ast}}\left[\mathbf{y}^{T}(i)\left(I_{NM}-\beta(i)\left(\overline{L}\otimes
I_{M}\right)-\alpha(i)I_{NM}-P\right)^{2}\mathbf{y}(i)\right.
\\
&
\nonumber
 \left. +
\beta^{2}(i)\mathbf{y}^{T}(i)\left(\widetilde{L}(i)\right)^{2}\mathbf{y}(i)+\beta^{2}(i)\widehat{\mathbf{x}}^{T}(i)\left(\widetilde{L}(i)\right)^{2}\widehat{\mathbf{x}}(i)\right.
\\
&
\nonumber
\left.
+
\beta^{2}(i)\mathbf{y}^{T}(i)\left(\widetilde{L}(i)\right)^{2}\widehat{\mathbf{x}}(i)\right]
\\
&
\nonumber
\leq
\left(1-\beta(i)\lambda_{2}(\overline{L})\right)\mathbb{E}_{\mathbf{\theta}^{\ast}}\left[\left\|\mathbf{y}(i)\right\|^{2}\right]+c^{2}_{2}\beta^{2}(i)\mathbb{E}_{\mathbf{\theta}^{\ast}}\left[\left\|\mathbf{y}(i)\right\|^{2}\right]
\\
&
\nonumber
+
c_{2}^{2}c_{3}\beta^{2}(i)+\left(2\beta^{2}(i)c_{2}^{2}c_{3}^{\frac{1}{2}}\right)\mathbb{E}^{\frac{1}{2}}_{\mathbf{\theta}^{\ast}}\left[\left\|\mathbf{y}(i)\right\|^{2}\right]
\\
&
\nonumber
\leq
\hspace{-.03cm}
\left(
\hspace{-.03cm}
1-\beta(i)\lambda_{2}(\overline{L})+c^{2}_{2}\beta^{2}(i)
+2\beta^{2}(i)c_{2}^{2}c_{3}^{\frac{1}{2}}
\hspace{-.03cm}
\right)\mathbb{E}_{\mathbf{\theta}^{\ast}}\left[\left\|\mathbf{y}(i)\right\|^{2}\right]
\\
&
\nonumber
+
\beta^{2}(i)\left(c_{2}^{2}c_{3}+2c_{2}^{2}c_{3}^{\frac{1}{2}}\right)
\end{align}
where in the last step we used the inequality
\begin{equation}
\label{tildecons51}
\mathbb{E}^{\frac{1}{2}}_{\mathbf{\theta}^{\ast}}\left[\left\|\mathbf{y}(i)\right\|^{2}\right]\leq
\mathbb{E}_{\mathbf{\theta}^{\ast}}\left[\left\|\mathbf{y}(i)\right\|^{2}\right]+1
\end{equation}
Now similar to Lemma~\ref{circcons}, choose $j_{1}\geq j$ and
$0<c_{4}<\lambda_{2}(\overline{L})$, such that,
\[
1-\beta(i)\lambda_{2}(\overline{L})+c^{2}_{2}\beta^{2}(i)+2\beta^{2}(i)c_{2}^{2}c_{3}^{\frac{1}{2}}\leq
1-\beta(i)c_{4},~~\forall i\geq j_{1}
\]
Then, for $i\geq j_{1}$, from \eqref{tildecons50}
\begin{align}
&
\nonumber
\mathbb{E}_{\mathbf{\theta}^{\ast}}
\left[
\left\|\mathbf{y}(i+1)\right\|^{2}\right]
\leq\left(1-\beta(i)c_{4}\right)
\mathbb{E}_{\mathbf{\theta}^{\ast}}
\left[
\left\|\mathbf{y}(i)\right\|^{2}
\right]
+
\\
&
\nonumber
\hspace{4cm}
+
\beta^{2}(i)\left(c_{2}^{2}c_{3}+2c_{2}^{2}c_{3}^{\frac{1}{2}}\right)
\end{align}
from which we conclude that
$\lim_{i\rightarrow\infty}\mathbb{E}_{\mathbf{\theta}^{\ast}}
\left[\left\|\mathbf{y}(i)\right\|^{2}\right]=0$ by
Lemma~\ref{prop_alpha} (see also Lemma~\ref{circcons}.)

\end{proof}
}


%
%

\textbf{Proof of Theorem~\ref{theorem_untrans}}

{\small

\begin{proof}
Consistency follows from the fact that by Theorem~\ref{theorem_tilde} 
the sequence $\left\{\widetilde{\mathbf{x}}(i)\right\}_{i\geq 0}$
converges a.s. to $\mathbf{1}_{N}\otimes h(\mathbf{\theta}^{\ast})$,
and the function $h^{-1}(\cdot)$ exists and is continuous on the open set $\mathcal{U}$.

To establish the second claim, we note that, if $h^{-1}(\cdot)$ is
Lipschitz continuous, there exists constant $k>0$, such that
\[
\left\|h^{-1}(\widetilde{\mathbf{y}}_{1})-h^{-1}(\widetilde{\mathbf{y}}_{2})\right\|\leq
k\left\|\widetilde{\mathbf{y}}_{1}-\widetilde{\mathbf{y}}_{2}\right\|,~~\forall~\widetilde{\mathbf{y}}_{1},\widetilde{\mathbf{y}}_{2}\in\mathbb{R}^{M}
\]
Since $\mathcal{L}_{2}$ convergence implies $\mathcal{L}_{1}$, we
then have from Theorem~\ref{theorem_tilde} 
 for $1\leq n\leq N$
\begin{align}
\label{mainNLU4}
&
\lim_{i\rightarrow\infty}\left\|\mathbb{E}_{\mathbf{\theta}^{\ast}}
\left[\mathbf{x}_{n}(i)-\mathbf{\theta}^{\ast}\right]\right\|
\leq
\lim_{i\rightarrow\infty}\mathbb{E}_{\mathbf{\theta}^{\ast}}
\left[\left\|\mathbf{x}_{n}(i)-\mathbf{\theta}^{\ast}\right\|\right]\nonumber
\\
&
\nonumber
\hspace{2cm}
=
\lim_{i\rightarrow\infty}\mathbb{E}_{\mathbf{\theta}^{\ast}}
\left[\left\|h^{-1}\left(\widetilde{\mathbf{x}}_{n}(i)\right)
-h^{-1}\left(h(\mathbf{\theta}^{\ast})\right)\right\|\right]\nonumber
\\
&
\nonumber
\hspace{2cm}
\leq
k\lim_{i\rightarrow\infty}\mathbb{E}_{\mathbf{\theta}^{\ast}}
\left[\left\|\widetilde{\mathbf{x}}_{n}(i)
-h(\mathbf{\theta}^{\ast})\right\|\right]\nonumber
\\
&
\nonumber
\hspace{2cm}
=
0
\end{align}
which establishes the theorem.
\end{proof}
}

\begin{IEEEbiography}{Soummya Kar} (S'05--M'10) received
the B.Tech. degree in Electronics and Electrical
Communication Engineering from the Indian Institute
of Technology, Kharagpur, India, in May 2005
and the Ph.D. degree in electrical and computer
engineering from Carnegie Mellon University,
Pittsburgh, PA, in 2010. From June 2010 to May 2011 he was with the EE Department at Princeton University as a Postdoctoral Research Associate.
 He is currently an Assistant Research Professor of ECE at Carnegie Mellon University.
 His research interests include performance
analysis and inference in large-scale networked systems, adaptive stochastic
systems, stochastic approximation, and large deviations.
\end{IEEEbiography}

\vspace*{-1.0cm}
\begin{IEEEbiography}{Jos\'{e} M. F. Moura} (S'71--M'75--SM'90--F'94) received  degrees from Instituto Superior T\'ecnico (IST),
Portugal, and from the Massachusetts Institute
of Technology (MIT), Cambridge, MA.
He  is University Professor at Carnegie Mellon University (CMU), was on the faculty at IST and has been visiting Professor at MIT.
%
 His interests include statistical and algebraic signal and image processing.

Dr.~Moura is Division~IX Director of the IEEE, was President of the IEEE Signal Processing Society (SPS), and  Editor in Chief of the  IEEE Transactions on Signal Processing.
%
 Dr.~Moura is a  Fellow from IEEE and the AAAS and an  Academy of Sciences of Portugal corresponding member. He has received several awards including the  SPS Technical Achievement Award. In 2010, he was elected University Professor at CMU.
\end{IEEEbiography}

\vspace*{-2cm}
\begin{IEEEbiography}{Kavita Ramanan} received her Ph.D. from the Division of Applied Mathematics at Brown University in 1998. After a post-doctoral position at the Technion in Haifa, Israel, Dr.~Ramanan joined the Mathematical Sciences Center of Bell Laboratories, Lucent, as a Member of Technical Staff. In 2003, Dr.~Ramanan returned to academia as an Associate Professor at the Department of Mathematical Sciences at Carnegie Mellon University (CMU), Pittsburgh. In 2010, she moved to the Division of Applied Mathematics at Brown University where she is a Professor of Applied Mathematics. Dr.~Ramanan works on probability theory, stochastic processes and their applications. Since 2008, she has also been an adjunct professor at the Chennai Mathematical Institute, India.
\end{IEEEbiography}
\end{document}